\documentclass{style/vldb}
\usepackage{amssymb}
\usepackage{amsmath}
\usepackage{color}
\usepackage{times}
\usepackage{cite}
\usepackage{enumitem}
\usepackage{arydshln}
\usepackage{url}
\usepackage{algorithm}
\usepackage{algorithmic}

\usepackage{flushend}
\usepackage{subcaption}
\usepackage{todonotes}
\usepackage{balance}
\usepackage{enumitem}
\setlist{nolistsep,leftmargin=*}

\usepackage{graphicx}
\usepackage{comment}
\newcommand\abol[1]{\textcolor{blue}{(Abol) #1}}

\newcommand\techrep[1]{#1}
\newcommand\submit[1]{}

\usepackage{xspace}

\newcommand{\stitle}[1]{\vspace{1ex}\noindent{\bf #1}}
\newcommand{\twodrrr}{{\sc 2drrr}\xspace}

\newcommand{\mdrrr}{{\sc mdrrr}\xspace}
\newcommand{\mdrc}{{\sc mdrc}\xspace}
\newcommand{\ksetr}{{\sc k-set$_r$}\xspace}
\newcommand{\hdrrms}{{\sc hd-rrms}\xspace}

\title{RRR: Rank-Regret Representative}

\numberofauthors{1}
\author
{
\alignauthor
Abolfazl Asudeh$^{\dag}$,
Azade Nazi$^{\dag\dag}$,
Nan Zhang$^{\ddag}$, 
Gautam Das$^{\dag\ddag}$,
H. V. Jagadish$^{\dag}$
\\
\affaddr {
	$^\dag$University of Michigan;
	$^{\dag\dag}$Microsoft Research;
	$^{\ddag}$Pennsylvania State University;\\
	$^{\dag\ddag}$University of Texas at Arlington
}
{
\email
	{
	$^{\dag}$\{asudeh, jag\}@umich.edu;
	$^{\dag\dag}$azade.nazi@microsoft.com;
	$^{\ddag}$nan@ist.psu.edu;
	$^{\dag\ddag}$gdas@uta.edu
	}
}
}
\date{}

\newtheorem{theorem}{Theorem} 
\newtheorem{lemma}[theorem]{Lemma}
\newtheorem{definition}{Definition}
\usepackage[font={small,bf}]{caption} 
\begin{document}
\maketitle

\begin{abstract}
Selecting the best items in a dataset is a common task in data exploration.
However, the concept of ``best'' lies in the eyes of the beholder: different users may consider different attributes more important, and hence arrive at different rankings.
Nevertheless, one can remove ``dominated'' items and create a ``representative'' subset of the data set, comprising the ``best items'' in it.
A Pareto-optimal representative is guaranteed to contain the best item of each possible ranking, but it can be almost as big as the full data. Representative can be found if we relax the requirement to include the best item for every possible user, and instead just limit the users' ``regret''.
Existing work defines regret as the loss in score by limiting consideration to the representative instead of the full data set, for any chosen ranking function.

However, the score is often not a meaningful number and users may not understand its absolute value. Sometimes small ranges in score can include large fractions of the data set. In contrast, users do understand the notion of rank ordering.
Therefore, alternatively, we consider the position of the items in the ranked list for defining the regret and propose the {\em rank-regret representative} as the minimal subset of the data containing at least one of the top-$k$ of any possible ranking function. This problem is NP-complete. We use the 
geometric interpretation of items to bound their ranks on ranges of functions and to utilize combinatorial geometry notions for developing effective and efficient approximation algorithms for the problem.
Experiments on real datasets demonstrate that we can efficiently find small subsets with small rank-regrets.
\end{abstract}

\section{Introduction}\label{sec:intro}

Given a dataset with multiple attributes,
the challenge is to combine the values of multiple attributes to arrive at a rank.
In many applications, especially in databases with numeric attributes, a weight vector $\vec{w}$ is used to express user preferences in the form of a linear combination of the attributes, i.e., $\sum w_i A_i$.
Finding flights based on a linear combination of some criteria such as price and duration~\cite{QueryReranking}, diamonds based on depth and carat~\cite{qr2}, and houses based on price and floor area~\cite{qr2} are a few examples.

The difficulty is that the concept of ``best'' lies in the eyes of the beholder.
Different users may consider different attributes more important, and hence arrive at very different rankings.
In the absence of explicit user preferences, the system can remove dominated items, and offer the remaining Pareto-optimal~\cite{pareto} set as representing the desirable items in the data set.
Such a skyline (resp. convex hull) is the smallest subset of the data that is guaranteed to contain the top choice of a user based on any monotonic (resp. linear) ranking function.
Borzsony et. al.~\cite{skylineoperator} initiated the skyline research in the database community and since then a large body of work has been conducted in this area.
A major issue with such representatives is that they can be a large portion of the dataset~\cite{asudeh2017,har2011expected}, especially when there are multiple attributes.
Hence, several researchers have tackled~\cite{chan2006, vlachou2010ranking} the challenge of finding a small subset of the data for further consideration.  

One elegant way to find a smaller subset is to define the notion of {\em regret} for any particular user. That is, how much this user loses by restricting consideration only to the subset rather than the whole set. The goal is to find a small subset of the data such that this regret is small for every user, no matter what their preference function.
There has been considerable attention given to the regret-ratio minimizing set~\cite{nanongkai2010,asudeh2017} problem and its variants~\cite{nanongkai2012interactive, zeighami2016minimizing,kessler2015k,chester2014,agarwal2017efficient,cao2017k,kumar2018faster}.
For a function $f$ and a subset of the data, let $m_o$ be the maximum score of the tuples in dataset based on $f$ and $m_a$ be the one for the subset. The regret-ratio of the subset for $f$ is the ratio of $(m_o - m_a)$ to $m_o$.
The classic regret-ratio minimizing set problem aims to find a subset of size $r$ that minimizes the maximum regret-ratio for any possible function.
Other variations of the problem are pointed out in \S~\ref{sec:related}.

Unfortunately, in most real situations, the actual score is a ``made up'' number with no direct significance.  This is even more so the case when attribute values are drawn from different domains.  In fact, the score itself could also be on a made-up scale.  Considering the regret as a ratio helps, but is far from being a complete solution.  For example, wine ratings appear to be on a 100 point scale, with the best wines in the high 90s.  However, wines rated below 80 almost never make it to a store, and the median rating is roughly 88 (exact value depends on rater).  Let's say the best wine in some data set is at 90 points.  A regret of 3 points gives a very small regret ratio of .03, but actually only promises a wine with a rating of 87, which is below median!  In other words, a small value of regret ratio can actually result in a large swing in rank.  In the case of wines at least the rating scales see enough use that most wine-drinkers would have a sense of what a score means.  But consider choosing a hotel.  If a website takes your preferences into account and scores a hotel at 17.2 for you, do you know what that means?   If not, then how can you meaningfully specify a regret ratio?

Although ordinary users may not have a good sense of actual scores, they almost always understand the notion of rank.
Therefore,  as an alternative to the regret-ratio, we consider the 
position of the items in the ranked list and propose the position distance of items to the top of the list as the {\em rank-regret} measure. We define the {\em rank-regret} of a subset of the data to be $k$, if it contains at least one of the top-$k$ tuples of any possible ranking function.

Since items in a dataset are usually not uniformly distributed by score, solutions that minimize regret-ratio do not typically minimize rank-regret.
In this paper, we seek to find the smallest subset of the given data set that has {\em rank-regret} of $k$.
We call this subset the {\em order $k$ rank-regret representative} of the database.
(We will write this as $k$-RRR, or simply as RRR when $k$ is understood from context).
The order $1$ rank-regret representative of a database (for linear ranking functions) is its convex hull: guaranteed to contain the top choice of any linear ranking function.  
The convex hull is usually very large: almost the entire data set with five or more dimensions~\cite{asudeh2017,har2011expected}.
By choosing a value of $k$ larger than $1$, we can drastically reduce the size of the rank-regret representative set, while guaranteeing everyone a choice in their top $k$ even if not the absolute top choice.

Unfortunately, finding RRR is NP-complete, even for three dimensions.
However, 
we find a bound on the maximum rank of an item for a function and use it for designing efficient approximation algorithms. We also find the connection of the RRR problem with well-known notions in combinatorial geometry such as $k$-set~\cite{Ed87}, a set of $k$ points in $d$ dimensional space separated from the remaining points by a $d-1$ dimensional hyperplane. We show how the $k$-set notion can be used to find a set that guarantees a rank-regret of $k$ and has size at most a logarithmic times the optimal solution. We then show how a smart partitioning of the function space offers an elegant way of finding the rank-regret representative. 

\stitle{Summary of contributions.} The following are the summary of our contributions in this paper:
\begin{itemize}
\item We propose the rank-regret representative as a way of choosing a small subset of the database guaranteed to contain at least one good choice for every user.
\item We provide a key theorem that, 
given the rank of an item for a pair of functions,
bounds the maximum rank of the item for any function ``between'' these functions.
\item For the special 2D case, we provide an approximation algorithm \twodrrr that guarantees to achieve the optimal output size and the approximation ratio of 2 on the rank-regret.
\item We identify the connection of the problem with the combinatorial geometry notion of $k$-set. We review the $k$-set enumeration can be modeled by graph traversal.
Using the collection of $k$-sets, for the general case with constant number of dimensions,
we model the problem by geometric hitting set, and propose the approximation algorithm \mdrrr that guarantees the rank-regret of $k$ and a logarithmic approximation-ratio on its output size. We also propose a randomized algorithm for $k$-set enumeration, based on the coupon collector's problem.
\item We propose a function space partitioning algorithm \mdrc that, for a fixed number of dimensions, guarantees a fixed approximation ratio on the rank-regret.
As confirmed in the experiments, applying a greedy strategy while partitioning the space makes this algorithm both efficient and effective in practice.
\item We conduct extensive experimental evaluations on two real datasets to verify the efficiency and effectiveness of our proposal.
\end{itemize}

In the following, we first formally define the terms, provide the problem formulation, and study its complexity in \S~\ref{sec:pre}. 
We provide the geometric interpretation of items, a dual space, and some statements in \S~\ref{sec:dual} that play key roles in the technical sections.
In \S~\ref{sec:2d}, we study the 2D problem and propose an approximation algorithm for it.
\S~\ref{sec:md} starts by revisiting the $k$-set notion and its connection to our problem. Then we provide the hitting set based approximation algorithm, as well as the function space partitioning based algorithm, for the general multi dimensional case.
Experiment results and related work are provided in \S~\ref{sec:exp} and~\ref{sec:related}, respectively, and the paper is concluded in \S~\ref{sec:conclusion}.

\section{Problem Definition}\label{sec:pre}

\begin{figure*}[!ht] 
    \begin{minipage}[t]{0.2\linewidth}
        \centering
        \vspace{-40mm}
        \begin{tabular}{|l|c|c|}
               \hline
              	id   &  $x_1$&$x_2$\\ \hline
             	$t_1$&  0.8& 0.28 \\ \hline
                $t_2$& 0.54& 0.45\\ \hline
                $t_3$& 0.67& 0.6\\ \hline
                $t_4$& 0.32& 0.42\\ \hline
                $t_5$& 0.46& 0.72\\ \hline
                $t_6$& 0.23& 0.52\\ \hline
                $t_7$& 0.91& 0.43\\ \hline
        \end{tabular}
        \vspace{7mm}\caption{A 2D dataset}
        \label{fig:toydata}
    \end{minipage}
    \begin{minipage}[t]{0.38\linewidth}
        \includegraphics[width = \textwidth]{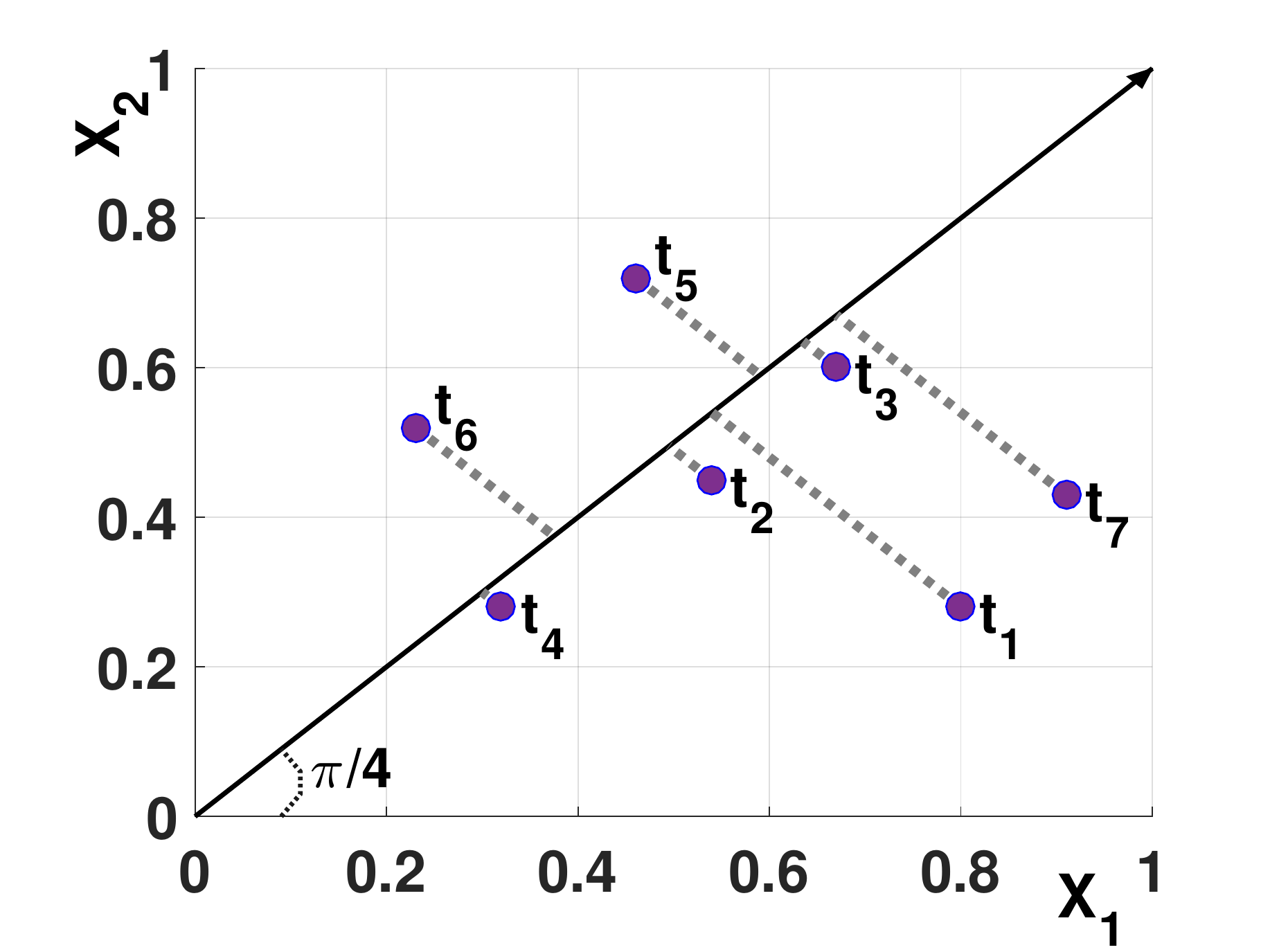}
        \vspace{-6mm}\caption{Items of Fig.~\ref{fig:toydata} ordered by $f=x_1+x_2$}
        \label{fig:toy1}
    \end{minipage}
    \begin{minipage}[t]{0.38\linewidth}
         \includegraphics[width = \textwidth]{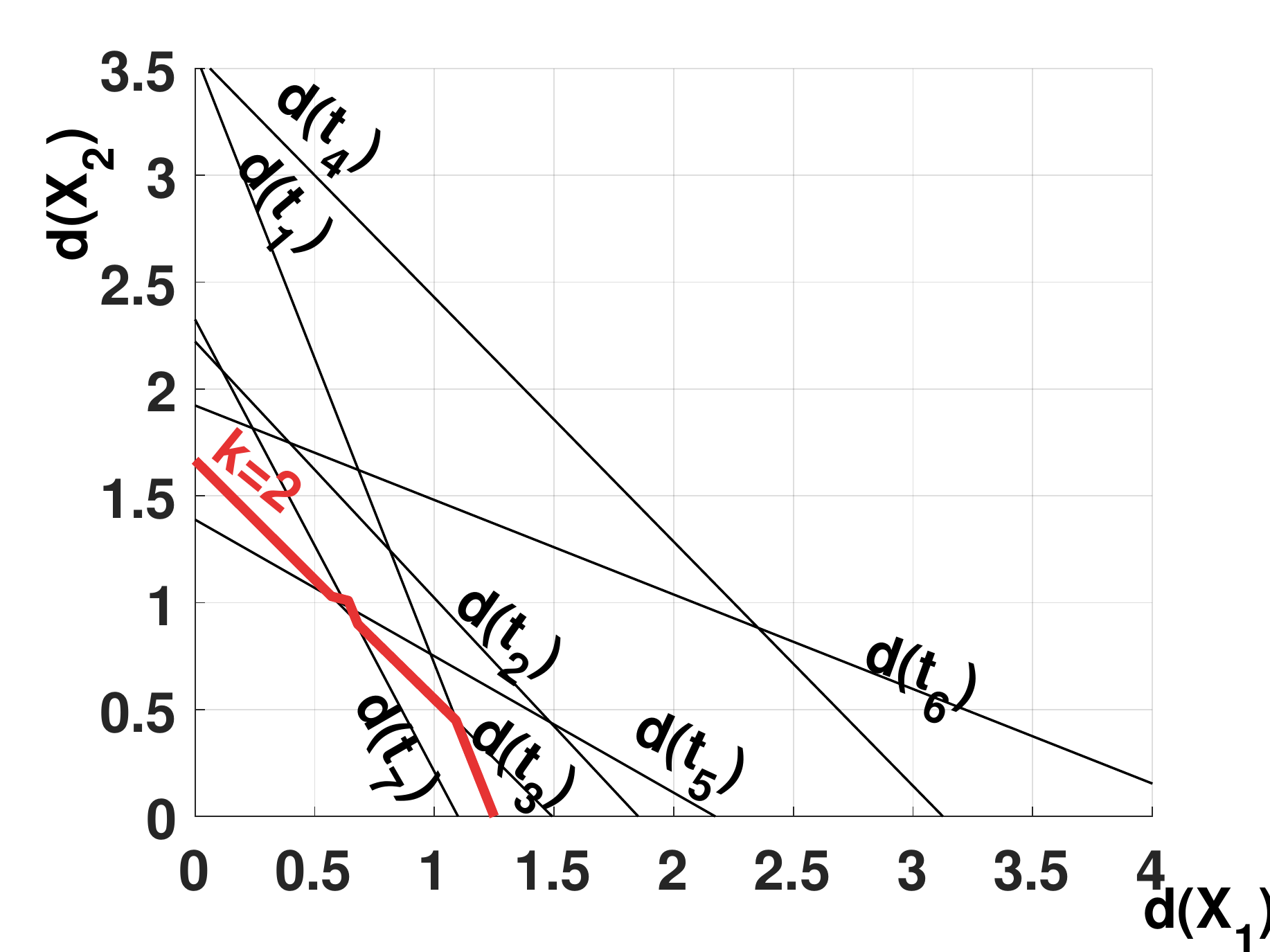}
        \vspace{-6mm}\caption{Dual presentation of items in Fig.~\ref{fig:toydata}}
        \label{fig:toy2}
    \end{minipage}
\end{figure*}

\noindent{\bf Database:}
Consider a database $\mathcal{D}$ of $n$ tuples, each consisting of $d$ 
attributes $\mathcal{A} = \{ A_1, A_2, \cdots , A_d \}$ that may be involved in a user's preference function\footnote{\small each tuple could also include additional attributes that are not involved in the user preferences. We do not consider these attributes for the purpose of this paper.}. Without loss of generality, we consider $A_i \in \mathbb{R}$ for all $i \in [1, d]$. 
We represent each tuple $t \in \mathcal{D}$ as a $d$-dimensional vector $\{t[1], t[2], \cdots t[d]\}$.

\vspace{2mm}
\noindent{\bf Ranking function:}
Consider a {\em ranking function} $f:\mathbb{R}^d\rightarrow\mathbb{R}$ that maps each tuple $t$ to a numerical score. We further assume, through applying any arbitrary tie-breaker, that no two tuples in the database have the same score - i.e., $\forall t_i, t_j \in \mathcal{D}$ with $i \neq j$, there is $f(t_i) \neq f(t_j)$. We say $t_i$ {\em outranks} $t_j$ if and only if $f(t_i)>f(t_j)$. For each $t\in\mathcal{D}$, let $\mathcal{r}_f(t)$ be the {\em rank} of $t$ in the ordered list of $\mathcal{D}$ based on $f$.  In other words, there are exactly $\mathcal{r}_f(t) - 1$ tuples in $\mathcal{D}$ that outrank $t$ according to $f$.

Ranking functions can take a wide variety of forms.  A popular type of ranking functions studied in the database literature is {\em linear ranking functions}, i.e., 
\begin{align}
f(t) = \sum_{i=1}^{d} w_i\cdot t[i], \label{equ:linearRank}
\end{align}
where $\vec{w}=\{w_1, w_2, \cdots ,w_d\}$ ($\forall i\in [1, d]$, $w_i > 0$) is a weight vector used to capture the importance of each attribute to the final score of a tuple. We use $\mathcal{L}$ to refer to the set of all possible linear ranking functions.

\vspace{2mm}
\noindent{\bf Maxima representation:} For a given database $\mathcal{D}$, if the set of ranking functions of interest is known - say $\mathcal{F}$ - then we can derive a compact {\em maxima representation} of $\mathcal{D}$ by only including a tuple $t \in \mathcal{D}$ if it represents the {\em maxima} (i.e., is the No.~1 ranked tuple) for at least one ranking function in $\mathcal{F}$. For example, if we focus on linear ranking functions in $\mathcal{L}$, then the maxima representation of $\mathcal{D}$ is what is known in the computational geometry and database literature as the {\em convex hull} \cite{dantzig1998linear} of $\mathcal{D}$.  Similarly, the set of {\em skyline} tuples \cite{skylineoperator}, a superset of the convex hull, form the maxima representation for the set of all monotonic ranking functions~\cite{asudeh2016discovering}.


A problem with the maxima representation is its potentially large size. For example, depending on the ``curvature'' of the shape within which the tuples are distributed, even in 2D, the convex hull can be as large as $O(n^{1/3})$~\cite{har2011expected}. The problem gets worse in higher dimensions~\cite{weil1993stochastic, asudeh2017}. As shown in \cite{asudeh2017}, in practice, even for a database with dimensionality as small as $d=5$, the convex hull can often be as large as $O(n)$.

To address this issue, we propose in this paper to relax the definition of maxima representation in order to reduce its size.  Specifically, instead of requiring the representation to contain the top-1 item for every ranking function, we allow the representation to stand so long as it contains at least one of the top-$k$ items for every ranking function.  This tradeoff between the compactness of the representation and the ``satisfaction'' of each ranking function is captured in the following formal definitions of {\em rank regret}:
\begin{definition}\label{def:r2f}
Given a subset of tuples $X\subseteq \mathcal{D}$ and a ranking function $f$, the rank-regret of $X$ for $f$ is the minimum rank of all tuples in $X$ according to $f$.
Formally,
$$
RR_f(X) = \underset{\forall t\in X}{\mbox{argmin}} (\mathcal{r}_f(t))
$$
\end{definition}

\begin{definition}\label{def:r2}
Given a subset of tuples $X\subseteq \mathcal{D}$ and a set of ranking functions $\mathcal{F}$, the rank-regret of $X$ for $\mathcal{F}$ is
the maximum rank-regret of $X$ for all functions in $\mathcal{F}$ - i.e.,
$$
RR_\mathcal{F}(X) = \underset{\forall f\in \mathcal{F}}{\mbox{argmax}} (RR_f(X))
$$
\end{definition}
\begin{definition}\label{def:rrr}
Given a set of ranking functions $\mathcal{F}$ and a user input $k \geq 1$, we say $X \subseteq \mathcal{D}$ is a rank-regret representation of $\mathcal{D}$ if and only if $X$ has the rank-regret of at most $k$ for $\mathcal{F}$, and no other subset of $\mathcal{D}$ satisfies this condition while having a smaller size than $X$. Formally:
\begin{align}
\nonumber \min\limits_{\forall X\subseteq \mathcal{D}} ~ & ~ \| X\| \\
\nonumber s.t. ~ & ~ RR_\mathcal{F}(X) \leq k
\end{align}
\end{definition}

\stitle{Problem Formulation:}
Finding the rank-regret representative of the dataset $\mathcal{D}$ is our objective in this paper. Therefore, we define the rank-regret representative (RRR) problem as following:

\vspace{0.02in}
\medskip\noindent
\framebox[\columnwidth]{\parbox{0.9\columnwidth}{ \textbf{\textsc{Rank-Regret Representative (RRR) Problem:}}
\\ \textit{
Given
a dataset $\mathcal{D}$,
a set of ranking functions $\mathcal{F}$,
and a user input $k$,
find the rank-regret representative of $\mathcal{D}$ for $\mathcal{F}$ and $k$ according to Definition~\ref{def:rrr}.
}}}\\

We note that there is a dual formulation of the problem - i.e., a user specifies the output size $|X|$, and aims to find $X$ that has the minimum rank-regret.
Interestingly, a solution for the RRR problem can be easily adopted for solving this dual problem.
Given the solver for RRR, for the set size $x$, one may apply a binary search to vary the value of $k$ in the range $[1,n]$ and, for each value of $k$, call the solver to find RRR. If the output size is larger than $x$, then the search continues in the upper half of the search space for $k$, or otherwise moves to the lower half.
Given an optimal solver for RRR, this algorithm is guaranteed to find the optimal solution for the dual problem at a cost of an additional $\log n$ factor in the running time.

In the rest of the paper, we focus on $\mathcal{L}$, the class of linear ranking functions.

\stitle{Complexity analysis:} 
The decision version of RRR problem asks if there exists a subset of size $r$ of $\mathcal{D}$ that satisfies the rank-regret of $k$.  Somewhat surprisingly, even though no solution for RRR exists in the literature, we can readily use previous results to prove the NP-completeness of RRR.  Specifically, the $(k,\epsilon)$-regret problem studied in Agrawal et al.~\cite{agarwal2017efficient} asks if there exists a set that guarantees the maximum regret-ratio of $\epsilon$ from the top $k$-th item of any linear ranking function.  Note that the $(2,0)$-regret problem is the equivalent of RRR problem for $k=2$.  Given that the NP-completeness proof in \cite{agarwal2017efficient} covers the $(2,0)$-case when $d \geq 3$, through a reduction to the NP-completeness of the convex polytope vertex cover (CPVC) problem proven by Das et al.~\cite{das1997complexity}, the NP-completeness of RRR follows.

We would like to reemphasize that even though the complexity of RRR was established in existing work, RRR is still a novel problem to study because all previous work in the regret ratio area focused on the case where $\epsilon > 0$.  In other words, they seek approximations on the absolute score achieved by tuples in the compact representation - a strategy which, as discussed in the introduction, could lead to a significant increase on rank regret because many tuples may congregate in a small band of absolute scores.  RRR, on the other hand, focus on the rank perspective (i.e., $\epsilon = 0$) and assumes no specific distribution of the absolute scores.

\section{Geometric interpretation of \\items}\label{sec:dual}
In this section, we use the geometric interpretation of items, explain a dual transformation, and propose a theorem that plays a key role in designing the RRR algorithms.

Each item $t\in \mathcal{D}$ with $d$ scalar attributes can be viewed as a point in $\mathbb{R}^d$.
As an example, Figure~\ref{fig:toydata} shows a sample dataset with $n=7$ items, defined over $d=2$ attributes.
Figure~\ref{fig:toy1} shows these items as the points in $\mathbb{R}^2$.
In this space, every linear preference function $f$ with the weight vector $\vec{w}=\{w_1, w_2, \cdots,w_d\}$ can be viewed as a ray that starts at the origin and passes through the point $\{w_1, w_2, \cdots,w_d\}$.
For each item $t\in\mathcal{D}$, consider the orthogonal to the ray of $f$ that passes through $t$; let the projection of $t$ on $f$ be the intersection of this line with the ray of $f$.
The ordering of items based on $f$ is the same as the ordering of the projection of them on $f$ where the items that are farther from the origin are ranked higher.
For example, Figure~\ref{fig:toy1} shows the ray of the function $f=x_1+x_2$, as well as the ordering of items based on it.
As shown in the figure, the items are ranked as $t_7$, $t_3$, $t_5$, $t_1$, $t_2$, $t_6$, and $t_4$, based on $f=x_1+x_2$.
Every ray starting at the origin in $\mathbb{R}^d$ is represented by $d-1$ angles. For example in $\mathbb{R}^2$, every ray is identified by a single angle. In Figure~\ref{fig:toy1}, the ray of function $f=x_1+x_2$ is identified by the angle $\theta=\pi/4$.

Small changes in the weights of a function will move the corresponding ray slightly, and hence change the projection points of items. However, it may not change the ordering of items.
In fact, while the function space is continuous and the number of possible weight vectors is infinite, the number of possible ordering between the items is, combinatorially, bounded by $n!$.

In order to study the ranking of items based on various functions, throughout this paper, we consider the dual space~\cite{Ed87} that transforms a tuple $t$ in $\mathbb{R}^d$ to the hyperplane $\mathsf{d}(t)$ as follows:
\begin{align}\label{eq:dual}
\mathsf{d}(t):~ \sum\limits_{i=1}^d t[i].x_i=1
\end{align}
In the dual space, the ray of a linear function $f$ with the weight vector $w = \{w_1, w_2,\cdots, w_d\}$ remains the same as the original space: the origin-starting ray that passes through the point $w$. 
The ordering of items based on a function is specified by the intersection of hyperplanes $\mathsf{d}(t_i)$ with it. 
However, unlike the original space, the intersections that are closer to the origin are ranked higher.
Using Equation~\ref{eq:dual}, every tuple in two dimensional space gets transformed to the line $\mathsf{d}(t):~ t[1].x_1 + t[2].x_2=1$. 
Figure~\ref{fig:toy2} shows the items in the example dataset of Figure~\ref{fig:toydata} in the dual space.
Looking at the intersection of dual lines with the $x_1$ axis in Figure~\ref{fig:toy2}, one can see that the ordering of items based on $f=x_1$ is $t_7$, $t_1$, $t_3$, $t_2$, $t_5$, $t_4$, and $t_6$; hence, for any set $X$ containing $t_7$ or $t_1$, for $f=x_1$ (i.e., $w=\{1,0\}$), $RR_f(X)\leq 2$.

The set of dual hyperplanes defines a dissection of $\mathbb{R}^d$ into connected convex cells named as the arrangement of hyperplanes~\cite{Ed87}.
The borders of the cells in the arrangement are $d-1$ dimensional facets.
For example, in Figure~\ref{fig:toy2}, the arrangement of dual lines dissect the $\mathbb{R}^2$ space into connected convex polygons. The borders of the convex polygons are one dimensional line segments.
For every facet in the arrangement consider
a line segment starting from the origin and ending on it.
Let the level of the facet be the number of hyperplanes intersecting this line segment.
We define a top-$k$ border (or simply $k$-border) as the set of facets having level $k$.
For example, the red chain consisting of piecewise linear segments in Figure~\ref{fig:toy2}, shows the top-$k$ border for $k=2$. 
For any function $f$, the hyperplanes intersecting the ray of $f$ on or below the top-$k$ border are the top-$k$.
Looking at the red line in Figure~\ref{fig:toy2}, one may confirm that:
\begin{itemize}
\item The top-$k$ border is not necessarily convex.
\item A dual hyperplane $\mathsf{d}(t_i)$ may contain more than one facet of the top-$k$ border. For example, $\mathsf{d}(t_3)$ in Figure~\ref{fig:toy2} contains two line segments of the top-$2$ border.
\end{itemize}
In the following, we propose an important theorem that is the key to designing the 2D algorithm, as well as the practical algorithm in MD.

\begin{theorem}\label{th:max2k}
For any item $t\in\mathcal{D}$ consider two (if any) functions $f$ and $f'$ where $\mathcal{r}_f(t)\leq k_1$ and $\mathcal{r}_{f'}(t)\leq k_2$.
Also, consider a line segment $l_{f,f'}$ starting from a point on the ray of $f$ and ending at a point on the ray of $f'$.
For any function $f''$ that its ray intersects $l_{f,f'}$, $\mathcal{r}_{f''}(t)\leq k_1+k_2$.
\end{theorem}

\begin{proof}
We use the dual space and prove the theorem by contradiction.
In the dual space, consider the 2D plane passing through the rays of $f$ and $f'$ -- referred as $\mathbb{R}^2_{f,f'}$.
Note that $\mathbb{R}^2_{f,f'}$ is the affine space for the origin starting rays that intersect $l_{f,f'}$.
The intersection of each hyperplane $\mathsf{d}(t_i)$ and this plane is a line that we name as $\mathsf{L}(t_i)$.
The arrangement of lines $\mathsf{L}(t_i)$, $\forall t_i\in \mathcal{D}$, identify the orderings of items $t\in\mathcal{D}$ based on any origin-starting ray (function) that falls in $\mathbb{R}^2_{f,f'}$. 
This is similar to Figure~\ref{fig:toy2} in that the arrangement of lines $\mathsf{d}(t_i)$ identify the possible ordering of items in Figure~\ref{fig:toydata}.
For any pair of items $t_1$ and $t_2$, the intersection of the lines $\mathsf{L}(t_1)$ and $\mathsf{L}(t_2)$
shows the function (the origin-starting ray that passes through the intersection) that ranks $t_1$ and $t_2$ equally well, while on one side of this point $t_1$ outranks $t_2$, but $t_2$ outranks $t_1$  on the other side.
Note that since $\mathsf{L}(t_1)$ and $\mathsf{L}(t_2)$ are both (one dimensional) lines, they intersect at most once.

Now consider the point $t$ and its corresponding line $\mathsf{L}(t)$ in the arrangement.
Since $\mathcal{r}_f(t)\leq k_1$, there exist at most $k-1$ lines below it on the ray of $f$.
Moving from the ray of $f$ toward the ray of $f'$, in order for $t$ to have a rank greater than $k_1+k_2$, $\mathsf{L}(t)$ has to intersect with at least $k_2$ lines $\mathsf{L}(t_i)$ in a way that after the intersection points (toward $f'$) those points outrank $t$.
Since every pair of lines has at most one intersection point, $\mathsf{L}(t)$ will not intersect with those lines any further.
As a result, those (at least) $k_2$ points keep outranking $t$, and thus $t$ cannot have a rank smaller than or equal to $k_2$ again, which contradicts the fact that $\mathcal{r}_{f'}(t)\leq k_2$.
\end{proof}


Intuitively, Theorem~\ref{th:max2k} states that if $f''$ lies ``between'' $f$ and $f'$, then the rank of an item based on $f''$ is at worst the summation of its rank in $f$ and $f'$.
We use this result in the next section, as well as \S~\ref{sec:md}, for providing approximation algorithms for RRR.

\section{RRR in 2D}\label{sec:2d}
In this section, we study the special case of two dimensional (2D) data in which $d=2$. 
In \S~\ref{sec:pre}, we discussed the complexity of the problem for $d\geq 3$. However, we believe that the complexity of the problem is due to the complexity of covering the possible top-$k$ results and therefore, provide an approximation algorithm for 2D.
We consider the items in the dual space and use Theorem~\ref{th:max2k} as the key for designing the algorithm \twodrrr. Later on, in \S~\ref{sec:md}, we also use this theorem for designing a practical algorithm for the multi-dimensional cases.

Based on our discussion about the top-$k$ border in the previous section, each dual line may contain multiple segments of the top-$k$ border. As a results,
for each item, the set of functions for which the item is in the top-$k$, is a collection of disjoint intervals.
Based on Theorem~\ref{th:max2k}, if we take the union of these intervals (i.e., the convex closure), 
we get a single interval, in which the item is guaranteed to be in the top-2$k$. This, we are effectively applying Theorem~\ref{th:max2k} to get the 2-approximation factor.

At a high-level, the algorithm \twodrrr consists of two parts. It  first makes an angular sweep of a ray anchored at the origin, from the x-axis (angle $0^\circ$) toward the y-axis (angle $\pi/2^\circ$)  so that for every item $t\in\mathcal{D}$, it finds the first (smallest angle) and the last function (largest angle) for which $t$ is in top-$k$.
Then it transforms the problem into an instance of one-dimensional range cover~\cite{bronnimann1995GHS} and solves it optimally.

The first part, i.e., the angular sweep, is described in Algorithm~\ref{alg:2dfr}.
For every item $t$ the algorithm aims to find the first ($b[t]$) and the last ($e[t]$) function for which $\mathcal{r}_f(t)\leq k$.
Algorithm~\ref{alg:2dfr} initially orders the items based on their $x_1$-coordinates and puts them in a list $L$ that keeps tracks of orderings while moving from x to y-axis. 
It uses a min-heap data structure to maintain the ordering exchanges between the adjacent items in $L$.
Please note that each ordering exchange is always  between two adjacent items in $L$.
Using Equation~\ref{eq:dual}, the angle of the ordering exchange between two items $L_i$ and $L_{i+1}$ is as follows:
$$
\theta_{L_i,L_{i+1}} = \arctan \frac{L_{i+1}[1]-L_i[1]}{L_i[2]-L_{i+1}[2]}
$$

For the items that are initially in the top-$k$, Algorithm~\ref{alg:2dfr} sets $b[t]$ to the angle $0^\circ$.
Then, it sweeps the ray and pops the next ordering exchange from the heap. Upon visiting an ordering exchange, the algorithm updates the ordered list $L$. If the exchange occurs between the items at rank $k$ and $k+1$:
(i) if this is the first time $L_{k+1}$ enters the top-$k$, the algorithm sets $b[L_{k+1}]$ as the current angle, and (ii) for the item $L_k$ that  leaves the top-$k$, it sets $e[k]$ to the current angle.
The algorithm will update $e[k]$ later on if it becomes a top-$k$ again.
Figure~\ref{fig:toy3} shows the ranges for the example dataset in Figure~\ref{fig:toydata} and $k=2$ ($k$-border is shown in Figure~\ref{fig:toy2}).

\begin{figure}[!t]
    \centering
    \includegraphics[width=0.38\textwidth]{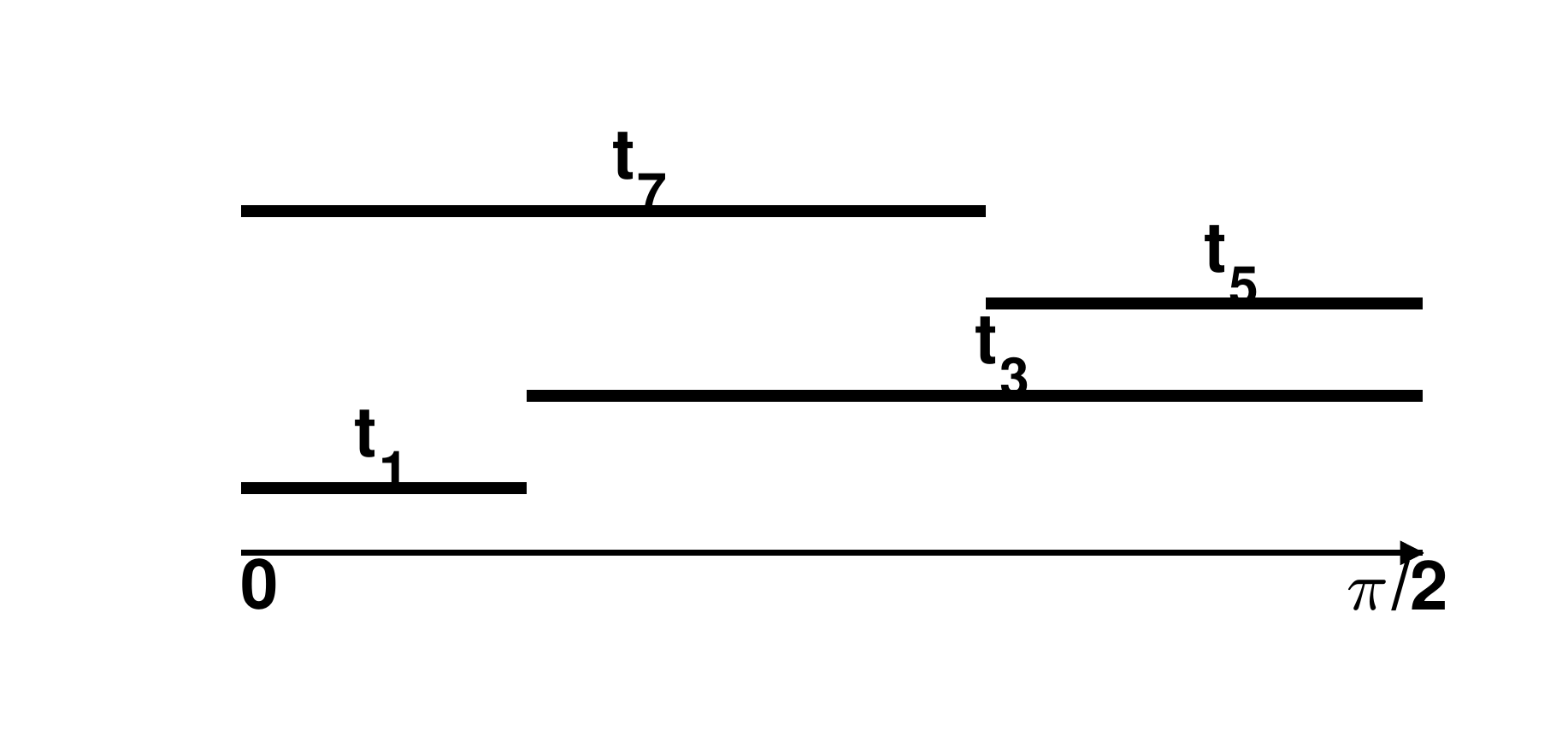}
    \vspace{-7mm}\caption{The ranges for Figure~\ref{fig:toy2}}
    \label{fig:toy3}
\end{figure}

\vspace{3mm}
\begin{algorithm}[!h]
\caption{{\bf FindRanges} \\
         {\bf Input:} 2D dataset $\mathcal{D}$, $n$, $k$ 
        }
\begin{algorithmic}[1]
\label{alg:2dfr}
    \STATE heap = {\it new} min-heap$()$; visited = {\it new} set$()$
    \STATE $L$ = sort $\mathcal{D}$ based on $x$
    \FOR{$i=1$ to $n-1$}
        \IF{$L_i[2]<L_{i+1}[2]$ \scriptsize{\tt /* skip if $L_i$ dominates $L_{i+1} */$} }
            \STATE heap.push( ($\arctan \frac{L_{i+1}[1]-L_i[1]}{L_i[2]-L_{i+1}[2]},L_i$) )
        \ENDIF
    \ENDFOR
    \STATE {\bf for} $i=1$ to $k$ {\bf do} $b[L_i]=0$
    \WHILE{heap is not empty}
        \STATE ($\theta$, $t$) = heap.pop() \scriptsize{\tt // let $i$ be the index of $t$ in $L$}
        \IF{$i == k$}
            \STATE {\bf if} $b[L_{i+1}]==null$ {\bf then} $b[L_{i+1}] = \theta$
            \STATE $e[L_{i}] = \theta$
        \ENDIF
        \STATE swap $L_i$ and $L_{i+1}$
        \IF{($L_{i-1}[1]<L_{i}[1]$ {\bf or} $L_{i-1}[2]<L_{i}[2]$) {\bf and} $(L_{i-1},L_{i})\notin$ visited)}
            \STATE heap.push( ($\arctan \frac{L_{i}[1]-L_{i-1}[1]}{L_{i-1}[2]-L_{i}[2]},L_{i-1}$) )
            \STATE visited.$add((L_{i-1},L_{i}))$
        \ENDIF
        \IF{($L_{i+1}[1]<L_{i+2}[1]$ {\bf or} $L_{i+1}[2]<L_{i+2}[2]$ {\bf and} $(L_{i+1},L_{i+2})\notin$ visited)}
            \STATE heap.push( ($\arctan \frac{L_{i+2}[1]-L_{i+1}[1]}{L_{i+1}[2]-L_{i+2}[2]},L_{i+1}$) )
            \STATE visited.$add((L_{i+1},L_{i+2}))$
        \ENDIF
    \ENDWHILE
    \STATE {\bf for} $i=1$ to $k$ {\bf do} $e[L_i]=\pi/2$
    \STATE {\bf return} $b$, $e$
\end{algorithmic}
\end{algorithm}

After computing the ranges for the items, the problem is transformed into a one dimensional range cover instance.
The objective is to cover the function space (the range between $0^\circ$ and $\pi/2^\circ$) using the least number of ranges.
The greedy approach leads to the optimal solution for this problem -- that is, at every iteration, select the range with the maximum coverage of the uncovered space.

At every iteration, the uncovered space is identified by a set of intervals.
Due to the greedy nature of the algorithm, the range of each remaining item intersects with at most one uncovered interval.
\begin{figure}[!tb]
    \centering
    \includegraphics[width=0.35\textwidth]{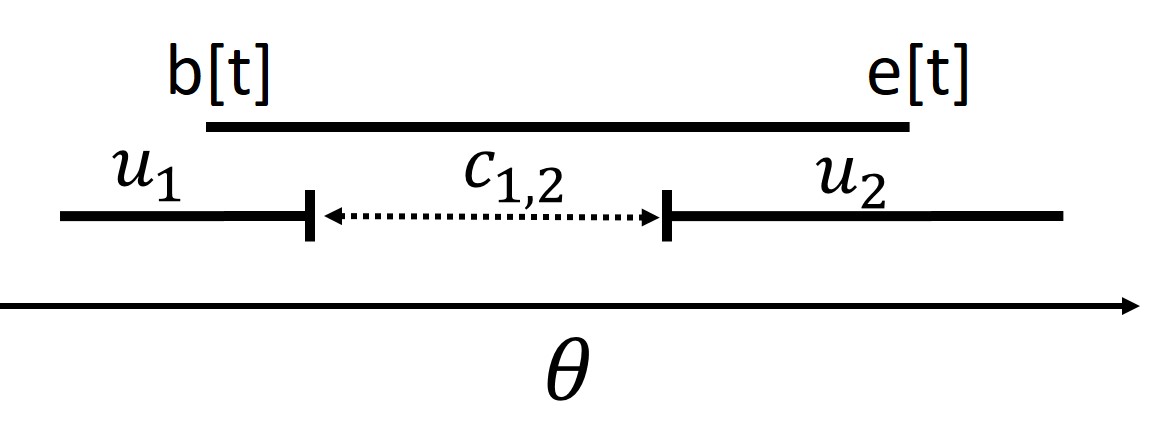}
    \caption{A contradictory example}
    \label{fig:rangecontradiction}
\end{figure}
To explain this by contradiction, consider an item $t$ that its range intersect with two (or more) uncovered intervals (Figure~\ref{fig:rangecontradiction}).
Let $u_1$ and $u_2$ be these intervals. Also, let us name the covered space between $u_1$ and $u_2$ as $c_{1,2}$.
(i) Since the range of $t$ intersects with both $u_1$ and $u_2$, $c_{1,2}$ is contained within the range of $t$, which implies the range of $t$ is larger than $c_{1,2}$.
(ii) $c_{1,2}$ should be covered by the range of at least one previously selected item $t'$. Also, since the ranges of items are continuous, the range of $t'$ cannot be larger than $c_{1,2}$.
As a result, the range of $t'$ is less than the range of $t$, which contradicts the fact that the ranges are selected greedily.

Using this observation, after finding the ranges for each item, \twodrrr (Algorithm~\ref{alg:2d}) uses a sorted list to keep track of the uncovered intervals.
The elements of the list are in the form of $\langle \theta_i, \vdash/\dashv \rangle$, where $\vdash$ (resp. $\dashv$) specifies that this is the beginning (resp. the end) of an uncovered interval.

\vspace{3mm}
\begin{algorithm}[!h]
\caption{\twodrrr \\
         {\bf Input:} 2D dataset $\mathcal{D}$, $n$, $k$
        }
\begin{algorithmic}[1]
\label{alg:2d}
    \STATE $b$,$e$ = {\bf FindRanges($\mathcal{D}$,$n$,$k$)}
    \STATE $\Psi$ = {\it new} set()
    \STATE $U$ = $[\langle 0, \vdash\rangle, \langle \pi/2, \dashv \rangle]$
    \WHILE{$|U|>0$}
        \STATE cov$_m$ $= 0$;
        \FOR{$t_i$ in $\mathcal{D}\backslash \Psi$}
            \STATE{\bf if} $b[t_i]== null$ {\bf then continue}
            \STATE $k$ = the index of the element in $U$ that $b[t_i]$ fall before it (found by applying binary search)
            \STATE{\bf if} $U_k[2]==~ \dashv$ {\bf then} cov $ = \min(U_k[1],e[t_i]) - b[t_i]$ \\ {\bf else} cov $ = \max (0, e[t_i] - U_k[1]$)
            \STATE{\bf if} cov $>$ cov$_m$ {\bf then} t $ = t_i$; cov$_m$ $=$ cov; $k_m = k$
        \ENDFOR
        \STATE $\Psi$.add($t$)
        \IF{$U_{k_m}[2] == \vdash$}
            \STATE {\bf if} $U_{k_m+1}[1]\leq e[t]$ {\bf then} remove $U_{k_m}$ and $U_{k_m+1}$
            \STATE {\bf else} $U_{k_m}[1] = e[t]$
        \ELSE
            \IF{$U_{k_m}[1]>e[t]$}
                \STATE $U.$insert$(k_m, \langle b[t], \dashv\rangle)$; $U.$insert$(k_m+1, \langle e[t], \vdash\rangle)$
            \ELSE
                \STATE $U_{k_m}[1] = b[t]$
            \ENDIF
        \ENDIF
    \ENDWHILE
    \STATE {\bf return} $\Psi$
\end{algorithmic}
\end{algorithm}

At every iteration, for each item that has still not been selected, the algorithm applies a binary search to find the element in $U_k$ that $b[t_i]$ falls right before it, i.e., $U_k[1]\geq b[t_i]$ and $\nexists k'<k$ such that $U_{k'}[1] \geq b[t_i]$.
Then depending on whether $U_k$ specifies the beginning ($\vdash$) or the end ($\dashv$) of an uncovered interval, it computes how much of the uncovered region $t_i$ covers.
The algorithm chooses the item with the maximum coverage, adds it to the selected set, and updates the uncovered intervals accordingly.
It stops when no more uncovered intervals are left.

As an example, for the dataset in Figure~\ref{fig:toydata}, if we execute Algorithm~\ref{alg:2d} on the ranges provided in Figure~\ref{fig:toy3}, it returns the set $\{t_3, t_1\}$.


\begin{theorem}\label{th:2drrr1}
The algorithm \twodrrr runs in $O(n^2\log n)$ time.
\end{theorem}
\begin{proof}
Intuitively, the summation of the cost of each iteration of the greedy algorithm is used to derive the running time.
Please find the details of the proof in \techrep{Appendix~\ref{ap:proofs}}\submit{the technical report~\cite{techreport}}.
\end{proof}

\begin{theorem}\label{th:2drrr2}
The output size of \twodrrr is not more than the size of the optimal solution for RRR.
\end{theorem}
\begin{proof}
The proof follows from the fact that the ranges identified by Algorithm~\ref{alg:2dfr} provide a superset for each top-$k$ result.
Please refer to \techrep{Appendix~\ref{ap:proofs}}\submit{the technical report~\cite{techreport}} for the details.
\end{proof}

\begin{theorem}\label{th:2drrr3}
The output of \twodrrr guarantees the maximum rank-regret of $2k$.
\end{theorem}
\begin{proof}
This result is easy to prove, by applying Theorem~\ref{th:max2k}. The details are provided in \techrep{Appendix~\ref{ap:proofs}}\submit{the technical report~\cite{techreport}}.
\end{proof}
\section{RRR in MD}\label{sec:md}
In multi-dimensional cases (MD) where $d>2$, the continuous function space becomes problematic, the geometric shapes become complex, and even the operations such as computing the volume of a shape and the set operations become inefficient.
Therefore, in this section, we use the $k$-set notion~\cite{Ed87} to take an alternative route for solving the RRR problem by transforming the continuous function space to discrete sets.
This leads to the design of an approximation algorithm that guarantees the rank-regret of $k$, introduces a log approximation-ratio  in the output size, and runs in time polynomial for a constant number of dimensions. We will explain the details of this algorithm in \S~\ref{subsec:mdapprox}.
Then, in \S~\ref{subsec:mdrc}, we propose the function-space partitioning algorithm \mdrc that uses the result of Theorem~\ref{th:max2k} in its design for solving the problem without finding the $k$-sets.
Note that proposed algorithms in this section are also applicable for 2D.
\subsection{k-Set and Its Connection to RRR}\label{sec:kset}

A $k$-set is an important notion in combinatorial geometry with applications including half-space range search~\cite{halfspace1, halfspace2}.
Given a set of points in $\mathbb{R}^d$, a $k$-set is a collection of exactly $k$ points in the point set that are strictly separable from the rest of points using a $d-1$ dimensional hyperplane.

Consider a finite set $P$ of $n$ points in the euclidean space $\mathbb{R}^d$. 
A hyperplane $h$ 
partitions it into $P^+ = P\cap h^+$ and $P^- = P\cap h^-$, called half spaces of $P$, where $h^+$ (resp. $h^-$) is the open half space above\footnote{\small We use the word above (resp. below) to refer to the half space that does not contain (resp. contains) the origin.} (resp. below) $h$~\cite{Ed87}. 
The hyperplane $h$ in the Euclidean space $\mathbb{R}^d$ can be uniquely defined by a point $\rho$ on it and a $d$ dimensional normal vector $v$ orthogonal to it, and has the following equation:  

\begin{align}
v[1](x_1 - \rho [1]) + v[2](x_2 - \rho [2]) + \cdots + v[d](x_d - \rho [d])  = 1
\end{align}

A half space $S$ of $P$ is a $k$-set if $card(S) = k$.
Without loss of generality, we consider the positive half spaces and $v[i]\geq 0$.
That is, $S\subseteq P$ is a $k$-set if
$\exists$ a point $\rho$ and the positive normal vector $v$ such that $S = h(\rho ,v)^+$ and $card( h(\rho,v)^+ ) = k$.
For example, the empty set is a $0$-set and each point in the convex hull of $P$ is a $1$-set. We use $\mathcal{S}$ to refer to the collection of $k$-sets of $P$; i.e., $ \mathcal{S} = \{ S\subseteq P | S$ is a $k$-set$\}$.
For example, Figure~\ref{fig:toy4} shows the collection of $k$-sets for $k=2$ for the dataset of Figure~\ref{fig:toydata}. As we can see, the $2$-sets are $\mathcal{S}=\{ \{t_1, t_7\}, \{t_7, t_3\},\{t_3, t_5\}\}$.

\begin{figure}[!t]
    \centering
    \includegraphics[width=0.38\textwidth]{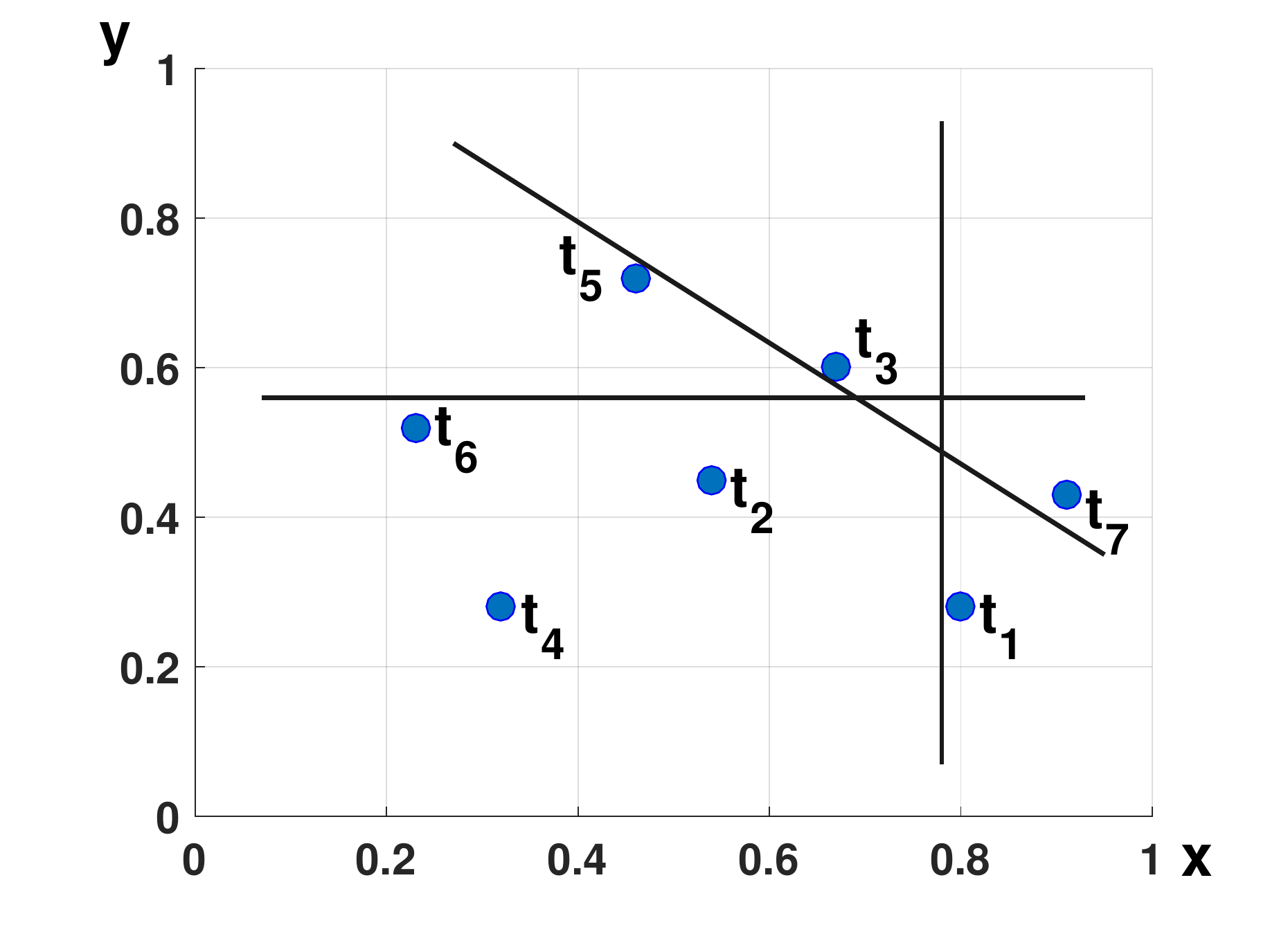}
    \vspace{-3mm}\caption{The $k$-sets of Figure~\ref{fig:toydata} for $k=2$.}
    \label{fig:toy4}
\end{figure}

If we consider items $t\in\mathcal{D}$ as points in $\mathbb{R}^d$, the notion of $k$-sets is interestingly related to the notion of top-$k$ items, as the following arguments show:
\begin{itemize}
\item A hyperplane $h(\rho,v)$ describes the set of all points with the same score as point $\rho$, for the ranking function $f$ with the weight vector $v$, i.e., the set of attribute-value combinations with the same scores as $\rho$ based on the ranking function $f$.
\item If we consider a hyperplane  $h(\rho,v)$ where $card( h(\rho,v)^+ ) = k$, the set of points belonging to $h(\rho,v)^+$ is equivalent to the top-$k$ items of $\mathcal{D}$ for the ranking function with weight vector $v$.
\end{itemize}

\begin{lemma}\label{lemma:1}
Let $\mathcal{S}$ be the collection of all $k$-sets for the points corresponding to the items $t\in\mathcal{D}$.
For each possible ranking function $f$, there exists a $k$-set $S\in\mathcal{S}$ such that top-$k$($f$)=$S$.
\end{lemma}

\begin{proof}
We provide the proof by contradiction. Please refer to \techrep{Appendix~\ref{ap:proofs}}\submit{the technical report~\cite{techreport}} for the details.
\end{proof}

Based on Lemma~\ref{lemma:1}, all possible answers to top-$k$ queries on linear ranking functions can be captured by the collection of $k$-sets.
This will help us in solving the RRR problem
in \S~\ref{subsec:mdapprox}.
As we shall explain in \S~\ref{sec:related}, the best known upper bound on the number of $k$-sets in $\mathbb{R}^2$ and $\mathbb{R}^3$ are $O(n k ^{1/3})$~\cite{dey1998} and $O(n k^{3/2})$~\cite{sharir2000}. For $d>3$, the best known upper bound is $O(n^{d-\varepsilon})$~\cite{ABFK92}, where $\varepsilon>0$ is a small constant\footnote{\small Note that this is polynomial for a constant $d$.}.
However, as we shall show in \S~\ref{sec:exp}, in practice $|\mathcal{S}|$ is significantly smaller than the upper bound.

In \techrep{Appendix~\ref{ap:kset-enum}}\submit{the technical report~\cite{techreport}}, we review the $k$-set enumeration.
For the 2D case, a ray sweeping algorithm (similar to Algorithm~\ref{alg:2dfr}) that follows the $k$-border finds the collection of $k$-sets. For higher dimensions, the enumeration can be modeled as a graph traversal problem~\cite{andrzejak1999optimization}. The algorithm considers the $k$-set graph $G(V,E)$ in which the vertices are the $k$-sets and there is an edge between two $k$-sets if the size of their intersection is $k-1$. We discuss the connectivity of the graph, and explain how to traverse it and enumerate the $k$-sets.

Next, we use the $k$-set notion for developing an approximation algorithm for RRR that guarantees a rank-regret of $k$ and a logarithmic approximation ratio on the output size.

\subsection{MDRRR: Hitting-Set Based Approximation Algorithm}\label{subsec:mdapprox}
As we discussed in \S~\ref{sec:kset} the collection of $k$-sets contains the set of all possible top-$k$ results for the linear ranking functions. As a result, a set of tuples $X \subseteq \mathcal{D}$ that contains at least one item from each $k$-set is guaranteed to have at least one of the items in the top-$k$ of any linear ranking function; which implies that $X$ satisfies the rank regret of $k$.
On the other hand, since every $k$-set $S=h(\rho,v)^+$ is at least the top-$k$ of the linear function $f$ with the weight vector $v$, a subset $X' \subseteq \mathcal{D}$ that does not contain any of the items of a $k$-set $S$ does not satisfy the rank regret of $k$.

\begin{algorithm}[!t]
\caption{\mdrrr\\
		 {\bf Input:} collection of $k$-sets $\mathcal{S}$
		}
\begin{algorithmic}[1]
\label{alg:GHS}
\STATE $D= \underset{\forall s_i\in\mathcal{S}}{\cup} S_i$
\STATE Set weight of each point to one
\WHILE{True}
	\STATE $X$ = Select the $\epsilon$-net
	\IF{$X$ is not hitting set}
        \FOR{$S$ in $\mathcal{S}$}
        	\IF{points in a set k-set $S$ missed by $X$ }
            	\STATE Double the weights of the points in $S$
            \ENDIF
		\ENDFOR
	\ELSE
		\STATE {\bf return} ($X$)
	\ENDIF
\ENDWHILE
\end{algorithmic}
\end{algorithm}

One can see that given the collection of $k$-sets, our RRR problem is similar to the {\em minimum hitting set problem}~\cite{karp1972reducibility}. 
Given a universe of $n$ items $\mathcal{D}$, and a collection of sets $\mathcal{S}$ where each set $S \in \mathcal{S}$ is a subset of $\mathcal{D}$, the minimum hitting set problem asks for the smallest set of items $X' \subseteq \mathcal{D}$ such that $X'$ has a non-empty intersection with every set $S$ of $\mathcal{S}$. The minimum hitting set problem is known to be NP-complete~\cite{karp1972reducibility} and the existing approximation algorithm provides a factor of $O(\log n)$ from the optimal size $c$. 
A deterministic polynomial time algorithm with an improved factor of $O(\delta \log \delta c)$ had been proposed by~\cite{bronnimann1995GHS} for a specific instance of this problem  -- the {\em geometric hitting set problem} --  where $\delta$ is the Vapnik Chervonenkis dimension (VC-dimension).
The VC-dimension is defined as the cardinality of the largest set of points $Y \subseteq \mathcal{D}$ that can be {\em shattered} by $\mathcal{S}$, i.e., the system introduced by $\mathcal{S}$ on $Y$ contains all the subsets of $Y$~\cite{vapnik2013nature}. 
In the RRR problem, since the $k$-sets are defined by half spaces, the VC-dimension is $d$ (the number of attribute)~\cite{VCofHalfSpace, bronnimann1995GHS}.


Next we formally show the mapping of the RRR problem into the  geometric hitting set problem, and provide the detail of approximation algorithm.

\medskip\noindent
 \framebox[\columnwidth]{\parbox{0.9\columnwidth}{ \textsc{Mapping to Geometric hitting set problem:}
Given a set space $R=(D, \mathcal{S})$, where $\mathcal{S}$ is the collection of $k$-sets and $D= \underset{\forall s_i\in\mathcal{S}}{\cup} S_i$ is the set of points, 
find the smallest set $X \subseteq D$ such that $\forall S\in\mathcal{S}, \exists t\in X$ s.t. $t\in S$.
}}\\

In  \mdrrr (Algorithm~\ref{alg:GHS}), we use the approximation algorithm for the geometric hitting set problem that is proposed in~\cite{bronnimann1995GHS} using the concept of $\epsilon$-nets~\cite{haussler1987epsilon}. More formally, an $\epsilon$-net of $D$ for $\mathcal{S}$ is a set of points $X \subseteq D$ such that $X$ contains a point for every $S\in\mathcal{S}$ with size of at least $\epsilon |D|$. Algorithm~\ref{alg:GHS} shows the psudocode of \mdrrr, the approximation algorithm that uses the mapping to geometric hitting set problem. The algorithm initializes the weight of each point to one. It then iteratively, in polynomial time, selects (using weighted sampling) a small-sized set of tuples $X \subseteq D$ 
that intersects all highly weighted sets in $\mathcal{S}$. More formally if a set $X \subseteq D$ intersects each $k$-set $S$ of $\mathcal{S}$ with weight larger than $\epsilon W(D)$, where $W(D)$ is the total weights of of points in $D$, then $X$ is an $\epsilon$-net. If $X$ is not a hitting set (lines 4-9), then the algorithm doubles the weight of the points in the particular sets $S$ of $\mathcal{S}$ missed by $X$.

 \smallskip\noindent{\bf Discussion:} In summary, considering the one-to-one mapping between the RRR problem and the geometric hitting set problem over the collection of $k$-sets, we can see that:
\begin{itemize}
\item \mdrrr guarantees rank-regret of $k$. That is because \mdrrr is guaranteed to return at least one item from each $k$-set in $\mathcal{S}$, the set of all top-$k$ results.
\item \mdrrr guarantees the approximation ratio of $O(d\log dc)$, where $c$ is the optimal output size and $d$ is the number of attributes.
\item \mdrrr runs in polynomial time. This is because 
it has been shown in~\cite{bronnimann1995GHS} that the number of iterations the algorithm must perform is at most $\mathcal{O}(c \log \frac{n'}{c})$, where $n'$ is the number of points in $D$, and $c$ is the size of the optimal hitting set. 
Moreover, recall that \mdrrr needs the collection of $k$-sets, which can be enumerated by traversing the
 $k$-set graph \techrep{(c.f Appendix~\ref{ap:kset-enum})} which runs in polynomial time.
\end{itemize}

\smallskip\noindent
Nevertheless, although it runs in polynomial time, the \mdrrr algorithm is quite impractical as described above. It needs the collection of $k$-sets ($\mathcal{S}$), as input.
Therefore, its efficiency depends on the $k$-set enumeration and the size of $|\mathcal{S}|$.
Although, as we shall show in \S~\ref{sec:exp}, in practice the size of $|\mathcal{S}|$ is reasonable and 
\techrep{as explained in Appendix~\ref{ap:kset-enum},} the $k$-set graph traversal algorithm is linear in $|\mathcal{S}|$, the algorithm does not scale beyond a few hundred items in practice.
The reason is that while exploring each $k$-set, it needs to solve much as $n$ linear programs, each of size $n$ constraints over $d$ variables.
This makes the enumeration extremely inefficient. Therefore, we need to explore practical alternatives to the $k$-set enumeration algorithm.

In the next subsection, we propose a more practical randomized algorithm \ksetr for $k$-set enumeration.

\subsubsection{\ksetr: Sampling for  $k$-set enumeration}

\begin{algorithm}[!tb]
\caption{{\bf \ksetr}\\
		 {\bf Input:} dataset $\mathcal{D}$, termination condition $c$
		}
\begin{algorithmic}[1]
\label{alg:kset_baseline}
\STATE $\mathcal{S}_r=\{\}$ ,counter=0
\WHILE{counter$\leq c$}
	\STATE \texttt{\scriptsize // generate a sample function}
	\FOR{$i=1$ to $d$}
		\STATE $w_i=|N(0,1)|$ \texttt{\scriptsize // N(0,1) draws a sample from the standard normal distribution}
	\ENDFOR
	\STATE \texttt{\scriptsize // find the corresponding $k$-set}
	\STATE $S$ = top-$k$($\mathcal{D}$, $f_w$)
	\IF{$S\in\mathcal{S}_r$}
		\STATE counter = counter+1
	\ELSE
		\STATE add $S$ to $\mathcal{S}_r$
		\STATE counter = 0
	\ENDIF
\ENDWHILE
\STATE {\bf return} ($\mathcal{S}_r$)
\end{algorithmic}
\end{algorithm}

Here we propose a sampling-based alternative for the $k$-set enumeration.
There is a many to one mapping between the linear ranking functions and the $k$-sets. That is, while 
a $k$-set is the top-$k$ of infinite number of linear ranking functions, every ranking function is mapped to only one $k$-set, the set of top-$k$ tuples for that function.
Instead of the exact enumeration of the $k$-sets, which requires solving expensive linear programming problems for the discovery of the $k$-sets, we propose a randomized approach based on the {\em coupon collector's problem}~\cite{couponcollector}. The coupon collector's problem describes the ``collect the coupons and win'' contest. Given a set of coupons, consider a sampler that every time picks a coupon uniformly at random, with replacement. The requirement is to keep sampling until all coupons are collected. Given a set of $\nu$ coupons, it has been shown that the expected number of samples to draw is in $\Theta(\nu\log \nu)$.
We use this idea for collecting the $k$-sets by generating random ranking functions and taking their top-$k$ results as the $k$-sets.
This is similar to the coupon collector's problem setting, except that the probabilities of retrieving the $k$-sets are not equal. For each $k$-set, this probability depends on the portion of the function space for which it is the top-$k$.
Therefore, rather than applying a $k$-set enumeration algorithm, Algorithm~\ref{alg:kset_baseline}, repeatedly generates random functions and computes their corresponding $k$-sets, stopping when it does not find a new $k$-set after a certain number of iterations.
The algorithm returns the collection of $k$-sets it has discovered, as $\mathcal{S}_r$.
Recall that the function space in MD, is modeled by the universe of origin-starting rays.
The set of points on the surface of the (first quadrant of the) unit hypersphere represent the universe of origin-starting rays.
Therefore, uniformly selecting points from the surface of the hypersphere in $\mathbb{R}^d$, is equivalent to uniformly sampling the linear functions.
Algorithm~\ref{alg:kset_baseline} adopts the method proposed by
Marsaglia~\cite{marsaglia1972choosing} for uniformly sampling points on the surface of the unit hypersphere, in order to generate random functions. It generates the weight vector of the sampled function as the absolute values of $d$ random normal variables. 
We note that since the $k$-sets are not collected uniformly by \ksetr, its running time is not the same as coupon collector's problem, but as we shall show in \S~\ref{sec:exp}, it runs well in practice.

After finding $\mathcal{S}_r$, using Algorithm~\ref{alg:kset_baseline}, we
pass it, instead of $\mathcal{S}$ to \mdrrr.
Since Algorithm~\ref{alg:kset_baseline} does not guarantee the discovery of all $k$-sets, the output of the hitting set algorithm does not guarantee the rank-regret of $k$ for the missing $k$-sets. However, the missing $k$-sets (if any) are expected to be in the very small regions that has never been hit by a randomly generated function.
Also, the fact that the adjacent $k$-sets in the $k$-set graph vary in only one item,  further reduces the chance that a missing $k$-set is not covered.
Therefore, this is very unlikely that the top-$k$ of a randomly generated function is not within the output.

On the other hand, since Algorithm~\ref{alg:kset_baseline} finds a subset of $k$-sets, the output size for running the hitting set on top of the subset (i.e., $\mathcal{S}_r$) is not more than the output size of running the hitting set on $\mathcal{S}$. As a result, the output size remains within the logarithmic approximation factor.

\subsection{MDRC: Function Space Partitioning}\label{subsec:mdrc}

\begin{figure}[!bt]
\centering
\includegraphics[width=0.18\textwidth]{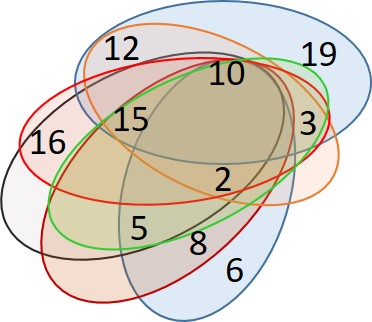}
\caption{Illustration of overlap between the $k$-sets of a sample of 20 items from the DOT dataset (c.f. \S~\ref{sec:exp}) while $d=2$}\label{fig:ksetsexample}
\end{figure}

Given the collection of $k$-sets, the hitting set based approximation algorithm \mdrrr guarantees the rank-regret of $k$ while introducing a logarithmic approximation in its output size.
Despite these nice guarantees, \mdrrr still suffers from $k$-set enumeration, as it can only be executed after the $k$-sets have been discovered.
Therefore, as we shall show in \S~\ref{sec:exp}, in practice it does not scale well for large problem instances.
One observation from the $k$-set graph is the high overlap between the $k$-sets, as the adjacent $k$-sets differ in only one item.
As a result many of them may share at least one item.
For example, we selected 20 random items from the DOT (Department of Transportation) dataset (c.f. \S~\ref{sec:exp}) while setting $d=2$. By performing an angular sweep of a ray from the x-axis to the y-axis while following the $k$-border (see Figure~\ref{fig:toy2}), we enumerated the $k$-sets. In Figure~\ref{fig:practicalrunning}, we illustrate the overlap between these $k$-sets.
The figure confirms the high overlap between the $k$-sets where the item with id 2 appears in all except one of the sets. This motivates the idea of finding these items without enumerating the $k$-sets.
In addition, the top-$k$ of two similar functions (where the angle between their corresponding rays is small) are more likely to intersect. 

We uses these observations in this subsection and propose the function-space partitioning algorithm \mdrc which (similar to the 2D algorithm \twodrrr) leverages Theorem~\ref{th:max2k} in its design.
The algorithm is based on the extension of Theorem~\ref{th:max2k} that bounds the rank of an item that appears in the top-$k$ of the functions corresponding to the corners of a convex polytope in the function space.
 
\mdrc considers the function space in which every function (i.e., a ray starting from the origin) in $\mathbb{R}^d$ is identified as a set of $d-1$ angles.
Rather than discovering the $k$-sets and transforming the problem to a hitting set instance, here our objective is to cover the {\em continuous function space} (instead of the discrete $k$-set space).
Intuitively, 
we propose a recursive algorithm which, at every recursive step, considers a hyper-rectangle in the function space, and either assigns a tuple to the functions in the space, or uses a round robin strategy on the $d-1$ angles to break down the space in two halves, and to continue the algorithm in each half. 
This partitioning strategy is similar to the Quadtree data structure~\cite{finkel1974quad}.
The reason for choosing this strategy is to maximize the similarity of the functions in the corners of the hyper-rectangles to increase the probability that their top-$k$ sets intersect.
\mdrc also follows a {\em greedy} strategy in covering the function space, by partitioning a hyper-rectangle only if it cannot assign a tuple to it.

Consider the space of possible ranking functions in $\mathbb{R}^d$.
This is identified by a set of $d-1$ angles $\Theta = \{\theta_1, \theta_2,\cdots , \theta_{d-1}\}$, where $\theta_i\in [0,\pi/2]$.
To explain the algorithm, consider the binary tree where each node is associated with a hyper-rectangle in the angle space, specified by a range vector of size $d-1$.
The root of the tree is the complete angle space, that is the hyper-rectangle defined between the ranges $[0,\pi/2]$ on each dimension.
Let the level the nodes increase from top to bottom, with the level of the root being $0$.
Every node at level $l$ uses the angle $\theta_{l\%(d-1) + 1}$ to partition the space in two halves, the negative half (left children) and the positive half (the right child).
Figure~\ref{fig:practicalrunning} illustrates an example of such tree for 3D. The root uses the angle $\theta_1$ to partition the space. The left child of the root is associated with the rectangle specified by the ranges $\{[0,\pi/4],[0,\pi/2]\}$ and the right child shows the one by $\{[\pi/4,\pi/2],[0,\pi/2]\}$. 
The nodes at level $1$ use the angle $\theta_{1\%2 + 1} = \theta_2$ for partitioning the space.

\begin{figure}[!tb]
\centering
\includegraphics[width=0.35\textwidth]{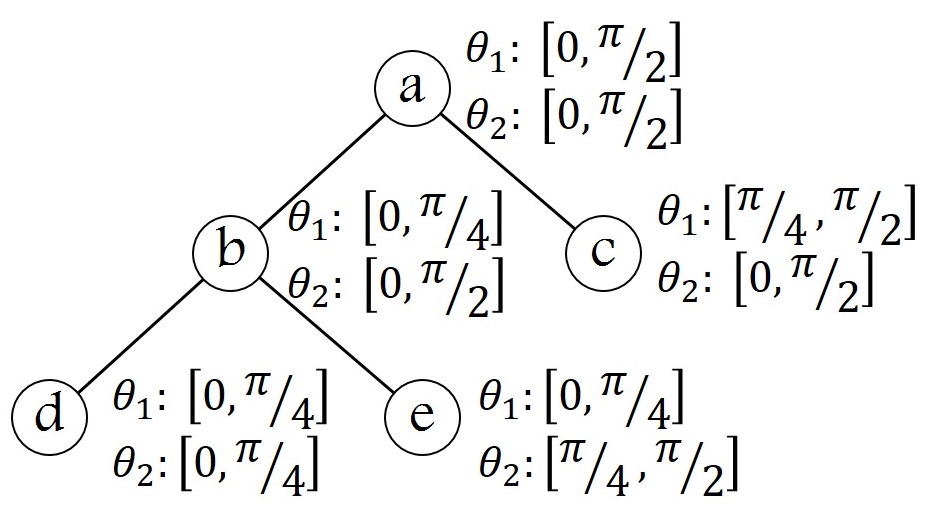}
\caption{Illustration of space partitioning and the recursion tree of Algorithm~\ref{alg:mdpractical2}}\label{fig:practicalrunning}
\end{figure}

At every node, the algorithm checks the top-$k$ items in the corners of the node's hyper-rectangle and if there exists an item that is common to all of them, returns it.
Otherwise, it generates the children of the node and iterates the algorithm on the children. The algorithm combines the outputs of each of the halves as its output.
Algorithm~\ref{alg:mdpractical2} shows the pseudocode of the recursive algorithm \mdrc.
The algorithm is started by calling \mdrc$(\mathcal{D}, n, d, k, 0, \\ \{[0,\pi/2]~| \forall 0<i<d\})$.

\begin{algorithm}[!tb]
\caption{{\bf \mdrc}\\
		 {\bf Input:} The dataset $\mathcal{D}$, $n$, $d$, $k$, level of the node: $l$, ranges: $R$
		}
\begin{algorithmic}[1]
\label{alg:mdpractical2}
	\STATE $C =$ corners of the hypercube specified by $R$
	\STATE $I = \underset{\forall c_i\in C}{\cap}~\mbox{top-}k(\mathcal{D}, c_i)$
    \STATE {\bf if} $|I|>0$ {\bf then return} $I[1]$
    \STATE $i = l\%(d-1) + 1$
    \STATE mid = $\frac{R[i][1]+R[i][2]}{2}$
    \STATE $lR = rR = R$
    \STATE $lR[i][2] =$ mid; $rR[i][1] =$ mid;
    \STATE {\bf return } \mdrc$(\mathcal{D}, n, d, k, l+1, lR)\cup$ \mdrc$(\mathcal{D}, n, d, k, l+1, rR)$
\end{algorithmic}
\end{algorithm}

As a running example for the algorithm, let us consider Figure~\ref{fig:practicalrunning}.
The algorithm starts at the root, partitions the space in two halves, as the intersection of the top-$k$ of its hyper-rectangle's corners are empty, and does the recursion at nodes $b$ and $c$.
The node $c$ finds the item $t_c$ which appears in the top-$k$ of all of its corners and returns it to $a$.
Node $b$, however, cannot find such an item and does the recursion by partitioning its hyper-rectangle along the angle $\theta_2$.
Nodes $d$ and $e$ find the items $t_d$ and $t_e$ and return them to $b$ which returns $\{t_d,t_e\}$ to the root. The root returns $\{t_c, t_d, t_e\}$ as the representative.

\begin{theorem}\label{th:mdrc2}
The algorithm \mdrc guarantees the maximum rank-regret of $dk$.
\end{theorem}
\begin{proof}
This proof uses Theorem~\ref{th:max2k} to extend the maximum rank bound from one dimensional ranges to $(d-1)$ dimensions. Please find the details in \techrep{Appendix~\ref{ap:proofs}}\submit{the technical report~\cite{techreport}}.
\end{proof}
Theorem~\ref{th:mdrc2} uses the result of Theorem~\ref{th:max2k} to provide an upper bound on the maximum rank of the items assigned to each hyper-rectangle, for the functions inside it.
However, as we shall show in \S~\ref{sec:exp}, the rank-regret of its output in practice is much less. For all the experiments we ran, the output of \mdrc satisfied the maximum rank of $k$ for all settings.
Also, following the greedy nature in partitioning the function space, as we shall show in \S~\ref{sec:exp}, the output of \mdrc in all cases was less than 40.
In addition, in \S~\ref{sec:exp}, we show that this algorithm is very efficient and scalable in practice.


\begin{figure*}[ht]
    \begin{minipage}[t]{0.23\linewidth}
        \centering
        \includegraphics[scale=0.24]{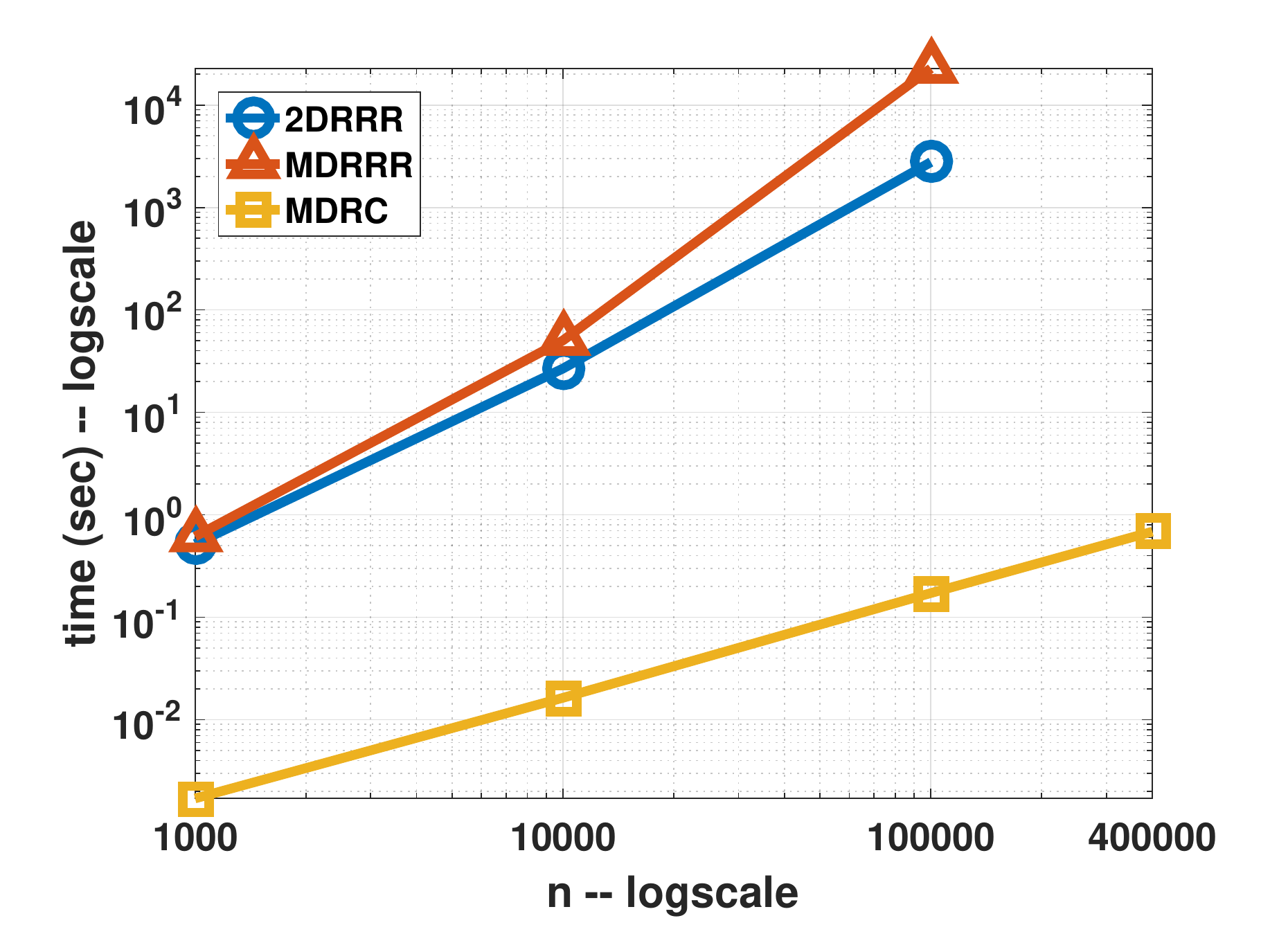}
        \vspace{-8mm}\caption{DOT dataset, 2D, Efficiency: Impact of dataset size ($n$)}
        \label{fig:DOT2DVN1}
    \end{minipage}
    \hspace{0mm}
    \begin{minipage}[t]{0.23\linewidth}
        \centering
        \includegraphics[scale=0.24]{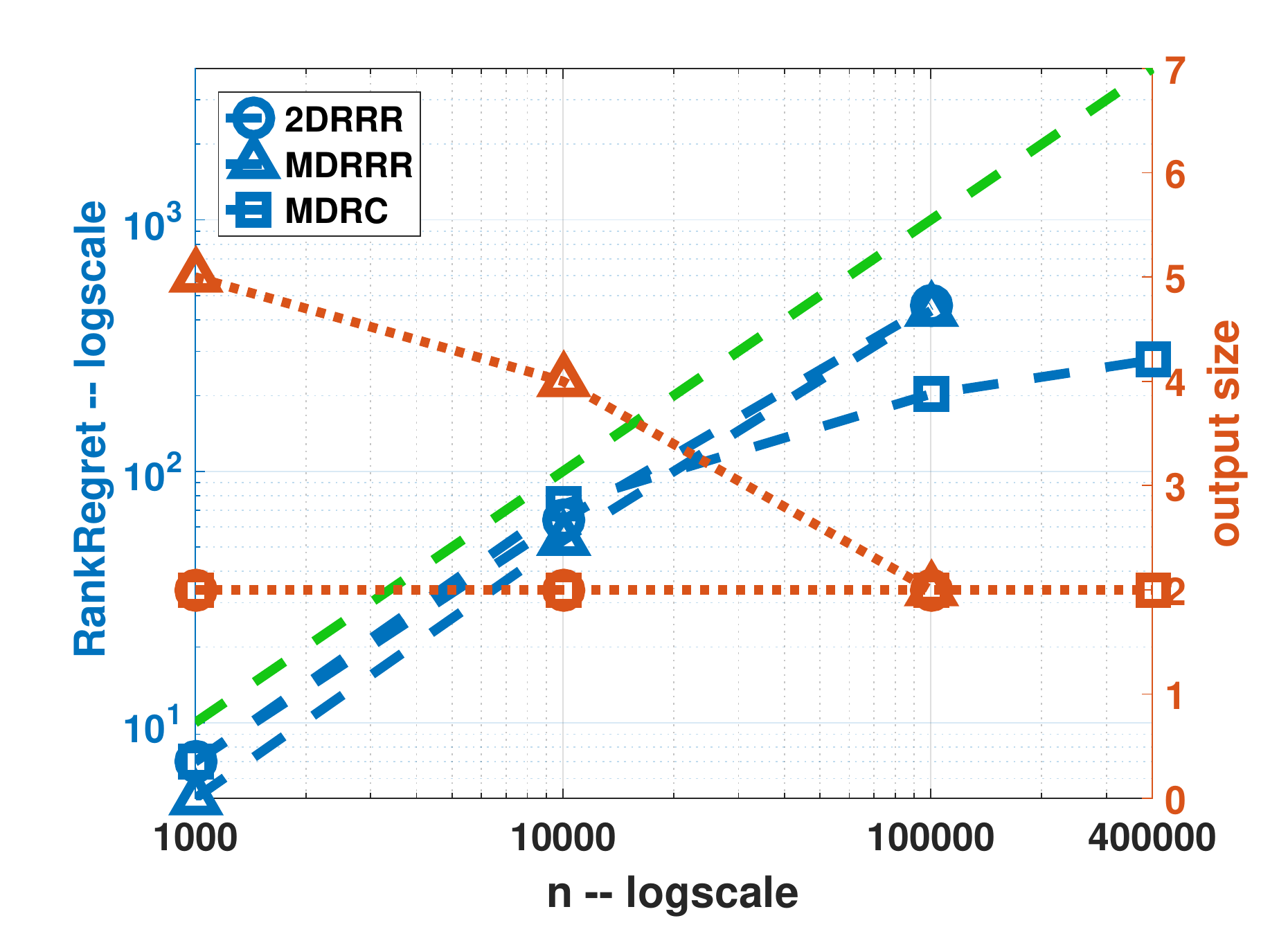}
        \vspace{-8mm}\caption{DOT dataset, 2D, Effectiveness: Impact of dataset size ($n$)}
        \label{fig:DOT2DVN2}
    \end{minipage}
    \hspace{3mm}
    \begin{minipage}[t]{0.23\linewidth}
        \centering
        \includegraphics[scale=0.24]{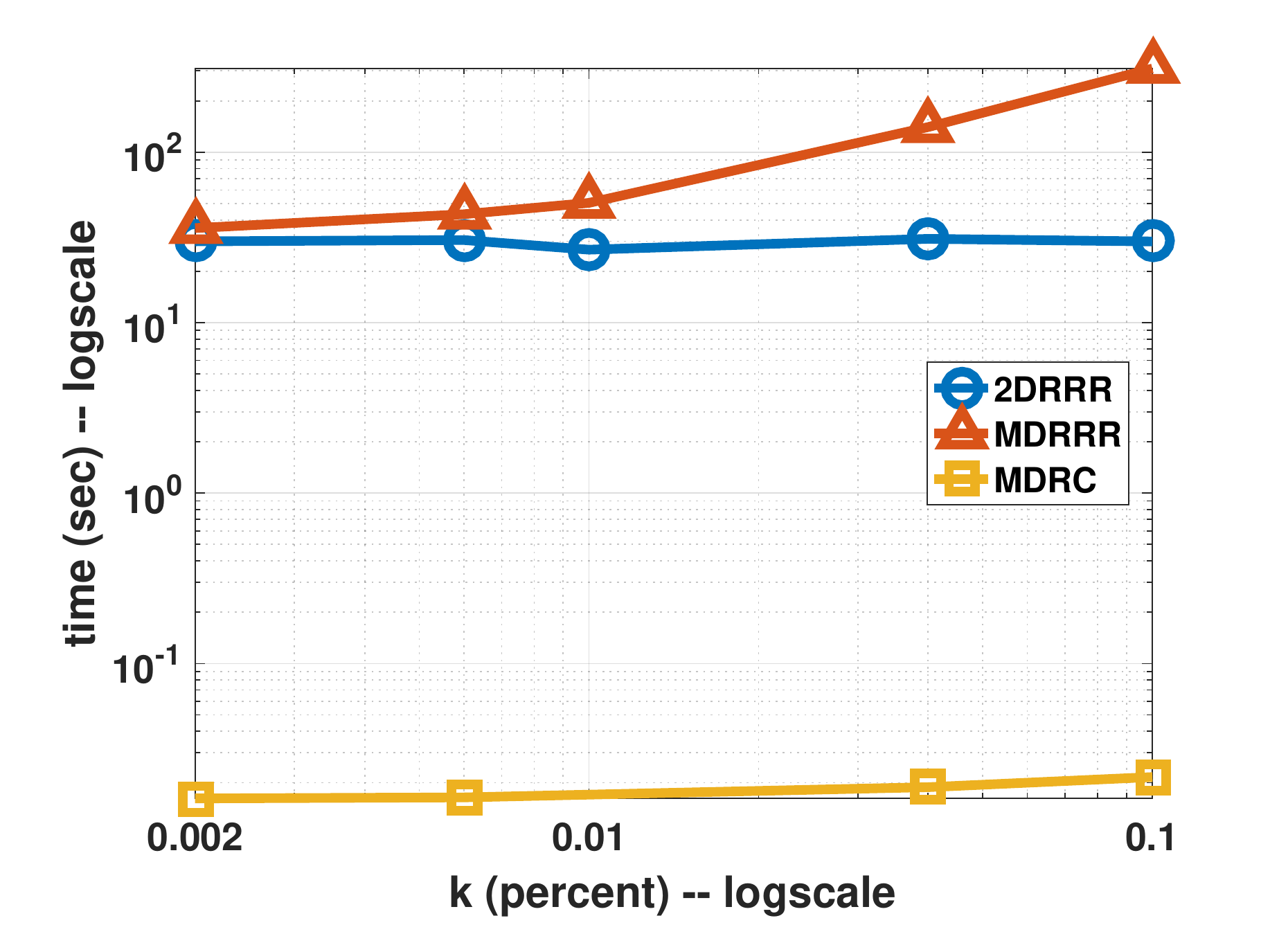}
        \vspace{-8mm}\caption{DOT dataset, 2D, Efficiency: Impact of $k$}
        \label{fig:DOT2DVK1}
    \end{minipage}
    \hspace{1mm}
    \begin{minipage}[t]{0.23\linewidth}
        \centering
        \includegraphics[scale=0.24]{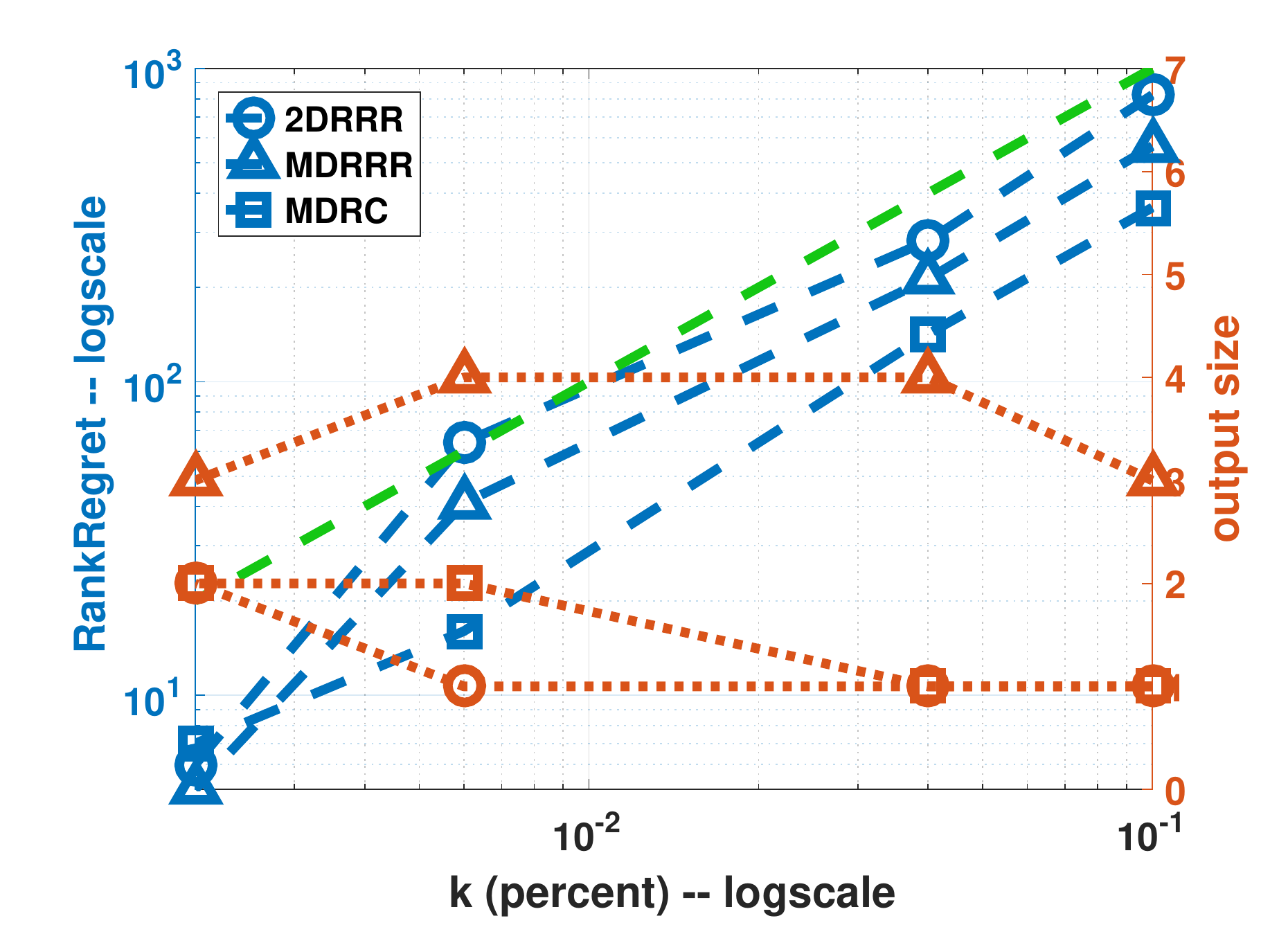}
        \vspace{-8mm}\caption{DOT dataset, 2D, Effectiveness: Impact of $k$}
        \label{fig:DOT2DVK2}
    \end{minipage}
    \hspace{-2mm}
\end{figure*}

\begin{figure*}[ht]
    \begin{minipage}[t]{0.23\linewidth}
        \centering
        \includegraphics[scale=0.24]{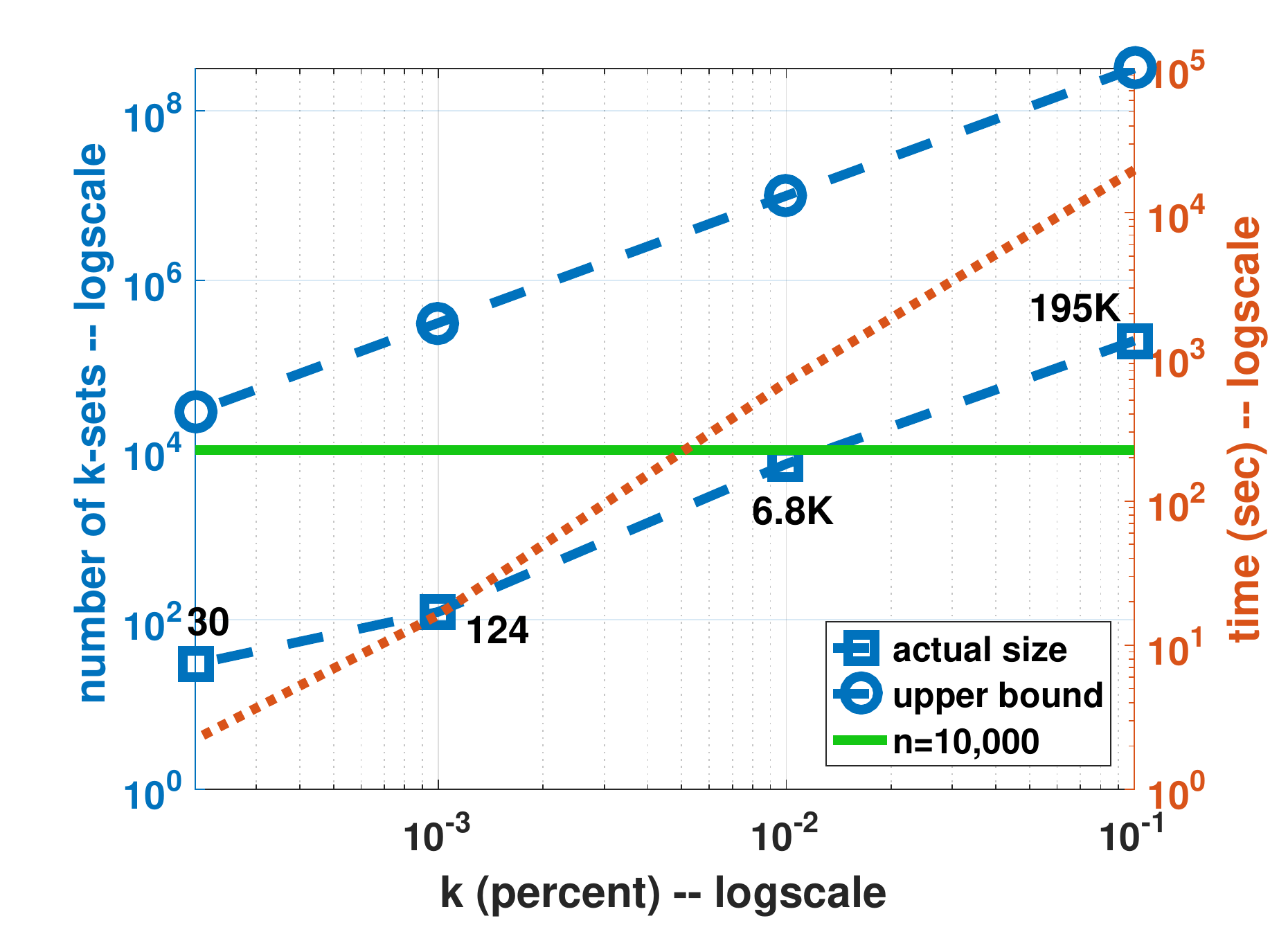}
        \vspace{-8mm}\caption{DOT dataset, MD: Impact of $k$ on $|\mathcal{S}|$}
        \label{fig:DOTKSVK}
    \end{minipage}
    \hspace{2mm}
    \begin{minipage}[t]{0.23\linewidth}
        \centering
        \includegraphics[scale=0.24]{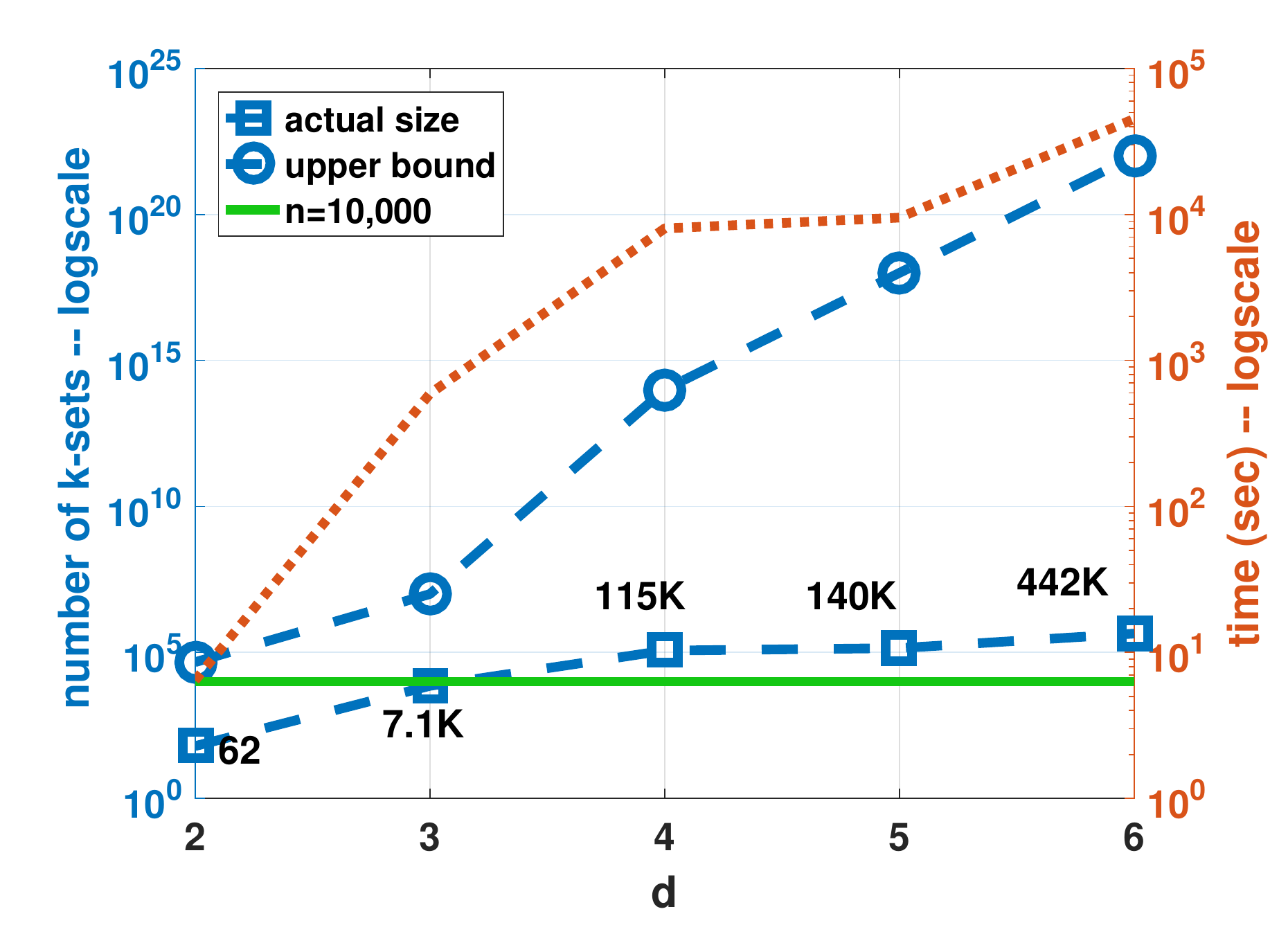}
        \vspace{-8mm}\caption{DOT dataset, MD: Impact of number of attributes ($d$) on $|\mathcal{S}|$}
        \label{fig:DOTKSVM}
    \end{minipage}
    \hspace{3mm}
    \begin{minipage}[t]{0.23\linewidth}
        \centering
        \includegraphics[scale=0.24]{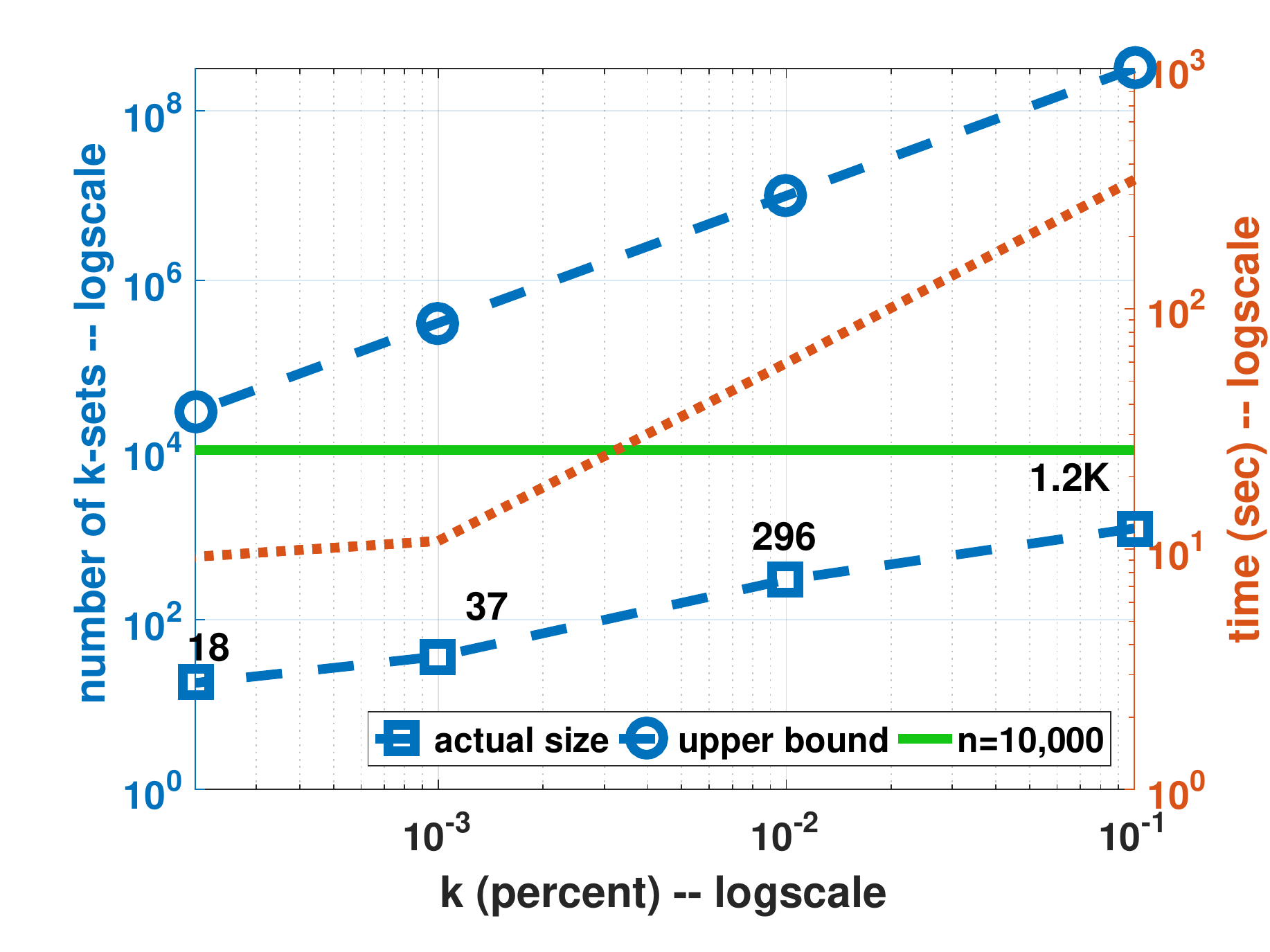}
        \vspace{-8mm}\caption{BN dataset, MD: Impact of $k$ on $|\mathcal{S}|$}
        \label{fig:BNKSVK}
    \end{minipage}
    \hspace{1mm}
    \begin{minipage}[t]{0.23\linewidth}
        \centering
        \includegraphics[scale=0.24]{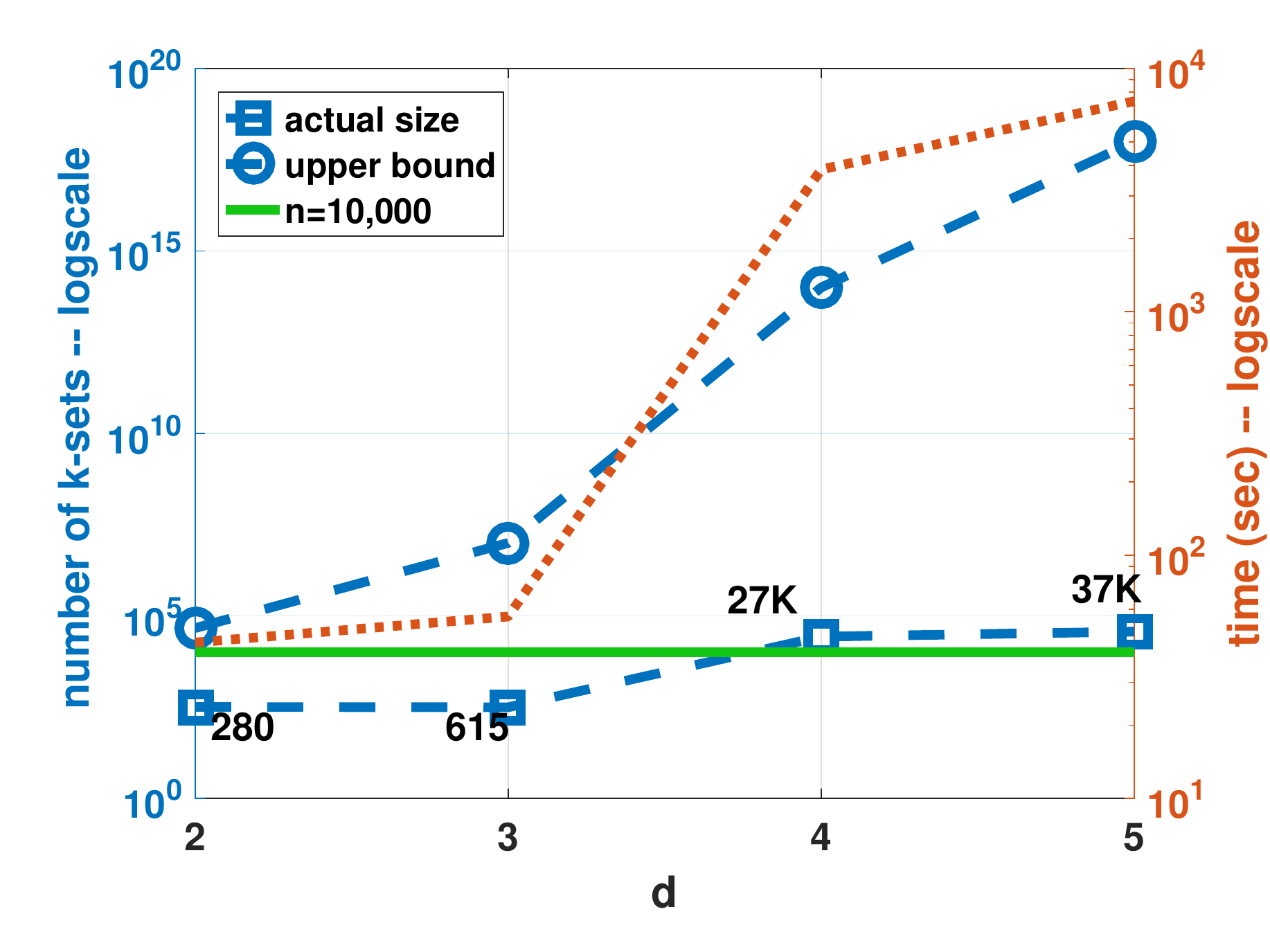}
        \vspace{-8mm}\caption{BN dataset, MD: Impact of number of attributes ($d$) on $|\mathcal{S}|$}
        \label{fig:BNKSVM}
    \end{minipage}
    \hspace{-2mm}
\end{figure*}

\begin{figure*}[ht]
    \begin{minipage}[t]{0.23\linewidth}
        \centering
        \includegraphics[scale=0.24]{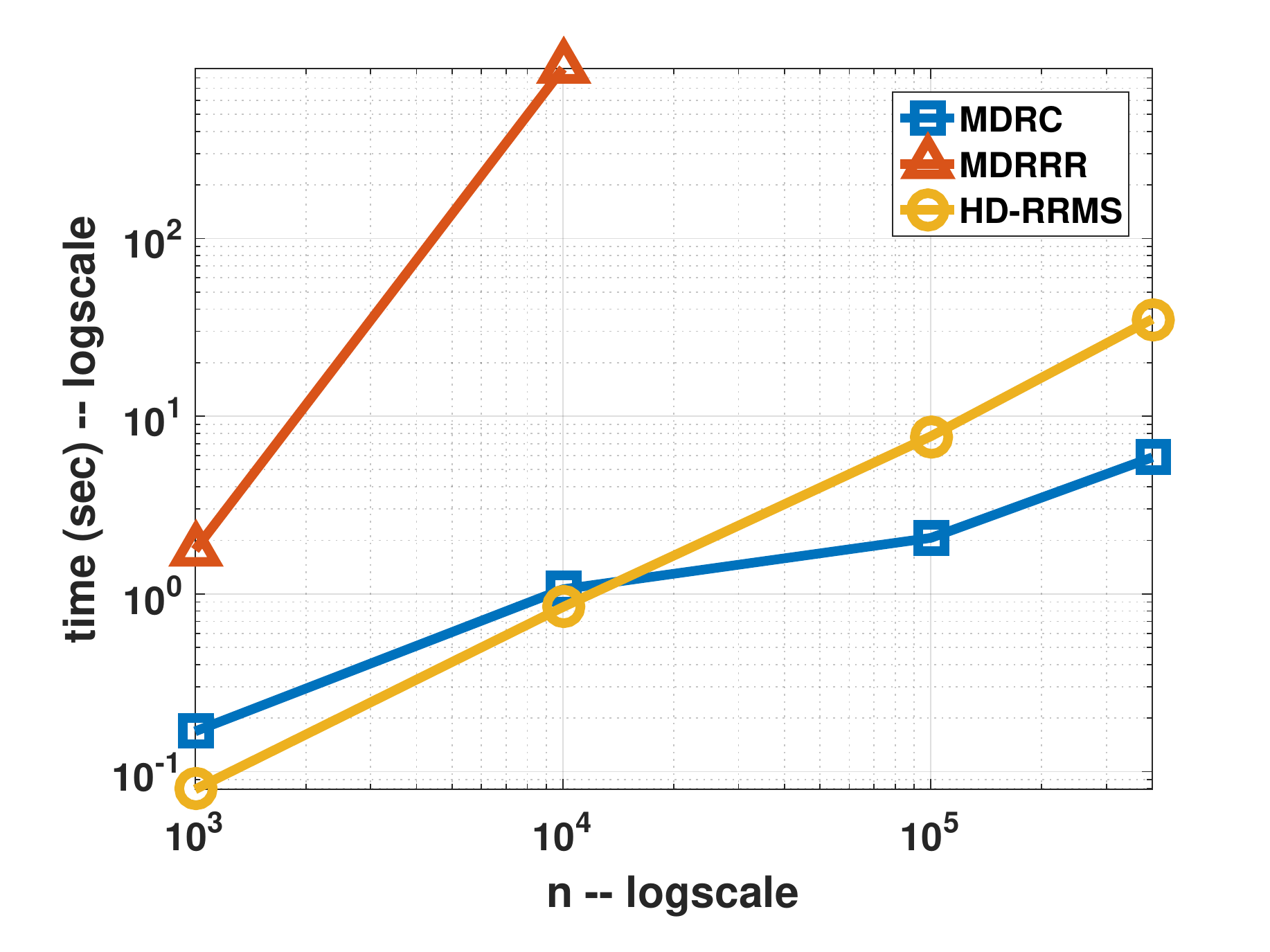}
        \vspace{-8mm}\caption{DOT dataset, MD, Efficiency: Impact of dataset size ($n$)}
        \label{fig:DOTMDVNTime}
    \end{minipage}
    \hspace{0mm}
    \begin{minipage}[t]{0.23\linewidth}
        \centering
        \includegraphics[scale=0.24]{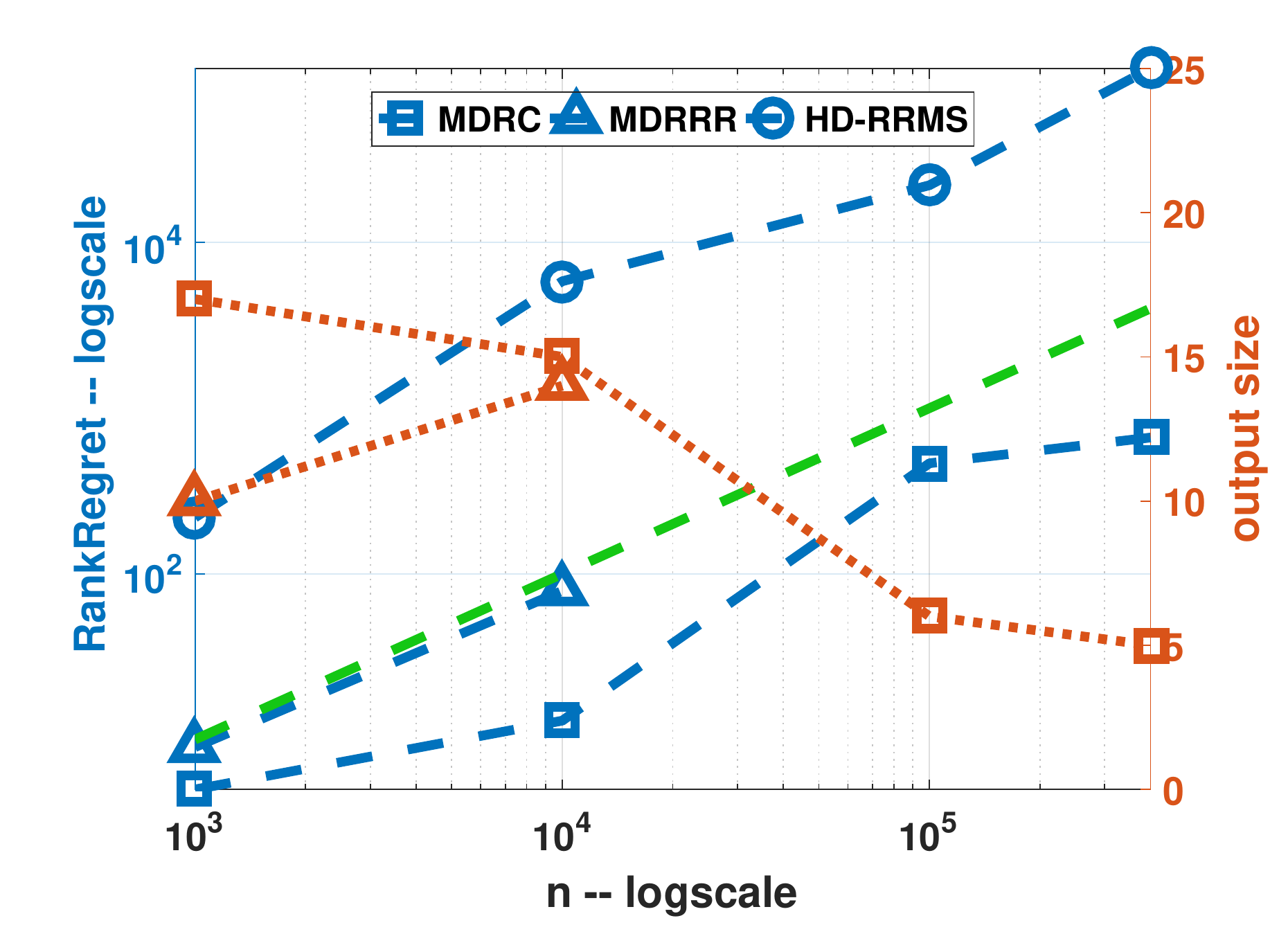}
        \vspace{-8mm}\caption{DOT dataset, MD, Effectiveness: Impact of dataset size ($n$)}
        \label{fig:DOTMDVNSize}
    \end{minipage}
    \hspace{3mm}
    \begin{minipage}[t]{0.23\linewidth}
        \centering
        \includegraphics[scale=0.24]{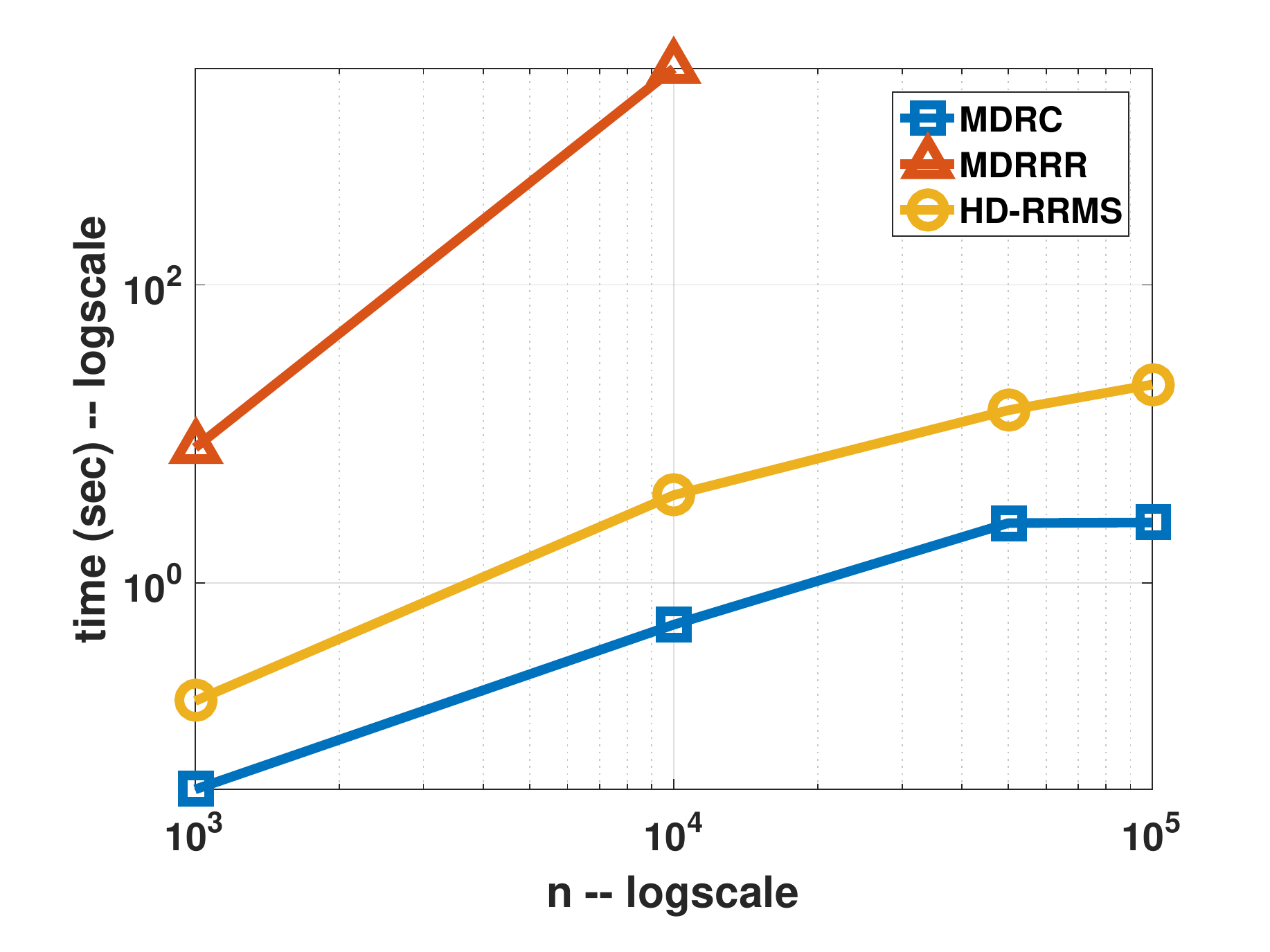}
        \vspace{-8mm}\caption{BN dataset, MD, Efficiency: Impact of dataset size ($n$)}
        \label{fig:BNMDVNTime}
    \end{minipage}
    \hspace{1mm}
    \begin{minipage}[t]{0.23\linewidth}
        \centering
        \includegraphics[scale=0.24]{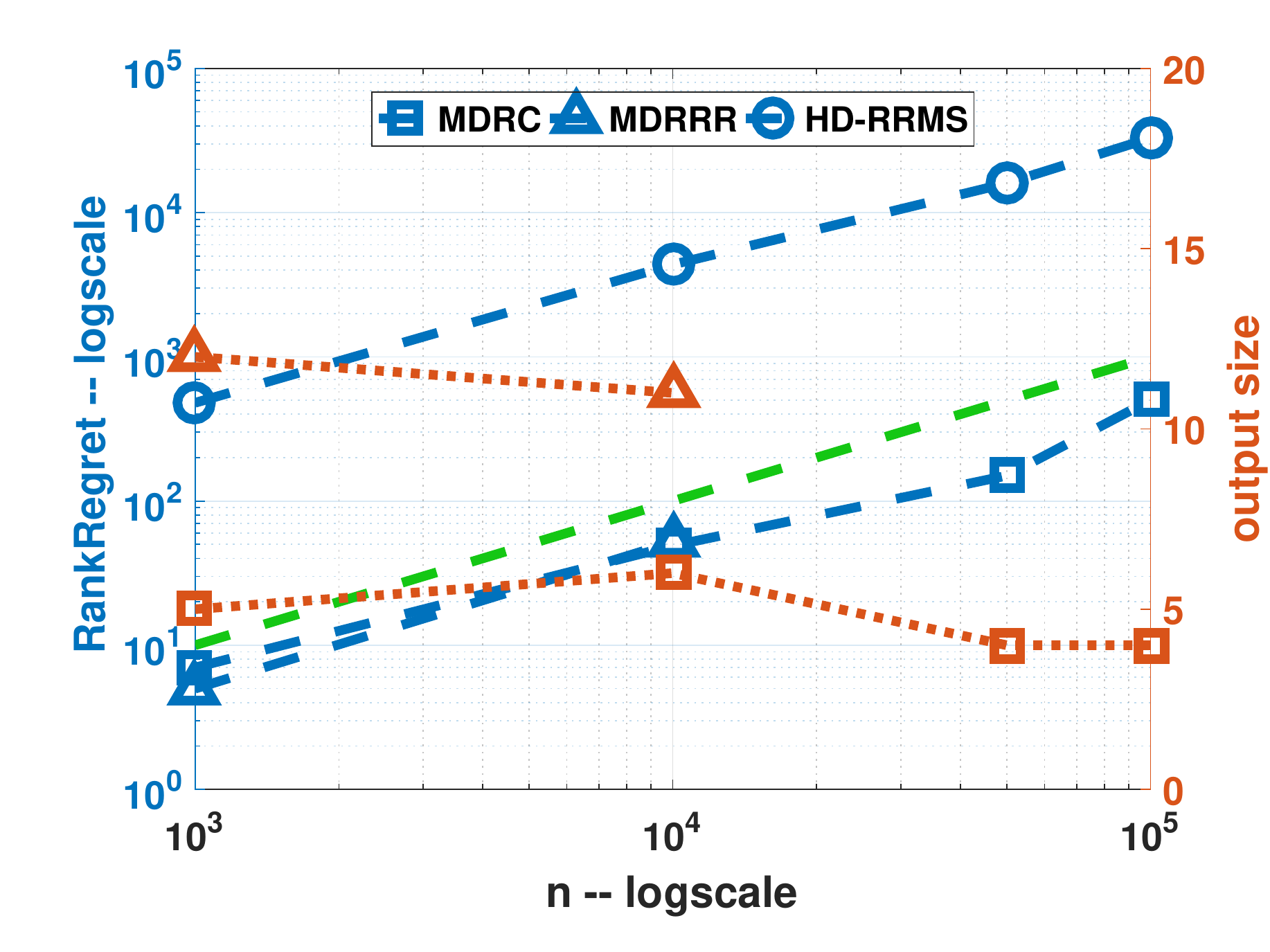}
        \vspace{-8mm}\caption{BN dataset, MD, Effectiveness: Impact of dataset size ($n$)}
        \label{fig:BNMDVNSize}
    \end{minipage}
    \hspace{-2mm}
\end{figure*}

\begin{figure*}[ht]
    \begin{minipage}[t]{0.23\linewidth}
        \centering
        \includegraphics[scale=0.24]{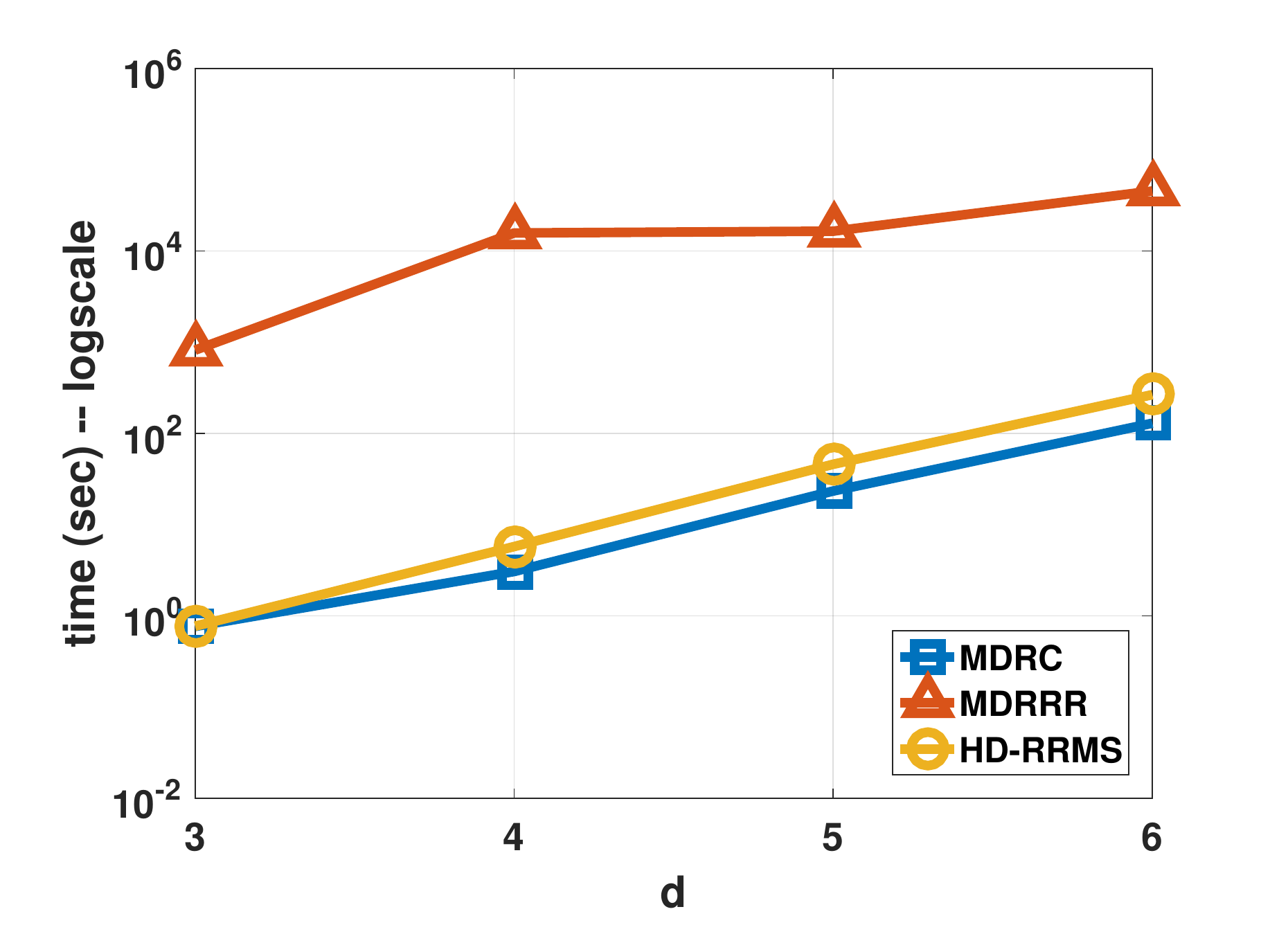}
        \vspace{-8mm}\caption{DOT dataset, MD, Efficiency: Impact of number of attributes ($d$)}
        \label{fig:DOTMDVMTime}
    \end{minipage}
    \hspace{0mm}
    \begin{minipage}[t]{0.23\linewidth}
        \centering
        \includegraphics[scale=0.24]{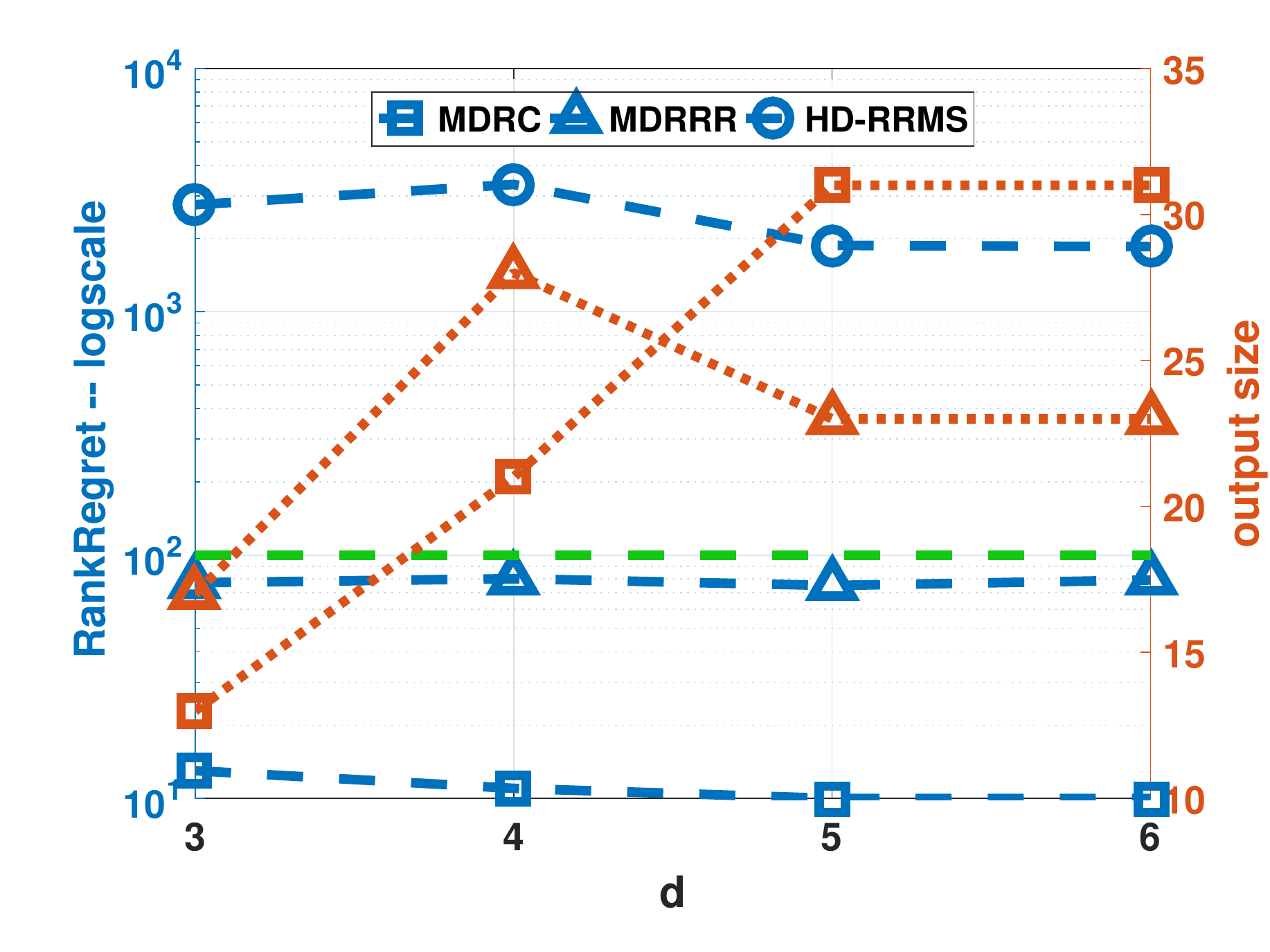}
        \vspace{-8mm}\caption{DOT dataset, MD, Effectiveness: Impact of number of attributes ($d$)}
        \label{fig:DOTMDVMSize}
    \end{minipage}
    \hspace{3mm}
    \begin{minipage}[t]{0.23\linewidth}
        \centering
        \includegraphics[scale=0.24]{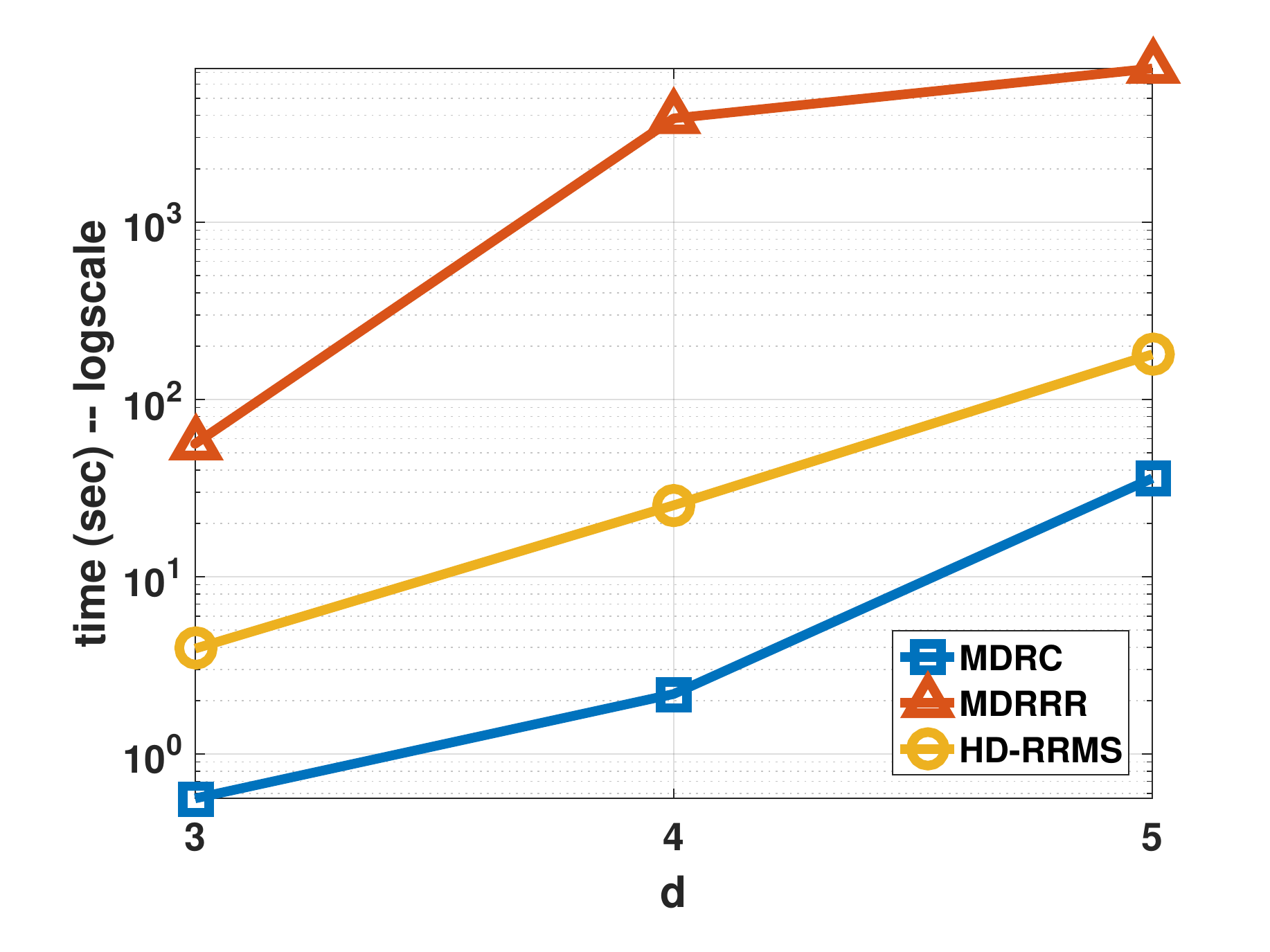}
        \vspace{-8mm}\caption{BN dataset, MD, Efficiency: Impact of number of attributes ($d$)}
        \label{fig:BNMDVMTime}
    \end{minipage}
    \hspace{1mm}
    \begin{minipage}[t]{0.23\linewidth}
        \centering
        \includegraphics[scale=0.24]{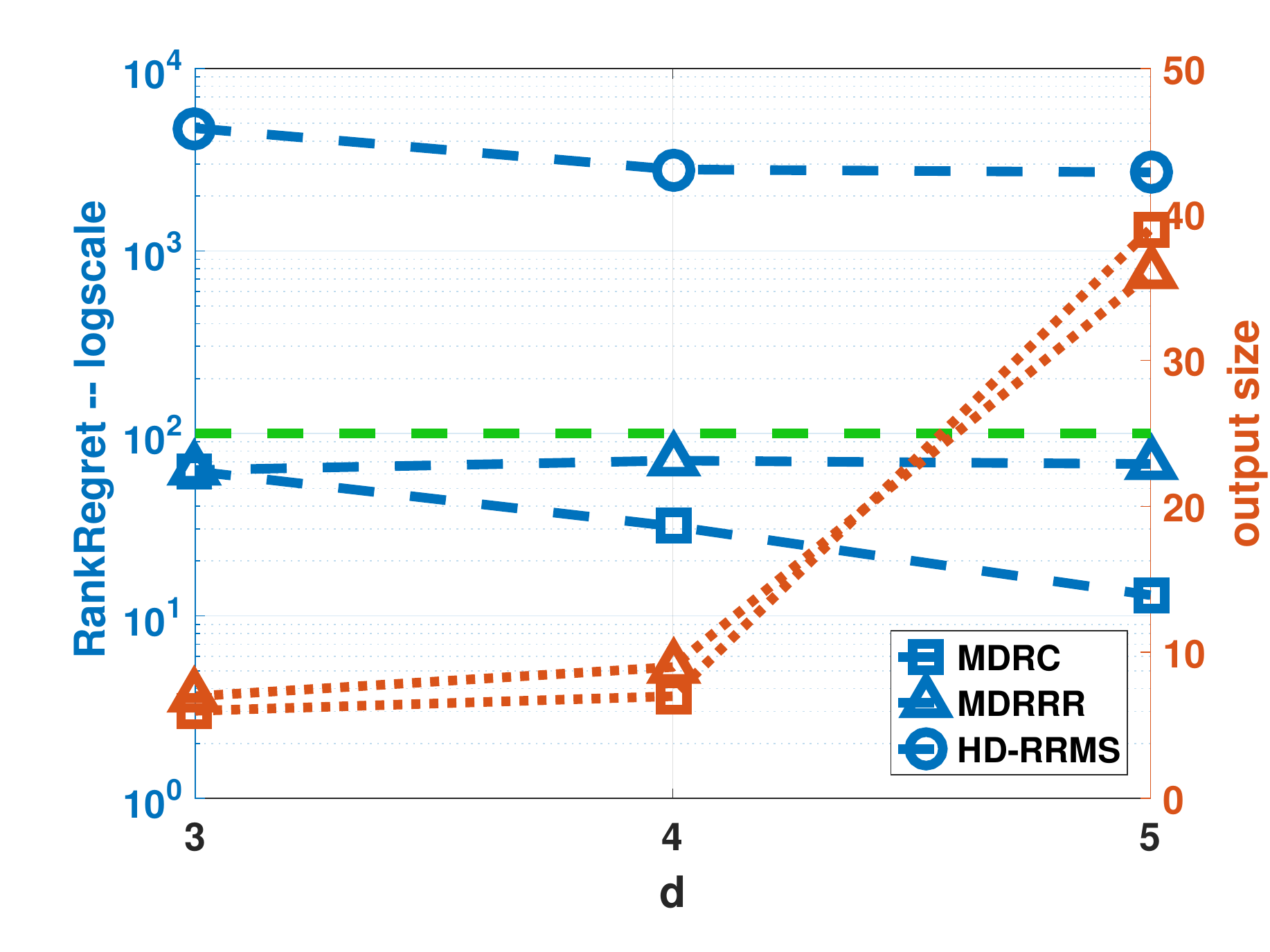}
        \vspace{-8mm}\caption{BN dataset, MD, Effectiveness: Impact of number of attributes ($d$)}
        \label{fig:BNMDVMSize}
    \end{minipage}
    \hspace{-2mm}
\end{figure*}

\begin{figure*}[ht]
    \begin{minipage}[t]{0.23\linewidth}
        \centering
        \includegraphics[scale=0.24]{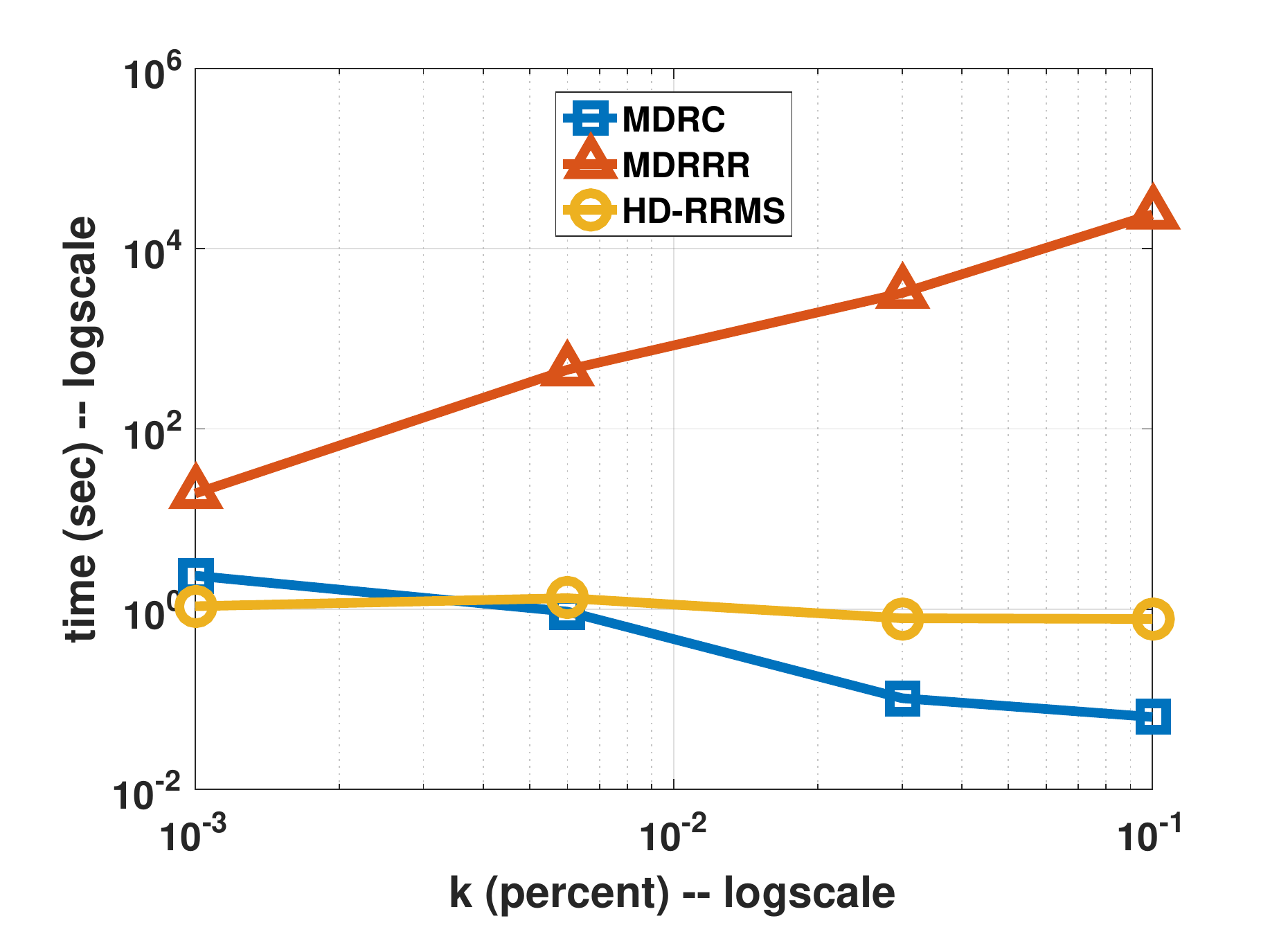}
        \vspace{-8mm}\caption{DOT dataset, MD, Efficiency: Impact of $k$}
        \label{fig:DOTMDVKTime}
    \end{minipage}
    \hspace{0mm}
    \begin{minipage}[t]{0.23\linewidth}
        \centering
        \includegraphics[scale=0.24]{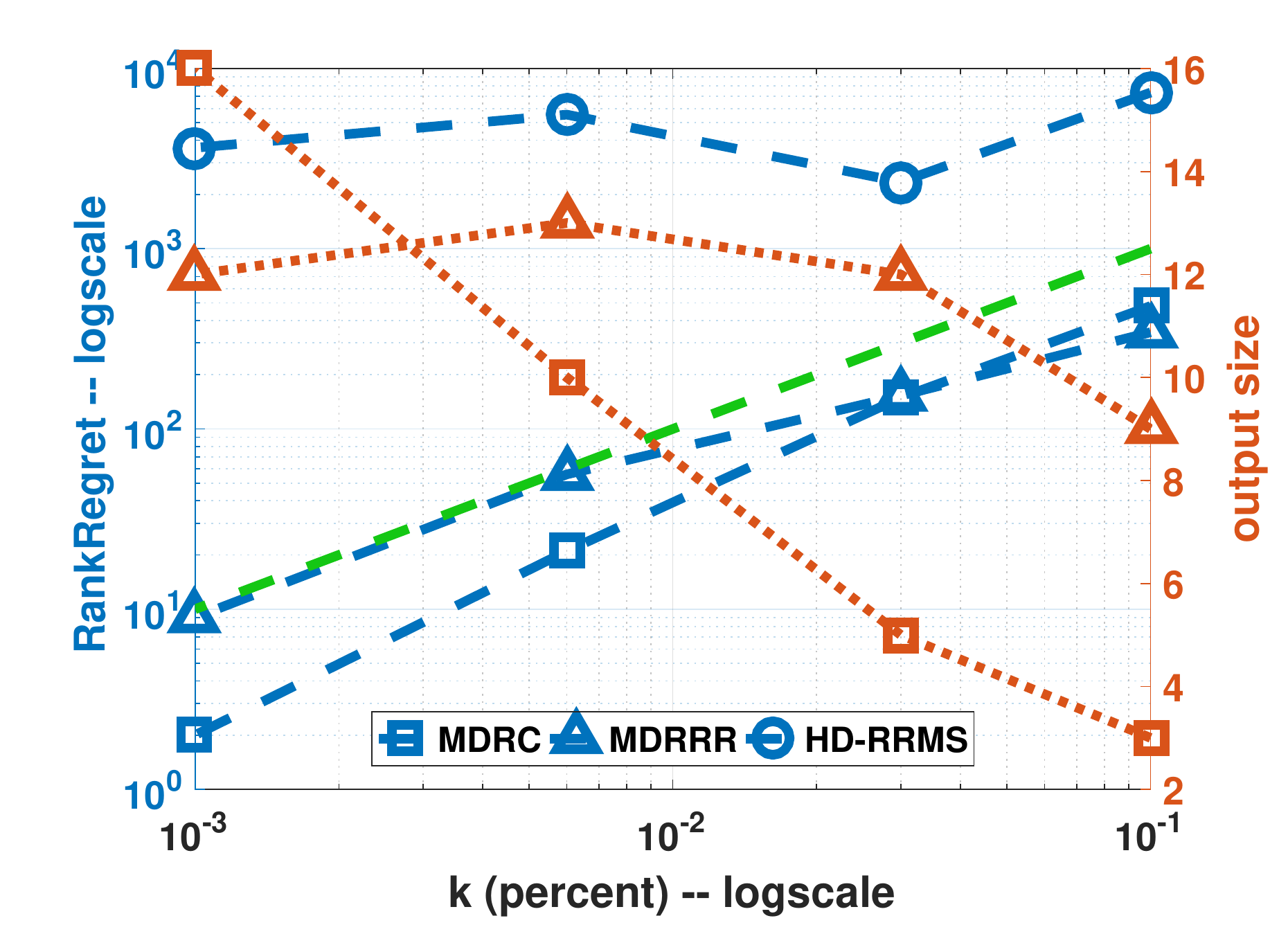}
        \vspace{-8mm}\caption{DOT dataset, MD, Effectiveness: Impact of $k$}
        \label{fig:DOTMDVKSize}
    \end{minipage}
    \hspace{3mm}
    \begin{minipage}[t]{0.23\linewidth}
        \centering
        \includegraphics[scale=0.24]{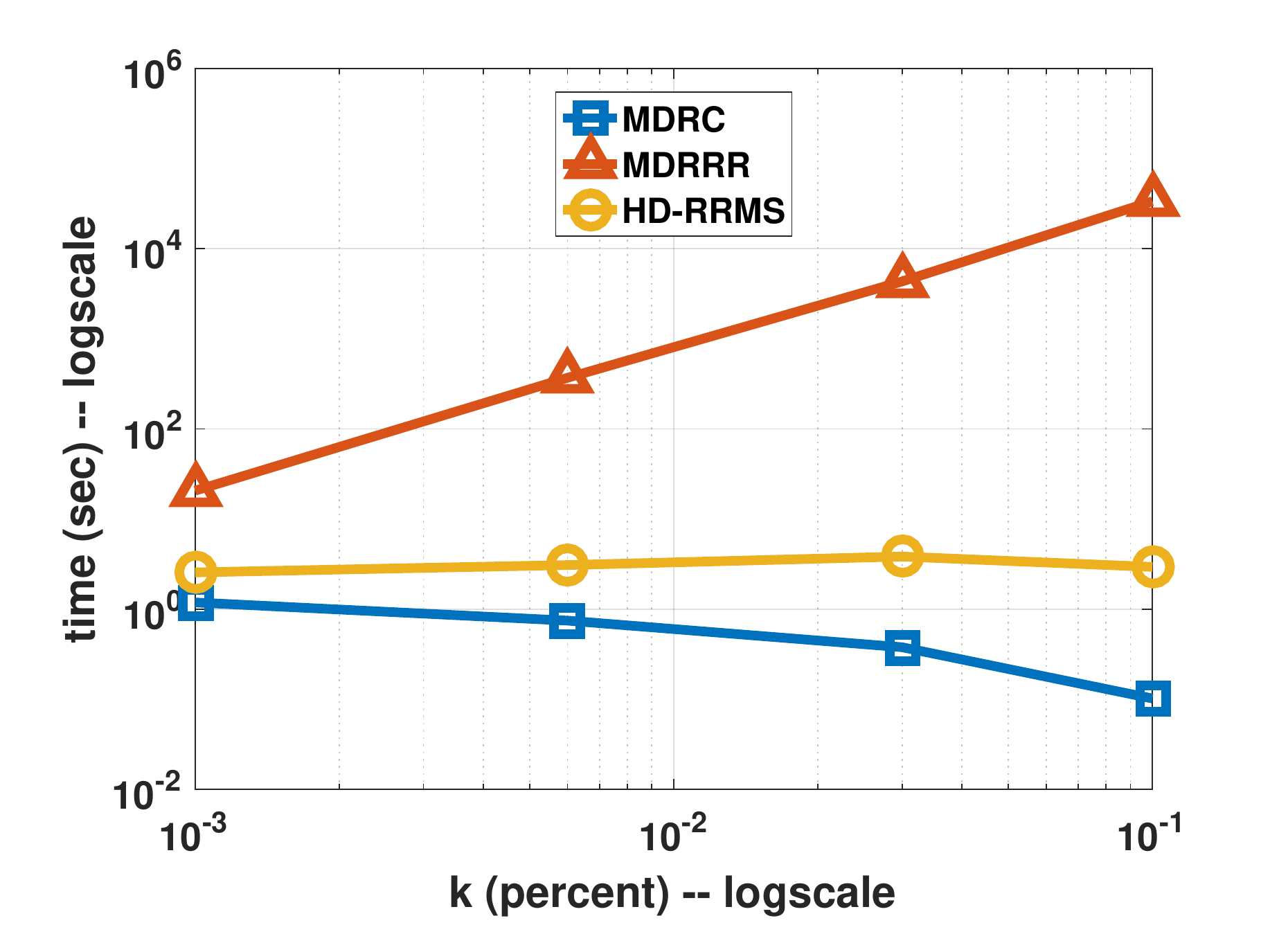}
        \vspace{-8mm}\caption{BN dataset, MD, Efficiency: Impact of $k$}
        \label{fig:BNMDVKTime}
    \end{minipage}
    \hspace{1mm}
    \begin{minipage}[t]{0.23\linewidth}
        \centering
        \includegraphics[scale=0.24]{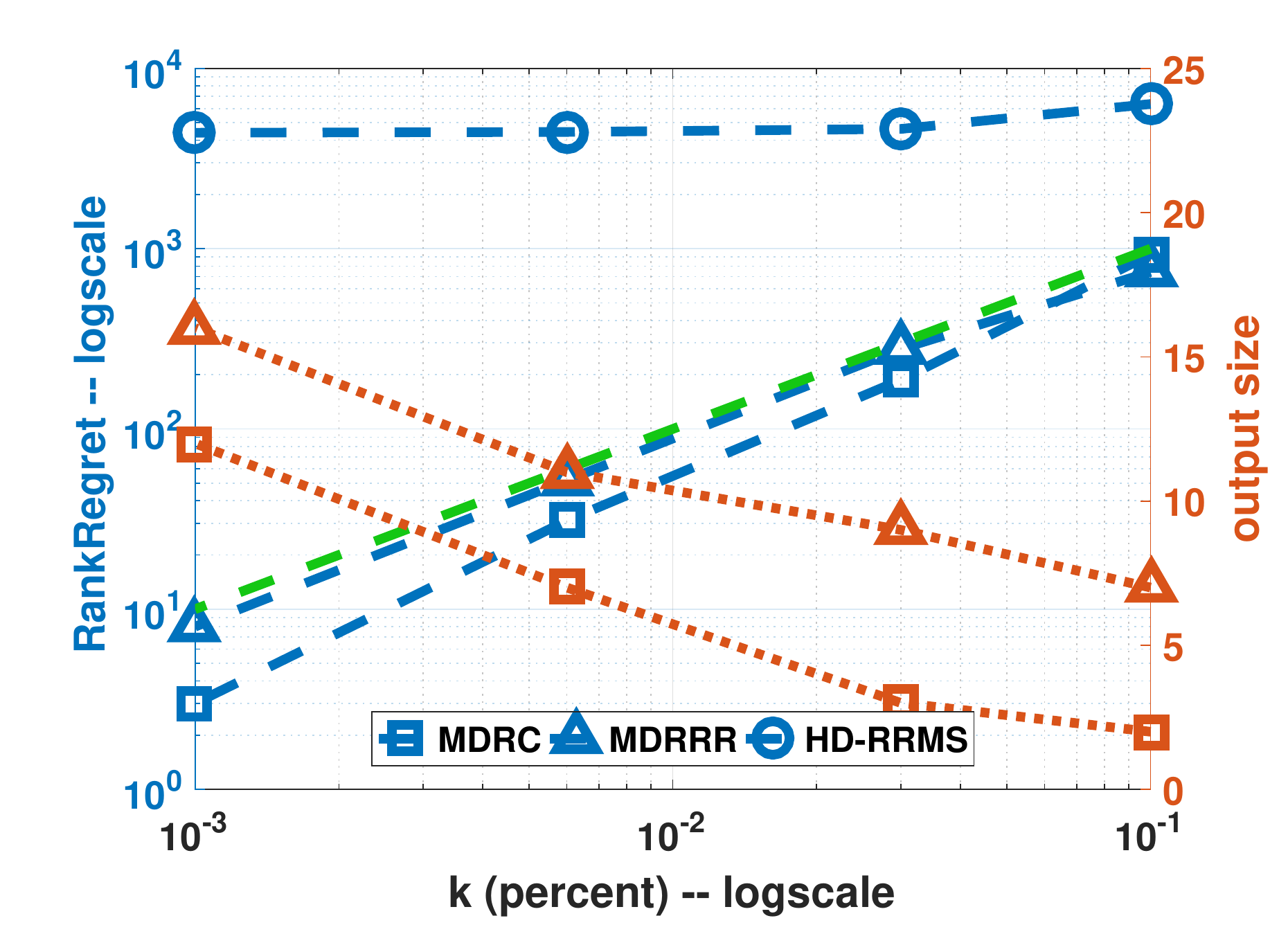}
        \vspace{-8mm}\caption{BN dataset, MD, Effectiveness: Impact of $k$}
        \label{fig:BNMDVKSize}
    \end{minipage}
    \hspace{-2mm}
\end{figure*}

\section{Experimental Evaluation}\label{sec:exp}
\subsection{Setup}

\stitle{Datasets.} 
To evaluate our algorithms to compute RRR, we conducted experiments over two real multi-attribute datasets that could potentially benefit from the user of rank regret.  
We describe these datasets next.

{\it US Department of Transportation flight delay database (DOT)\footnote{\small{\url{www.transtats.bts.gov/DL_SelectFields.asp?}}}:} This database is widely used by third-party websites to identify the on-time performance of flights, routes, airports, and airlines.
After removing the records with missing values, the dataset contains 457,892 records, for all flights conducted by the 14 US carriers in the last months of 2017, over the scalar attributes {\tt Dep-Delay}, {\tt Taxi-Out}, {\tt Actual-elapsed-\\time}, {\tt Arrival-Delay}, {\tt Air-time}, {\tt Distance}, {\tt Taxi-in}, and {\tt CRS-elapsed-time}. For {\tt Air-time} and {\tt Distance} higher values are preferred while for the rest of attributes lower values are better.

{\it Blue Nile (BN)\footnote{\small{\url{www.bluenile.com/diamond-search?}}}:} 
Blue Nile is the largest diamonds online retailer in the world. We collected its catalog that contained 116,300 diamonds at the time of our experiments.
We consider the scalar attribute {\tt Carat}, {\tt Depth}, {\tt LengthWidthRatio}, {\tt Table}, and {\tt Price}. For all attributes, except {\tt Price}, higher values are preferred.
The value of the diamonds highly depend on these measurement, small changes in these scores may mean a lot in terms of the quality of the jewel:
For example, while the listed diamonds range from 0.23 carat to 20.97, minor changes in the carat affects the price. We considered two similar diamonds, where one is 0.5 carat and the other is 0.53 carat. Even though all other measures are similar for both diamonds, the second is 30\% more expensive than the first one.
This is also correct for {\tt Depth}, {\tt LengthWidthRatio}, and {\tt Table}.
{\em Such settings where slight changes in the scores may dramatically affect the value} (and the rank) of the items, highlight the motivation of rank-regret.

We normalize each value $v$ of a higher-preferred attribute $A$ as $(v-\min(A))/(\max(A)-\min(A))$ and for each lower-preferred attribute $A$, we do it as $(\max(A) - v)/(\max(A)-\min(A))$.

\stitle{Algorithms evaluated:}
In addition to the theoretical analyses, we evaluate the algorithms proposed in this paper.
In \S~\ref{sec:2d}, we proposed \twodrrr, the algorithm that uses Theorem~\ref{th:max2k} to transform the problem into one dimensional range covering. This quadratic algorithm guarantees the approximation ratio of 2 on the maximum rank regret of its output.
In this section, we shall show that in all the cases it generated an output with maximum rank of $k$. For 2D, we implemented the ray-sweeping algorithm (similar to Algorithm~\ref{alg:2dfr}) that enumerates the $k$-sets by following the changes in the $k$-border (Figure~\ref{fig:toy2}). We also implemented the $k$-set graph based enumeration \techrep{explained in Appendix~\ref{ap:kset-enum}} for MD. We did not include the results here, but we observed that it does not scale beyond a few hundred items (that is because it need to solve much as $O(nk)$ linear programs for a single $k$-set). Instead, we apply the randomized algorithm \ksetr for finding the $k$-sets (while setting the termination condition $c$ to 100).
The MD algorithms proposed in \S~\ref{sec:md} are the hitting-set based algorithm \mdrrr and the space function covering algorithm \mdrc.
As we explained in \S~\ref{sec:intro} and~\ref{sec:related}, all of the existing algorithms proposed for different varieties of regret-ratio consider the score difference, as the measure of regret and apply the optimization based on it.
Still to verify this, we consider comparing with them as the baseline.
As we shall further explain in \S~\ref{sec:related}, the
advanced algorithms for the regret-ratio problem are two
similar studies~\cite{agarwal2017efficient,asudeh2017} that both work based on discretizing the function space and applying hitting set, and therefore, provide similar controllable additive approximation factors. We adopt the {\sc hd-rrms} algorithm~\cite{asudeh2017} which as mentioned in~\cite{agarwal2017efficient, kumar2018faster} should perform similar to the one in~\cite{agarwal2017efficient}.
Since the input the algorithm is the index size, in order to be fair in the comparison, in all settings, we first run the algorithm \mdrc, and then pass the output size of it as the input to {\sc hd-rrms}.
 
\stitle{Evaluations:}
In addition to the efficiency, we evaluate the effectiveness of the proposed algorithms. That is, we study if the algorithms can find a small subset with bounded rank-regret based on $k$.
We consider the running time as the efficiency measure and the rank-regret of output set, as well as its size, for effectiveness. Computing the exact rank-regret of a set needs the construction of the arrangement of items in the dual space which is not scalable to the large settings. Therefore, in the experiments for estimating the rank-regret of a set in MD,
we draw 10,000 functions uniformly at random (based on Lines 4 to 6 of Algorithm~\ref{alg:kset_baseline}) and consider them for estimating the rank-regret.

\stitle{Default values:} For each experiment, we study the impact of varying one variable while fixing other attributes to their default values. The default values are as following: (i) dataset size ($n$): 10,000, (ii) number of attributes ($d$): 3, and (iii) $k$: top-1\%.

\stitle{Hardware and platform.}
All experiments were performed on a Linux machine with a 3.8 GHz Intel Xeon processor and 64 GB memory.
The algorithms are implemented using Python 2.7.

\subsection{Results}
\stitle{2D.}
We use a ray sweeping algorithm, similar to Algorithm~\ref{alg:2dfr}, to enumerates the $k$-sets by following the changes in the $k$-border. We also use the ray sweeping to find out the (exact) rank regret of a set in 2D.
Due to the space limitations, for 2D, we only provide the plots for the DOT dataset.
Figures~\ref{fig:DOT2DVN1} and~\ref{fig:DOT2DVN2} show the performance of the algorithms for varying the dataset size ($n$) from 1000 to 400,000.
The running times of \twodrrr and \mdrrr are dominated by the time required by the sweeping line algorithms for finding the ranges (Algorithm~\ref{alg:2dfr}) and the $k$-sets. Since these two algorithms have similar structure, their running times are similar.
Still, because the sweeping ray algorithm is quadratic, these algorithms did not scale beyond 100K items.
On the other hand \mdrc does not depend on finding the $k$-set or sweeping a line. Rather, it partitions the space until top-$k$ of two corners of each range intersect. Due to the binary search nature of the algorithm that breaks the space by half at every iteration, soon the functions in the two ends of each range become similar enough to share an item in their top-$k$.
Therefore, the algorithm performs very well in practice, and scales well for large settings. For example, it took less than a second to run \mdrc for 100K items, while \twodrrr and \mdrrr required several thousand seconds. See Figure~\ref{fig:DOT2DVN1}.
In Figure~\ref{fig:DOT2DVN2}, and all other plots with, two y-axes, the left axis show the rank-regret and the right one is the output size.
The dashed green line show the border for the rank-regret of 1\%.

The algorithm \twodrrr guarantees the optimal output size. For all settings its output also had the rank-regret of less than $k$, confirming that it returned the optimal solution. On the other hand, \mdrrr guarantees the rank-regret of $k$ and provides the logarithmic approximation ratio on its output size. This is also confirmed in the figure, where the rank regret of the output of \mdrrr is always below the green line. However, the size of its output is more than the optimal for two (out of three) settings. the space partitioning algorithm \mdrrr provides the output which in all cases satisfied the rank-regret of $k$ and also its output size was the minimum, confirming that it also discovered the optimal output.
In Figures~\ref{fig:DOT2DVK1} and~\ref{fig:DOT2DVK2}, we fix the dataset size and other variables to the default and study the effect of changing $k$ on the efficiency of the algorithm and the quality of their outputs. Similar to Figure~\ref{fig:DOT2DVN1}, \twodrrr and \mdrrr have similar running times (due to applying the ray sweeping algorithm) and \mdrc runs within a few milliseconds for all settings.
On the other hand, in Figure~\ref{fig:DOT2DVK2}, the output size of \mdrc is in all cases, except one, equal to the optimal output size (the output size of \twodrrr) while, due to its logarithmic approximation ratio, the hitting set based \mdrrr generates larger outputs.\mdrrr guarantees the rank-regret of $k$, which is confirmed in the figure. \mdrc also provided the maximum rank-regret of $k$ for all settings and \twodrrr did so for all, except $k=0.004\%$ for which its maximum rank regret was slightly above the threshold.
 
\stitle{k-set size.}
Next, we compare the actual size of $k$-sets with the theoretical upper-bounds, using the \ksetr algorithm.
To do so, we select the DOT and BN datasets, set number of items to 10K and study the impact of varying $k$ and $d$.
The results are provided in Figures~\ref{fig:DOTKSVK},~\ref{fig:DOTKSVM},~\ref{fig:BNKSVK}, and~\ref{fig:BNKSVM}.
The left-y-axis in the figures show the size and the right-y-axis show the running time of the \ksetr algorithm.
The horizontal green line in the figures highlight the number of items $n=10$K.
Figures~\ref{fig:DOTKSVK} and~\ref{fig:BNKSVK} show the results for varying $k$ for DOT and BN, respectively.
First, as observed in the figures, the actual sizes of the $k$-sets are significantly smaller than the best known theoretical upper-bound for 3D ($O(nk^{3/2})$~\cite{sharir2000}). In fact, the number of $k$-sets is closer to $n$ than the upper-bounds.
Second, the number of $k$-sets for $k=10\%$ is significantly larger than the number of $k$-sets for smaller values of $k$.
Recall that the $k$-sets are densely overlapping, as the neighboring $k$-sets in the $k$-set graph only differ in one item.
As $k$ increases (up until $k=50\%$), for each node of the $k$-set graph the number of candidate transitions to the neighboring $k$-sets increases which affect $|\mathcal{S}|$ as well.
Although significantly smaller than the upper bound, still the sizes are large enough to make the $k$-set discovery impractical for large settings. For example, running the \ksetr algorithm for the DOT dataset and $k=10\%$ took more than ten thousand seconds.
The observations for varying $d$ (Figures~\ref{fig:DOTKSVM} and~\ref{fig:BNKSVK}) are also similar.
Also, the gap between the theoretical upper-bound for $d\geq 4$ and the actual $k$-sets sizes show how loose the bounds are.

\stitle{MD.}
Here, we study the algorithms proposed for the general cases where $d\geq 3$. \mdrrr is the hitting set based algorithm that, given the collection of $k$-sets, guarantees the rank-regret of $k$ and a logarithmic increase in the output size. So far, the 2D experiments confirmed these bounds.
The other algorithm is the space partitioning algorithm \mdrc which is designed based on Theorem~\ref{th:max2k}.
Given the possibly large number of $k$-sets and the cost of finding them (even using the randomized algorithm \ksetr), this algorithm is designed to prevent the $k$-set enumeration. \mdrc uses the fact that the $k$-sets are highly overlapping and recursively partitions the space (see Figure~\ref{fig:practicalrunning}) into several hypercubes and stops the recursion for each hypercube as soon as the intersection of the top-$k$ items in its corners is not empty.
This algorithm performs very well in practice, as after a few iterations, the functions in the corners become similar enough to share at least one item in their top-$k$.
Also,
the maximum rank-regret of the item that appear in the top-$k$ of the corners of the hyper-rectangle for the functions inside the hypercube is much smaller than the bound provided in Theorem~\ref{th:mdrc2}. We so far observed it in the 2D experiments where in all cases the rank-regret of the output of \mdrc is less than $k$, while the output size also was always close to the optimal output size.

In addition to these algorithms, we compare the efficiency and effectiveness of our algorithms against, \hdrrms~\cite{asudeh2017}, the recent approximation algorithm proposed for regret-ratio minimizing problem. 
Since \hdrrms takes the index size as the input, we first run the \mdrc algorithm and pass its output size to \hdrrms.
Having a different optimization objective (on the regret-ratio), as we shall show, the output of \hdrrms fails to provide a bound on the rank-regret.
In the first experiment, fixing the other variables to their default values, we vary the dataset size $n$ from 1000 to 400,000 for DOT and from 1000 to 100,000 for BN.
Figures~\ref{fig:DOTMDVNTime},~\ref{fig:DOTMDVNSize},~\ref{fig:BNMDVNTime}, and~\ref{fig:BNMDVNSize} show the results.
Figures~\ref{fig:DOTMDVNTime} and~\ref{fig:DOTMDVNSize},~\ref{fig:BNMDVNTime} show the running time of the algorithms for DOT and BN, respectively.
Looking at these figures, first \mdrrr did not scale for 100K items. The reason is that \mdrrr needs the collection of $k$-sets in order to apply the hitting set. For a very large number of items even the \ksetr algorithm does not scale.
\hdrrms has a reasonable running time in all cases.
\mdrc has the least running time for large values of $n$ and in all cases it finished in less than a few seconds.
The reason is that after a few recursions, the functions in the corners of the hypercubes become similar and share an item in their top-$k$.
Figures~\ref{fig:DOTMDVNSize} and~\ref{fig:BNMDVNSize} show the effectiveness of the algorithms for these settings.
The left-y-axes show the maximum rank-regret of an output set while the right-y-axes show the output size.
The green lines show the rank-regret of $k$ border.
First, the output size for all settings is less than 20 items, which confirm the effectiveness of algorithms for finding a rank-regret representative.
As explained in \S~\ref{subsec:mdapprox}, \mdrrr guarantees the rank-regret of $k$, which is observed here as well.
As expected, \hdrrms fails to provide a rank-regret representative in all cases.
Both for DOT and BN, the maximum rank-regret of the output of \hdrrms are close to $n$, the maximum possible rank-regret.
For example, for DOT and $n=$400K, the rank-regret of \hdrrms was 112K, i.e., there exists a function for which the top-$1$ based on the output of \hdrrms has the rank 112,000.
Based on Theorem~\ref{th:mdrc2}, for these settings, the rank-regret of the output of \mdrc is guaranteed to be less than $4k$ for all cases. However, in practice we expect the rank-regret to be smaller than this. This is confirmed in both experiments for DOT (Figure~\ref{fig:DOTMDVNSize}) and BN (Figure~\ref{fig:BNMDVNSize}) where the output of \mdrc provided the rank-regret of $k$.

Next, we evaluate the impact of varying the number of dimensions. Setting $n$ to 10,000 and $k$ to $1\%$ of $n$ (i.e. 100), we the number of attributes, $d$, from $3$ to $6$ for DOT and from $3$ to $5$ for BN.
Figures~\ref{fig:DOTMDVMTime},~\ref{fig:DOTMDVMSize},~\ref{fig:BNMDVMTime}, and~\ref{fig:DOTMDVMSize} show the results.
The running times of the algorithms for DOT and BN are provided in Figures~\ref{fig:DOTMDVMTime} and~\ref{fig:BNMDVMTime}. 
Similar to the previous experiments, 
since the hitting set based algorithm \mdrrr requires the collection of $k$-sets, it was not efficient.
Both \hdrrms and \mdrc performed well in both experiments.
On the other hand, looking at Figures~\ref{fig:DOTMDVMSize} and~\ref{fig:DOTMDVMSize} \hdrrms fails to provide a rank-regret representative, as in all settings there the rank-regret of  its output was several thousands, while the maximum possible rank-regret is $n=10,000$.
The outputs of proposed algorithms in \S~\ref{sec:md}, as expected, satisfied the requested rank-regret.
Interestingly, the output of \mdrc had a lower rank-regret, especially for DOT where its rank-regret was around 10 for all settings.
The output of both \mdrrr and \mdrc was less than 40, for all settings and both datasets, which confirm the effectiveness of them as the representative. 

In the last experiment, we evaluate the impact of varying $k$.
For both datasets, while setting $n$ to 10,000 and $d$ to 3, we varied $k$ from 0.1\% of items (i.e., 10) to 10\% (i.e., 1000).
Figures~\ref{fig:DOTMDVKTime},~\ref{fig:DOTMDVKSize},~\ref{fig:BNMDVKTime}, and~\ref{fig:BNMDVKSize} show the results.
Looking at Figures~\ref{fig:DOTMDVKTime} and~\ref{fig:BNMDVKTime} which show the running time of the algorithms for DOT and BN, respectively, \mdrrr had the worst performance, and it got worse as $k$ increased. 
The bottleneck in \mdrrr is the $k$-set enumeration, and (looking at Figures~\ref{fig:DOTKSVK} and~\ref{fig:BNKSVK}) it increased by $k$, as the number of $k$-sets increased.
Both \hdrrms and \mdrc were efficient for all settings.
One interesting fact in these plots is that the running time of \mdrc decreases as $k$ increases. This is despite the fact that, as showed in Figures~\ref{fig:DOTKSVK} and~\ref{fig:BNKSVK}, the number of $k$-sets increased.
The reason for the decrease, however, is simple. The probability that the top-$k$ of corners of a hypercube share an item increases when looking at larger values of $k$ where each set contains more items.
Although \hdrrms was efficient in all settings, similar to the previous experiments it fails to provide a rank-regret representative as the rank-regret of its output is not bounded.
The outputs of \mdrrr and \mdrc, on the other hand, had smaller rank-regret than the requested  $k$ in all settings for both datasets. Again, the output sizes in all settings were less than 20, which confirm the effectiveness of them as the rank-regret representative.


\stitle{Summary of results:}
To summarize, the experiments verified the effectiveness and efficiency of our proposal.
While the adaptation of the regret-ratio based algorithm \hdrrms fails to provide a rank-regret representative, \twodrrr, \mdrrr, and \mdrc found small sets with small rank-regrets.
Although the rank-regret of the outputs of \twodrrr and \mdrc can be larger than $k$, in our experiments and our measurements those were always below $k$.
\mdrrr provided small outputs that as expected, always guarantees the rank-regret of $k$.
Interestingly, the output size of \mdrc was around the size of the one by \mdrrr, which verifies the effect of the greedy behavior of \mdrc.
The output sizes in all the experiments were less than 40, confirming the effectiveness of the representatives.
The quadratic \twodrrr and the hitting-set based algorithm \mdrrr scaled up to a limit, whereas \mdrc had low running time at all scales.

\section{Related Work}\label{sec:related}
The problem of finding preferred items of a dataset has been extensively investigated in recent years, and research has spanned multiple directions, most notably in top-$k$ query processing~\cite{ihab} and skyline discovery~\cite{skylineoperator}.
In top-$k$ query processing, the approach is to model the user preferences by a ranking/utility function which is then used to preferentially select  tuples. Fundamental results include access-based algorithms~\cite{fagin2003Optimal, fagin2003comparing, bruno2002, marian2004} and view-based algorithms~\cite{PREFER, das2006views}.
In skyline research, the approach is to compute subsets of the data (such as skylines and convex hulls) that serve as the data representatives in the absence of  explicit preference functions~\cite{skylineoperator,pareto, farhad}. Skylines and convex hulls can also serve as effective indexes for top-$k$ query processing~\cite{chang2000onion, xin, asudeh2016discovering}.

Efficiency and effectiveness have always been the challenges in the above studies. While top-$k$ algorithms depend on the existence of a preference function and may require a complete pass over all of the data before answering a single query, representatives such as skylines may become overwhelmingly large and ineffective in practice~\cite{asudeh2017,har2011expected}.
Studies such as~\cite{chan2006, vlachou2010ranking} are focused towards reducing the skyline size.
In an elegant effort towards finding a small representative subset of the data, Nanongkai et al.~\cite{nanongkai2010} introduced the regret-ratio minimizing representative.
The intuition is that a ``close-to-top'' result may satisfy the users' need. Therefore, for a subset of data and a preference function,
they consider the score difference between the top result of the subset versus the actual top result as the measure of regret, and seek the subset that minimizes its maximum regret over all possible linear functions. Since then, 
works such as~\cite{nanongkai2012interactive, zeighami2016minimizing,asudeh2017,kessler2015k,peng2014geometry,agarwal2017efficient,cao2017k,kumar2018faster} studied different challenges and variations of the problem. 
Chester et al.~\cite{chester2014}
generalize the regret-ratio notion to $k$-regret, in which the regret is considered to be the difference between the top result of the subset and the actual top-$k$ result (instead of the top-1 result). They also prove that the problem is NP-complete for variable number of dimensions.
\cite{cao2017k, agarwal2017efficient} prove that the $k$-regret problem is NP-complete even when $d=3$, using the polytope vertex cover problem~\cite{das1997complexity} for the reduction. As explained in \S~\ref{sec:pre}, this also proves that our problem is NP-complete for $d\geq 3$.
For the case of two dimensional databases, 
\cite{chester2014} proposes a quadratic algorithm and \cite{asudeh2017} improves the running time to $O(n\log n)$.
The cube algorithm and a greedy heuristic~\cite{nanongkai2010} are the first algorithms proposed for regret-ratio in MD.
Recently,~\cite{agarwal2017efficient,asudeh2017} independently propose similar approximation algorithms for the problem, both discretizing the function space and applying the hitting set, thus, providing similar controllable additive approximation factors. The major difference is that \cite{asudeh2017} considers the original regret-ratio problem while \cite{agarwal2017efficient} considers the $k$-regret variation.

It is important to note that the above prior works consider the score difference as the regret measure, making their problem setting different from ours, since we use the rank difference as the regret measure.

The geometric notions used in this paper, such as arrangement, dual space, and $k$-set, are explained in detail in~\cite{Ed87}.
Finding bounds on the number of $k$-sets of a point set do not lead to promising results on the upper bound of the size of $\mathcal{S}$.
Lovasz and Erdos~\cite{lovasz1971, erdos1973} initiated the study of $k$-set notion and provided an upper bound on the maximum number of $k$-sets in $\mathbb{R}^2$. The problem in $\mathbb{R}^2$ has also been studied in~\cite{EW85 , Tot01 , edelsbrunner1989circles , PSS92}.
The best known upper bound on the number of $k$-sets in $\mathbb{R}^2$ and $\mathbb{R}^3$ are $O(n k ^{1/3})$~\cite{dey1998} and $O(n k^{3/2})$~\cite{sharir2000}, respectively.
For higher dimensions, finding an upper bound on the number of $k$-sets has been extensively studied~\cite{Ed87, Tot01, De94, ABFK92}; the best known upper bound is $O(n^{d-\varepsilon})$~\cite{ABFK92}, where $\varepsilon>0$ is a small constant. The problem of enumerating all
$k$-sets has been studied in ~\cite{edelsbrunner1986constructing, chan1999remarks} for 2D and~\cite{andrzejak1999optimization} for MD.
\section{Final Remarks}\label{sec:conclusion}
In this paper, we proposed a rank-regret measure that is easier for users to understand, and often more appropriate, than regret computed from score values.  
We defined {\em rank-regret representative} as the minimal subset of the data containing at least one of the top-$k$ of any possible ranking function.
Using a geometric interpretation of items, we bound the maximum rank of items on ranges of functions and utilized combinatorial geometry notions for designing effective and efficient approximation algorithms for the problem.
In addition to theoretical analyses, we conducted empirical experiments on real datasets that verified the validity of our proposal. Among the proposed algorithms, \mdrc seems to be scalable in practice: in all experiments, within a few seconds, it could find a small subset with small rank-regret.
\pagebreak
\bibliographystyle{abbrv}
\bibliography{ref}

\begin{thebibliography}{10}

\bibitem{agarwal2017efficient}
P.~K. Agarwal, N.~Kumar, S.~Sintos, and S.~Suri.
\newblock Efficient algorithms for k-regret minimizing sets.
\newblock {\em LIPIcs}, 2017.

\bibitem{ABFK92}
N.~Alon, I.~B{\'a}r{\'a}ny, Z.~F{\"u}redi, and D.~J. Kleitman.
\newblock Point selections and weak $\varepsilon$-nets for convex hulls.
\newblock {\em Combinatorics, Probability and Computing}, 1(03):189--200, 1992.

\bibitem{andrzejak1999optimization}
A.~Andrzejak and K.~Fukuda.
\newblock Optimization over k-set polytopes and efficient k-set enumeration.
\newblock In {\em WADS}, 1999.

\bibitem{VCofHalfSpace}
P.~Assouad.
\newblock Densit\'e et dimension.
\newblock {\em Ann. Institut Fourier (Grenoble)}, 1983.

\bibitem{asudeh2017}
A.~Asudeh, A.~Nazi, N.~Zhang, and G.~Das.
\newblock Efficient computation of regret-ratio minimizing set: A compact
  maxima representative.
\newblock In {\em SIGMOD}. ACM, 2017.

\bibitem{asudeh2016discovering}
A.~Asudeh, S.~Thirumuruganathan, N.~Zhang, and G.~Das.
\newblock Discovering the skyline of web databases.
\newblock {\em VLDB}, 2016.

\bibitem{pareto}
A.~Asudeh, G.~Zhang, N.~Hassan, C.~Li, and G.~V. Zaruba.
\newblock Crowdsourcing pareto-optimal object finding by pairwise comparisons.
\newblock In {\em CIKM}, 2015.

\bibitem{QueryReranking}
A.~Asudeh, N.~Zhang, and G.~Das.
\newblock Query reranking as a service.
\newblock {\em VLDB}, 9(11), 2016.

\bibitem{skylineoperator}
S.~Borzsony, D.~Kossmann, and K.~Stocker.
\newblock The skyline operator.
\newblock In {\em ICDE}, 2001.

\bibitem{bronnimann1995GHS}
H.~Br{\"o}nnimann and M.~T. Goodrich.
\newblock Almost optimal set covers in finite vc-dimension.
\newblock {\em DCG}, 14(4):463--479, 1995.

\bibitem{bruno2002}
N.~Bruno, S.~Chaudhuri, and L.~Gravano.
\newblock Top-k selection queries over relational databases: Mapping strategies
  and performance evaluation.
\newblock {\em TODS}, 2002.

\bibitem{cao2017k}
W.~Cao, J.~Li, H.~Wang, K.~Wang, R.~Wang, R.~Chi-Wing~Wong, and W.~Zhan.
\newblock k-regret minimizing set: Efficient algorithms and hardness.
\newblock In {\em LIPIcs}, 2017.

\bibitem{chan2006}
C.-Y. Chan, H.~Jagadish, K.-L. Tan, A.~K. Tung, and Z.~Zhang.
\newblock Finding k-dominant skylines in high dimensional space.
\newblock In {\em SIGMOD}, 2006.

\bibitem{chan1999remarks}
T.~M. Chan.
\newblock Remarks on k-level algorithms in the plane.
\newblock {\em Manuscript, Department of Computer Science, University of
  Waterloo, Waterloo, Canada}, 1999.

\bibitem{chang2000onion}
Y.-C. Chang, L.~Bergman, V.~Castelli, C.-S. Li, M.-L. Lo, and J.~R. Smith.
\newblock The onion technique: indexing for linear optimization queries.
\newblock In {\em SIGMOD}, 2000.

\bibitem{halfspace1}
B.~Chazelle and F.~P. Preparata.
\newblock Halfspace range search: an algorithmic application of k-sets.
\newblock In {\em SOCG}. ACM, 1985.

\bibitem{chester2014}
S.~Chester, A.~Thomo, S.~Venkatesh, and S.~Whitesides.
\newblock Computing k-regret minimizing sets.
\newblock {\em VLDB}, 7(5), 2014.

\bibitem{halfspace2}
K.~L. Clarkson.
\newblock Applications of random sampling in computational geometry, ii.
\newblock In {\em SOCG}. ACM, 1988.

\bibitem{dantzig1998linear}
G.~B. Dantzig.
\newblock {\em Linear programming and extensions}.
\newblock Princeton university press, 1998.

\bibitem{das1997complexity}
G.~Das and M.~T. Goodrich.
\newblock On the complexity of optimization problems for 3-dimensional convex
  polyhedra and decision trees.
\newblock {\em Computational Geometry}, 8(3), 1997.

\bibitem{das2006views}
G.~Das, D.~Gunopulos, N.~Koudas, and D.~Tsirogiannis.
\newblock Answering top-k queries using views.
\newblock In {\em VLDB}, 2006.

\bibitem{dey1998}
T.~K. Dey.
\newblock Improved bounds for planar k-sets and related problems.
\newblock {\em DCG}, 19(3):373--382, 1998.

\bibitem{De94}
T.~K. Dey and H.~Edelsbrunner.
\newblock Counting triangle crossings and halving planes.
\newblock In {\em SOCG}, pages 270--273. ACM, 1993.

\bibitem{Ed87}
H.~Edelsbrunner.
\newblock {\em Algorithms in combinatorial geometry}, volume~10.
\newblock Springer Science \& Business Media, 1987.

\bibitem{edelsbrunner1989circles}
H.~Edelsbrunner, N.~Hasan, R.~Seidel, and X.~J. Shen.
\newblock Circles through two points that always enclose many points.
\newblock {\em Geometriae Dedicata}, 32(1):1--12, 1989.

\bibitem{EW85}
H.~Edelsbrunner and E.~Welzl.
\newblock On the number of line separations of a finite set in the plane.
\newblock {\em Journal of Combinatorial Theory, Series A}, 38, 1985.

\bibitem{edelsbrunner1986constructing}
H.~Edelsbrunner and E.~Welzl.
\newblock Constructing belts in two-dimensional arrangements with applications.
\newblock {\em SICOMP}, 15(1):271--284, 1986.

\bibitem{couponcollector}
P.~Erd{\H{o}}s.
\newblock On a classical problem of probability theory.
\newblock {\em Magy. Tud. Akad. Mat. Kut. Int. K{\H{o}}z.}, 6(1-2), 1961.

\bibitem{erdos1973}
P.~Erd{\H{o}}s, L.~Lov{\'a}sz, A.~Simmons, and E.~G. Straus.
\newblock Dissection graphs of planar point sets.
\newblock {\em A survey of combinatorial theory}, pages 139--149, 1973.

\bibitem{fagin2003comparing}
R.~Fagin, R.~Kumar, and D.~Sivakumar.
\newblock Comparing top k lists.
\newblock {\em Journal on Discrete Mathematics}, 2003.

\bibitem{fagin2003Optimal}
R.~Fagin, A.~Lotem, and M.~Naor.
\newblock Optimal aggregation algorithms for middleware.
\newblock {\em JCSS}, 2003.

\bibitem{finkel1974quad}
R.~A. Finkel and J.~L. Bentley.
\newblock Quad trees a data structure for retrieval on composite keys.
\newblock {\em Acta informatica}, 1974.

\bibitem{qr2}
Y.~D. Gunasekaran, A.~Asudeh, S.~Hasani, N.~Zhang, A.~Jaoua, and G.~Das.
\newblock Qr2: A third-party query reranking service over web databases.
\newblock In {\em ICDE Demo}, 2018.

\bibitem{har2011expected}
S.~Har-Peled.
\newblock On the expected complexity of random convex hulls.
\newblock {\em arXiv preprint arXiv:1111.5340}, 2011.

\bibitem{haussler1987epsilon}
D.~Haussler and E.~Welzl.
\newblock ɛ-nets and simplex range queries.
\newblock {\em DCG}, 2(2):127--151, 1987.

\bibitem{PREFER}
V.~Hristidis and Y.~Papakonstantinou.
\newblock Algorithms and applications for answering ranked queries using ranked
  views.
\newblock {\em VLDB}, 13(1), 2004.

\bibitem{ihab}
I.~F. Ilyas, G.~Beskales, and M.~A. Soliman.
\newblock A survey of top-k query processing techniques in relational database
  systems.
\newblock {\em CSUR}, 40(4):11, 2008.

\bibitem{karp1972reducibility}
R.~M. Karp.
\newblock Reducibility among combinatorial problems.
\newblock In {\em Complexity of computer computations}, pages 85--103.
  Springer, 1972.

\bibitem{kessler2015k}
T.~Kessler~Faulkner, W.~Brackenbury, and A.~Lall.
\newblock k-regret queries with nonlinear utilities.
\newblock {\em VLDB}, 8(13), 2015.

\bibitem{kumar2018faster}
N.~Kumar and S.~Sintos.
\newblock Faster approximation algorithm for the k-regret minimizing set and
  related problems.
\newblock In {\em ALENEX}. SIAM, 2018.

\bibitem{lovasz1971}
L.~Lov{\'a}sz.
\newblock On the number of halving lines.
\newblock {\em Ann. Univ. Sci. Budapest, E{\"o}tv{\"o}s, Sec. Math},
  14:107--108, 1971.

\bibitem{marian2004}
A.~Marian, N.~Bruno, and L.~Gravano.
\newblock Evaluating top-k queries over web-accessible databases.
\newblock {\em ACM Trans. Database Syst.}, 29(2), 2004.

\bibitem{marsaglia1972choosing}
G.~Marsaglia et~al.
\newblock Choosing a point from the surface of a sphere.
\newblock {\em The Annals of Mathematical Statistics}, 43(2), 1972.

\bibitem{nanongkai2012interactive}
D.~Nanongkai, A.~Lall, A.~Das~Sarma, and K.~Makino.
\newblock Interactive regret minimization.
\newblock In {\em SIGMOD}. ACM, 2012.

\bibitem{nanongkai2010}
D.~Nanongkai, A.~D. Sarma, A.~Lall, R.~J. Lipton, and J.~Xu.
\newblock Regret-minimizing representative databases.
\newblock {\em VLDB}, 2010.

\bibitem{PSS92}
J.~Pach, W.~Steiger, and E.~Szemer{\'e}di.
\newblock An upper bound on the number of planar k-sets.
\newblock {\em DCG}, 7(1):109--123, 1992.

\bibitem{peng2014geometry}
P.~Peng and R.~C.-W. Wong.
\newblock Geometry approach for k-regret query.
\newblock In {\em ICDE}. IEEE, 2014.

\bibitem{farhad}
M.~F. Rahman, A.~Asudeh, N.~Koudas, and G.~Das.
\newblock Efficient computation of subspace skyline over categorical domains.
\newblock In {\em CIKM}. ACM, 2017.

\bibitem{sharir2000}
M.~Sharir, S.~Smorodinsky, and G.~Tardos.
\newblock An improved bound for k-sets in three dimensions.
\newblock In {\em SOCG}. ACM, 2000.

\bibitem{Tot01}
G.~T{\'o}th.
\newblock Point sets with many k-sets.
\newblock {\em DCG}, 26(2), 2001.

\bibitem{vapnik2013nature}
V.~Vapnik.
\newblock {\em The nature of statistical learning theory}.
\newblock Springer science \& business media, 2013.

\bibitem{vlachou2010ranking}
A.~Vlachou and M.~Vazirgiannis.
\newblock Ranking the sky: Discovering the importance of skyline points through
  subspace dominance relationships.
\newblock {\em DKE}, 69(9), 2010.

\bibitem{weil1993stochastic}
W.~Weil and J.~Wieacker.
\newblock Stochastic geometry, handbook of convex geometry, vol. a, b,
  1391--1438, 1993.

\bibitem{xin}
D.~Xin, C.~Chen, and J.~Han.
\newblock Towards robust indexing for ranked queries.
\newblock In {\em VLDB}, 2006.

\bibitem{zeighami2016minimizing}
S.~Zeighami and R.~C.-W. Wong.
\newblock Minimizing average regret ratio in database.
\newblock In {\em SIGMOD}. ACM, 2016.

\end{thebibliography}
\techrep{
\appendix
\section{Proofs}\label{ap:proofs}

{\sc Theorem}~\ref{th:2drrr1}. {\it The algorithm \twodrrr is in $O(n^2\log n)$.}
\begin{proof}
The complexity of the algorithm \twodrrr depends is determined by Algorithms~\ref{alg:2dfr} and~\ref{alg:2d}.
Algorithms~\ref{alg:2dfr} first orders the items based on $x$ in $O(n\log n)$. Then in applies a ray sweeping from the x-axis toward $y$ and at every intersection applies constant number of operations. The upper bound on the number of intersections in $O(n^2)$ and therefore, it is the running time of Algorithms~\ref{alg:2dfr}.
Calling Algorithm~\ref{alg:2dfr}, generates at most $n$ ranges, each for an item.
Every iteration of Algorithm~\ref{alg:2d} is in $O(n\log n)$ as it applies a binary search on the set of uncovered intervals for each unselected item, and the number of uncovered intervals is bounded by $O(n)$. Given that the output size is bounded by $n$, Algorithm~\ref{alg:2d} is in $O(n^2\log n)$.
\end{proof}

{\sc Theorem}~\ref{th:2drrr2}. {\it The output size of \twodrrr is not more than the size of the optimal solution for RRR.}
\begin{proof}
Following the $k$-border, while sweeping a ray from $x$-axis to $y$, the top-$k$ results change only when a line above the border intersects with it.
For example, in Figure~\ref{fig:toy2}, moving from $x$-axis to $y$, in the intersection between $\mathsf{d}(t_3)$ and $\mathsf{d}(t_1)$, the top-$2$ changes from $\{t_7,t_1\}$ to $\{t_7,t_3\}$. Consider the collection of the top-$k$ results and the range of angles of rays (named as top-$k$ regions) that provide them.
Now consider the ranges that are generated by Algorithm~\ref{alg:2dfr} for each item. Let us name them here as the ranges of items. These ranges mark the first and last angle for which an item is in top-$k$.
For each top-$k$ region $R$, let the set items that their ranges cover it be $S_R$. Each top-$k$ region is covered by each and every item in its top-$k$. In addition the ranges of some other items cover each top-$k$ region. Therefore, $S_R$ is a superset for the top-$k$ of $R$.
An optimal solution with the minimum number of items from the collection of supersets that contains at least one item from each set, is not larger than the minimum number of such items from the collection of subsets. As a result, the output size of \twodrrr is not greater that the size of the optimal solution for the RRR problem.
\end{proof}

{\sc Theorem}~\ref{th:2drrr3}. {\it The output of \twodrrr guarantees the maximum rank-regret of $2k$.}
\begin{proof}
The proof is straightforward, following the Theorem~\ref{th:max2k}.
For each item $t$, Algorithm~\ref{alg:2dfr} finds a range that in its beginning and its end, $t$ is in the top-$k$.
Therefore, based on Theorem~\ref{th:max2k}, the rank of $t$ for each of the functions inside its range is no more than $2k$. Algorithm~\ref{alg:2d} covers the function space with the ranges generated by Algorithm~\ref{alg:2dfr}. Hence, for each function, there exists an item $t$ in the output where $\mathcal{r}_f(t)\leq 2k$.
\end{proof}

{\sc Lemma}~\ref{lemma:1}. {\it Let $\mathcal{S}$ be the collection of all $k$-sets for the points corresponding to the items $t\in\mathcal{D}$.
For each ranking function $f$, there exists a $k$-set $S\in\mathcal{S}$ such that top-$k$($f$)=$S$.}

\begin{proof}
The proof is straight-forward using contradiction.
Consider a ranking function $f$ with the weight vector $w$ where the top-$k$ is $S_f$ and $S_f$ does not belong to $\mathcal{S}$.
Let $t$ be the item for which $\mathcal{r}_f(t)=k$.
Consider the hyperplane $h(t,w)$.
For all the items in $t'\in S_f$ $f(t')\leq f(t)$
and for all items in $\mathcal{D}\backslash S_f$, $f(t') > f(t)$.
Hence, all the items in $S_f$ fall in the positive half space of $h$ -- i.e., $h(t,w)^+ = S_f$.
Since $|S_f|$ is $k$, $card( h(t,w)^+ ) = k$.
Therefore $h(t,w)^+ = S_f$ is a $k$-set and should belong to $\mathcal{S}$, which contradicts with the assumption that is does not belong to the collection of $k$-sets.
\end{proof}

{\sc Theorem}~\ref{th:mdrc2}. {\it \mdrc guarantees the maximum rank-regret of $dk$.}
\begin{proof}
The proof of this theorem is based on Theorem~\ref{th:max2k}. We also consider the arrangement lattice~\cite{Ed87} for this proof.
Every convex region in the $(d-1)$-dimensional space is constructed from the $d-2$ dimensional space convex facets as its borders. Each of the facets are constructed by $d-3$ dimensional facets, and this continues all the way down until the ($0$ dimensional) points. For example, the borders of a convex polyhedron in 3D, are two dimensional convex polygones; the borders of the polygones are (one dimensional) line segments, each specified by two points.
The arrangement lattice is the data structure that 
describe the convex polyhedron by its $i$ dimensional facets -- $\forall ~0\leq i\leq d$. The nodes at level $i$ of the lattice show the $i$ dimensional facets, each connected to its $i-1$ dimensional borders, as well as the $i+1$ dimensional facets those are a border for.

Now, let us consider the hyper-rectangle of each of the leaf nodes in the recursion tree of \mdrc (c.f. Figure~\ref{fig:practicalrunning}) and let $t$ be the tuple that appeared at the top-$k$ of all corners of the hyper-rectangle. Consider the arrangement lattice for the hyper-rectangle of the leaf node and let us move up from the bottom of the lattice, identifying the maximum rank of $t$ at each level of it.
Since $t$ is in the top-$k$ of both corners of each line segment in level 1, based on Theorem~\ref{th:max2k}, its rank for each point on the line is at most $2k$.
Level 2 of the lattice shows the  two dimensional rectangles, each built by the line segments at level 1.
For every point inside each rectangle at level 2, consider a line segment on the rectangle's affine space starting from one of its corners, passing through the point and ending on the edge of the rectangle.
Since the rank of the point on the corner is less than $k$ and for any point on the edge less than $2k$, based on Theorem~\ref{th:max2k}, the rank of $t$ for the points inside the rectangles at level 2 of lattice is at most $k+2k = 3k$.
Similarly, consider each hyper-rectangle at level $i$ of the lattice.
The hyper-rectangle is built by the $(i-1)$ dimensional hyper-rectangle at level $i-1$.
For every point inside the $i$ dimensional hyper-rectangle, consider the line segment starting from a corner of the hyper-rectangle, passing through the point and hitting the edge of it.
By induction, the rank of $t$ on the $(i-1)$ dimensional edges of hyper-rectangle is at most $ik$. Therefore, since the rank of $t$ on the corner is at most $k$, based on Theorem~\ref{th:max2k}, its rank for the point inside the $i$ dimensional hyper-rectangle is at most $k+ik = (i+1)k$.
Therefore, the rank of $t$ for every point inside the $(d-1)$ dimensional hyper-rectangle (the top of the lattice) is at most $k + (d-1)k = dk$.
\mdrc partitions the function space into hyper-rectangles that, for each, there exists an item $t$ in the top-$k$ in all of hyper-rectangle's corners (included in the output).
The rank of $t$ for every point inside the hyper-rectangle is at most $dk$. Since every function in the space belongs to a hyper-rectangle, there exists an item in the output that guarantees the rank of $dk$ for it.
\end{proof}

\section{k-set enumeration}\label{ap:kset-enum}
In this section we review the enumeration of $k$-sets for a dataset $\mathcal{D}$.
Especially, we explain an algorithm~\cite{andrzejak1999optimization} that transforms the $k$-set enumeration to a graph traversal problem.

In 2D, $k$-set enumeration can efficiently be done using a ray sweeping algorithm (Similar to Algorithm~\ref{alg:2dfr}) that starts from the x-axis and along the way to the y-axis monitors the changes in the top-$k$ border.
Similarly, in MD, one can consider the arrangement of hyperplanes in the dual space and, similar to 2D,
follow the $k$-border. This algorithm is in $O(|\mathcal{S}|kn^2 $LP$(d,n))$~\cite{andrzejak1999optimization} -- where LP$(d,n)$ is the required time for solving a linear programming with $d$ variables and $n$ constrains.

Alternatively, here we revisit the algorithm that is first proposed in~\cite{andrzejak1999optimization} and 
unlike the algorithms that enumerate the sets in the geometric dual space, transforms the problem to a (simple to understand and implement) {\em graph traversal} problem. 

\noindent
\begin{definition}[$k$-set Graph] \label{def:ksetgraph}
Consider a graph $G(V,E)$, where each node in the graph is a $k$-set, i.e., $\forall$ $k$-set $s_i\in \mathcal{S}$, $\exists v_i\in V$.
The edges are between the nodes that their corresponding $k$-sets share $k-1$ tuples. I.e.,
$\forall s_i,s_j \in \mathcal{S}$ where $|s_i\cap s_j| = k-1$, $\exists e_{ij}\in E$.
\end{definition}

\begin{figure}[!t]
	\centering
	\includegraphics[width=0.3\textwidth]{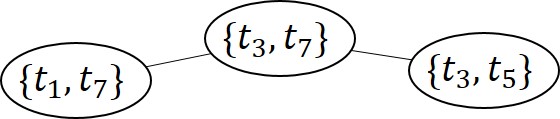}
    \vspace{-2mm}\caption{The $k$-sets for for Figure~\ref{fig:toy4}}
    \label{fig:toy5}
\end{figure}

\begin{theorem} \label{th:ksetconnected}
The $k$-set graph is connected.
\end{theorem}
\begin{proof}
Consider the dual space in which every item $t$ is transformed to the hyperplane $\sum_{i=1}^d t[i]x_i = 1$.
As explained in \S~\ref{sec:dual}, every linear function $f$ with the weight vector $w$ is the origin-staring ray that passes through the point $w$. The ordering of the items based on $f$ are the same as the ordering of the intersections of their hyperplanes with the ray of $f$.
Therefore, for every ray the $k$ closest intersection to the origin identify the top-$k$ and the $k$-th hyperplane intersecting it specifies the $k$-border (see Figure~\ref{fig:toy2}).
The $k$-border is constructed by a set of facets each belonging to a hyperplane.
The borders of the facets are intersections of pairs of hyperplanes.
The borders for which one hyperplane above the $k$-border intersects with it specify two adjacent nodes in the $k$-set graph.
Given two $k$-sets $S_i$ and $S_j$, consider two facets that have those belove the $k$-border. For any arbitrary pair of points on these two facets, the line that connects these two points specifies a sequence of facets that their corresponding $k$-sets define a path from the node of $S_i$ to $S_j$ in the $k$-set graph. Therefore, since the $k$-border is connected, the $k$-set graph is also connected.
\end{proof}

Theorem~\ref{th:ksetconnected} provides the key property for designing the $k$-set enumeration algorithm.
Since the $k$-set graph is connected, any traversal of it discovers all the $k$-sets.
Algorithm~\ref{alg:kset} shows the pseudo-code for enumerating the $k$-sets of a dataset $\mathcal{D}$.
The algorithm has an initial step in which it finds a $k$-set, followed by the traversal of the graph.

\vspace{3mm}
\begin{algorithm}[!h]
\caption{{\bf $k$-set} \texttt{\scriptsize // BFS traversal of the k-set graph}\\
		 {\bf Input:} dataset $\mathcal{D}$ \\
		 {\bf Output:} $k$-set $\mathcal{S}$ 
		}
\begin{algorithmic}[1]
\label{alg:kset}
\STATE $S = $ top-$k$ tuples in $\mathcal{D}$ on attribute $A_1$
\STATE $\mathcal{S} = \{ S\}$ 
\STATE {\it Enqueue}($S$)
\WHILE{{\it queue} is not empty}
	\STATE $S = $ {\it Dequeue}()
	\FOR{$t\in S$}
		\FOR{$t'\in \mathcal{D}\backslash S$}
        	\STATE $S' = S.$remove$(t).$add$(t')$
			\IF{$S'\notin \mathcal{S}$}
				\IF {$S'$ is a valid $k$-set} \label{line:valid}
					\STATE add $S'$ to $\mathcal{S}$
					\STATE {\it Enqueue}($S'$)
				\ENDIF
			\ENDIF
		\ENDFOR
	\ENDFOR
\ENDWHILE
\STATE {\bf return} ($\mathcal{S}$)
\end{algorithmic}
\end{algorithm}

\noindent
{\em Initial step.} The first step in the algorithm is to find a $k$-set as a starting point. To do so, the algorithm finds the top-$k$ items with the maximum values on the first attribute and considers it as the first node in the graph. 

\noindent
{\em Traversal.} Based on Theorem~\ref{th:ksetconnected} the $k$-set graph is connected. Thus, given a vertex in the graph applying a BFS (breaths first search) traversal on the graph will discover all of its nodes. 
After finding the first $k$-set, the algorithm adds it to a queue and continues the traversal until the queue is not empty.
At every iteration, the algorithm removes a $k$-set $S$ from the queue and replaces its items one after the other, with the items in $\mathcal{D}\backslash S$ to get a set $S'$.
Next, if $S'$ does not belong to $\mathcal{S}$, it checks if this is a valid $k$-set or not. It adds $S'$ to $\mathcal{S}$ and to the queue, if it is a $k$-set.

Finding out if a set $S'$ of size $k$ is a $k$-set can be done through solving a linear programming.
$S'$ is a $k$-set if there exists a hyperplane $h$ containing the point $\rho$ and the normal vector $v$ such that $S' = h(\rho ,v)^+$ and $card( h(\rho,v)^+ ) = k$.
Also, the hyperplane $h$ identifies the contour of the function $v$ for $\rho$, i.e., all the points on $h$ have the same score as $\rho$ based on $v$. 
Any point in $h(\rho ,v)^+$ should have a larger score than $\rho$ based on $v$ while all the other points in $\mathcal{D}\backslash S'$ (the points falling on the negative half space) should have a smaller score than $\rho$ based on $v$.
This formulates the following linear programing formulation:
\begin{align} \label{eq:lp}
\nonumber \mbox{find}& \,\,\,\,\,(\rho, v)\\
\nonumber \mbox{s.t.  }&  \\
\nonumber & \forall t\in S': \sum\limits_{i=1}^{d}v_it[i] \geq \sum\limits_{i=1}^{d}v_i\rho_i \\
          & \forall t\in \mathcal{D}\backslash S': \sum\limits_{i=1}^{d}v_it[i] < \sum\limits_{i=1}^{d}v_i\rho_i
\end{align}
S' is valid if there exist a $\rho$ and a $v$ that satisfy the Equation~\ref{eq:lp}.

For each node in the $k$-set graph, Algorithm~\ref{alg:kset} checks the sets generated by replacing an item in the $k$-set with an item from the rest of the dataset. There totally exists $k(n-k)$ such sets for each $k$-set.
Checking if a set is a valid $k$-set, using Equation~\ref{eq:lp}, requires solving a linear programming with $d$ variables and $n$ constraints. As a result, the time complexity of Algorithm~\ref{alg:kset} is $O(nk$LP$(d,n))$.

}
\end{document}